\documentclass[amsmath,amssymb,twocolumn,aps,prl,superscriptaddress,nofootinbib]{revtex4-2}

\usepackage[utf8]{inputenc}
\usepackage{algorithm2e}
\usepackage{graphicx}
\usepackage{amssymb}
\usepackage{enumitem}
\usepackage{amsmath}
\usepackage{amsthm}
\usepackage{qcircuit}
\usepackage{braket}
\usepackage{booktabs}
\usepackage{tabularx}
\usepackage{xcolor}
\usepackage{tikz}
\usetikzlibrary{fit,backgrounds,positioning,arrows.meta}
\usepackage{tikz-cd}
\usepackage{hyperref}

\newcommand{\rank}{\mathrm{rank}}
\newcommand{\rs}[1]{\langle #1 \rangle}
\newcommand{\ns}[1]{\langle #1 \rangle^\perp}
\newcommand{\apmtx}[1]{\mathrm{APM}(#1)}
\newcommand{\ord}[1]{\mathrm{ord}(#1)}
\newcommand{\pr}{\mathrm{Proj}}
\newcommand{\wt}[1]{\mathrm{wt}(#1)}
\newcommand{\supp}[1]{\mathrm{supp}(#1)}

\newcommand{\gth}[1]{\mathrm{girth}(#1)}
\newcommand{\CSS}[1]{\mathrm{CSS}(#1)}
\newcommand{\Aut}{\mathrm{Aut}}
\newcommand{\Z}{\mathbb{Z}}
\newcommand{\F}{\mathbb{F}}

\newcommand{\lead}[1]{\textit{#1.---}}

\newtheorem{theorem}{Theorem}
\newtheorem{corollary}[theorem]{Corollary}
\newtheorem{lemma}[theorem]{Lemma}
\newtheorem{proposition}[theorem]{Proposition}
\theoremstyle{definition}
\newtheorem{definition}{Definition}
\newtheorem{remark}{Remark}
\newtheorem{fact}[remark]{Fact}

\newcommand{\edit}[1]{{#1}}

\begin{document}

\title{Designer Codes from GALA:\\Compact, Self-Dual, and Rate-1/2 QEC on Reconfigurable Atom Arrays}

\author{Willers Yang}
\affiliation{
  QuEra Computing Inc., 1380 Soldiers Field Road, Boston, MA 02135, USA
}
\affiliation{
  Department of Computer Science, University of Chicago, Chicago, IL 60615, USA}
  \author{Casey Duckering}
\affiliation{
  QuEra Computing Inc., 1380 Soldiers Field Road, Boston, MA 02135, USA
}
  \author{Arpit Dua}
\affiliation{
  QuEra Computing Inc., 1380 Soldiers Field Road, Boston, MA 02135, USA
}
  \affiliation{Department of Physics, Virginia Tech, Blacksburg, VA 24061, USA}%

\begin{abstract}
\edit{Achieving useful fault-tolerant applications requires co-designing the quantum error-correcting code for footprint, logical instruction set, and hardware compatibility simultaneously. Recent breakthroughs in ultra-high-rate codes have the potential to reduce footprint of large scale quantum applications but their practical utility is limited without an addressable logical instruction set and efficient implementation on near term hardware. Reconfigurable neutral atom arrays, for example, support only hundreds of qubits with constrained reconfigurability from a few crossed acousto-optic deflectors (AOD).

We introduce the GALA codes, or Group-Action Lifts with Active orthogonality, that allows such co-design by construction. GALA lifts over a product group $G = H_k \times C_m$ (or $H_k \ltimes C_m^k$); the small non-abelian factor $H_k$ supplies active orthogonality that improves qubit overhead, reaching 1/2 encoding rate with above-weight distance, while the large abelian factor $C_m$ supplies symmetries that give code automorphisms and explicit, simple AOD move schedules. Logical addressibility and hardware compatibility follow directly from construction, rather than verified post-hoc. The GALA family contains several known rate-1/2 Kasai codes of Ref. \cite{kasai2026breakingorthogonalitybarrierquantum,zhao2026ultrahighratequantumerrorcorrection} while exposing simpler parameter bounds, logical operations, and ZX-dual variants with AOD-compatible fold-transversal Clifford gates. In particular, we find a compact self-dual $[[132,30,12]]$ with $1.0^{+0.9}_{-0.6}\times 10^{-9}$ per-round per-qubit logical error rate at $10^{-3}$ physical error rate, equipped with 3.1 ms syndrome-extraction cycles and an addressable logical instruction set; a ZX-dual, girth-6, rate-1/2 $[[672,336,12]]$ with a 6.76 ms cycles; and barrier-breaking $[[1752,880,14]]$ and $[[2232,1120,16]]$ smaller than previously known hardware-compatible rate-1/2 codes.}
\end{abstract}

\maketitle

Quantum low-density parity-check (qLDPC) codes can encode logical information at constant rate with bounded-weight checks, and are the leading route to low-overhead fault-tolerant quantum computation~\cite{Bravyi_2024,xu2024constant, yoder2025tourgrossmodularquantum,ismail2026fastparallelhighratestar,cain2026shorsalgorithmpossible10000}. Reconfigurable neutral-atom arrays (RNAA) have emerged as a natural hardware for them~\cite{bluvstein2024logical}: physical qubits are implemented as atoms held in optical tweezers, and long-range interactions are implemented using acousto-optic deflectors (AODs) that brings distant qubits into Rydberg range for parallel entangling operations~\cite{evered2023high}. In this setting the feasibility of implementing a code and its logical primitives rests less on gate depth, count, or locality, and more than on how regularly its stabilizer checks and \edit{logical instruction set} compile onto AOD-compatible physical reconfiguration. Moreover, present-day devices can hold only a few hundred to a thousand atoms~\cite{bluvstein2024logical} in the entangling zone, so a feasible code on such devices must also be compact. The relevant question regarding a code's near-term practicality is therefore not only whether it has good parameters, but also whether there exist space- and time-efficient implementations of memory and logical primitives on current qubit-constrained hardware.
\begin{figure}[ht!]
    \centering
    \includegraphics[width=\linewidth]{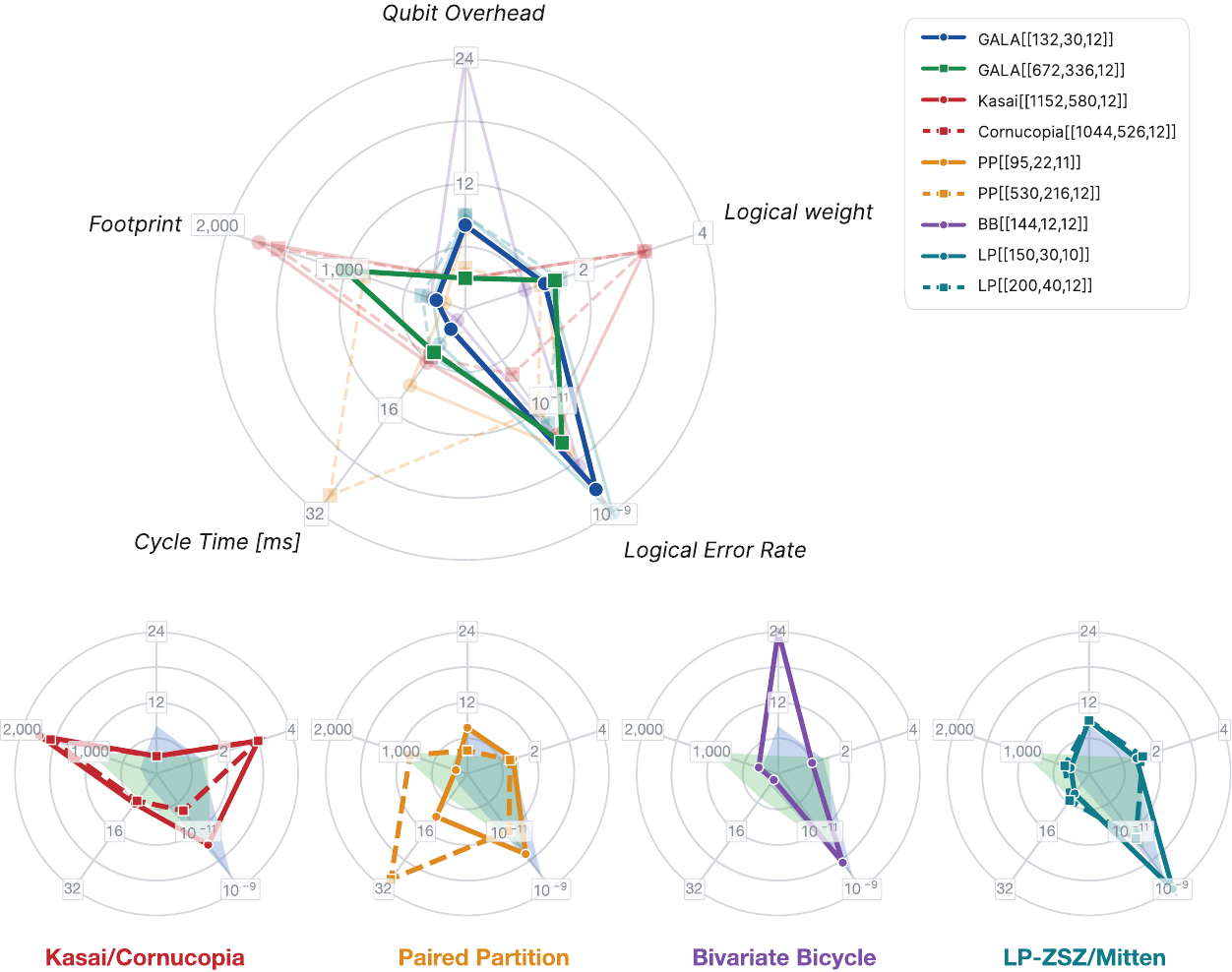}
    \caption{\emph{Top}: GALA codes simultaneously achieve low qubit overhead ($(n+n_{anc})/k$, top axis); \edit{low-weight symplectic logical operators ($\frac {1}{k}\sum_{i\in[k]} (\wt{X_i}+\wt{Z_i})/(2d) $, top right axis); gigaquop logical error rate under $p_{phys} = 10^{-3}$ (bottom right axis); short syndrome-extraction cycles on reconfigurable atom arrays ($[ms]$, bottom left axis), and compact footprint ($n +n_{anc} $, top left axis). Towards the center is better. \emph{Bottom}: Comparison to other high rate codes. Two GALA instances are shaded in blue (GALA$[[132,30,12]]$) and green (GALA$[[672,336,12]]$). GALA significantly outperforms on at least one axis across baselines, except from ZSZ-LP/Mitten codes, for which GALA achieves a similar profile. }} 
    \label{fig:hero}
\end{figure}
\begin{figure*}
    \centering
    \includegraphics[width=0.99\linewidth]{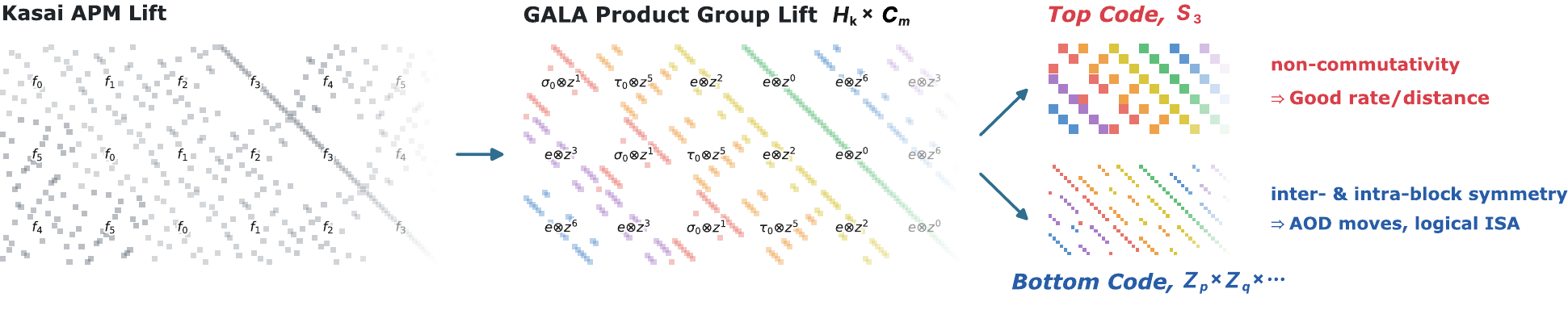}
    \caption{GALA: a family of \edit{circulant} two-block codes lifted over a product group $G=H_k\times C_m$: the small non-abelian $H_k$ supplies \emph{active orthogonality}, achieving high rate and above-check-weight distance, while the large abelian $C_m$ supplies symmetries that give both code automorphisms and parallel AOD row/column moves explicitly. }
    \label{fig:idea}
\end{figure*}

\edit{While strong qLDPC codes have continuously advanced the frontier of memory and computation efficiency in recent years, it remains rare to find codes that simultaneously achieves high encoding rate, addressable and fast logical processing, and hardware compatible implementations. Firstly, code instances achieving hardware compatibility often do so by constraining the search space to desirable properties such as planarity conditions of the Tanner graph \cite{Bravyi_2024, nixon2026vinecodeslowoverheadquantum,gu2026nearestneighbourgatesneedhighrate}, shuttling regularity \cite{Hong_2024} or AOD-compatibility \cite{Viszlai_2026} of its syndrome extraction, etc. While one can still achieve good rate and expressive native ISAs, in practice, best codes discovered in such constraint search spaces achieve modest encoding rates of}  $r\lesssim 10\%$.  A high-rate code, on the other hand, can lead to order-of-magnitude reduction in space-overhead, but usually pays the cost in longer logical \cite{cain2026shorsalgorithmpossible10000} \edit{processing speed due to longer QEC times and struggles against addressable and parallelizable logical operations. Indeed, logical operations in qLDPC codes are either supplied by surgery gadgets that requires one to disentangle logical degrees of freedoms stored in overlapping sets of physical qubits while maintaining hardware compatibility or code automorphism that relies on the inherent symmetries of the code, both of which becomes more challenging as more logical qubits are compressed in a single block. On the other end, of this trade-off, one can find qLDPC codes capable of implementing a Clifford ISA as $O(1)$ overhead code automorphisms, such as subsystem constructions \cite{malcolm2025computingefficientlyqldpccodes, Nakai_2026}, homological product codes of classical codes with large automorphism groups \cite{berthusen2025automorphismgadgetshomologicalproduct,yang2026spacetimeefficienthardwarecompatiblecomplexquantum}, or phantom codes \cite{koh2026entanglinglogicalqubitsphysical,koh2026achievinglimitsautomorphismgates}, while also incurring overhead in either stabilizer weight or size exponential in $k$. We can bypass these obstructions on compute efficiency if we relax the requirement for the automorphism ISA to be universal \cite{yang2026spacetimeefficienthardwarecompatiblecomplexquantum, koh2026achievinglimitsautomorphismgates}. For example, limited  symmetries of the lift group, while not universal, can enhance a code's suitability as hierarchical quantum memory \cite{Bravyi_2024} or enable efficient extractor architectures with reduced space overhead \cite{hong2026quantumldpccodesdesign, bhardwaj2026highrateqldpcprocessors} even in high-rate qLDPC codes. } 

\edit{A particularly }promising breakthroughs from Kasai that showed greater space-overhead reduction is possible, with discoveries of girth-$8$ and check weight-$12$ codes achieving ultra-high rate (UHR) exceeding $50\%$ \cite{kasai2026breakingorthogonalitybarrierquantum}. Lifting with non-commuting affine permutation matrices (APMs), Kasai's construction achieves rate $\frac 12$ by selectively enforcing the Calderbank-Shor-Steane (CSS) orthogonality condition only on an \emph{active} subset of block rows. The remaining \emph{latent} rows, which are no longer valid stabilizers, are then removed to introduce new logical degrees of freedom, whose distance is not necessarily limited by the latent rows' bounded weights. \edit{There have been two important improvements since the original construction. First, the code found in the original paper have block size $n =9216$, which} is too large to fit in any near-term device (at least not entirely in the entangling zone, which is necessary for low-depth). \edit{A new pair-partition (PP) construction was later proposed that contain compact instances $[[220, 50, 16]]$ and non-CSS $[[90, 21, 11]]$ at weight-10, and rate-$1/2$ $[[264, 134, 8]]$, $[[1356, 682, 14]]$ at weight $12$.~\cite{okada2026pairpartitionconstructionscpmbasedquantum, lee2026designprinciplesultrahighratequantum}. PP codes replace the circulant proto-matrix of the original construction with a more general pairing graph, reaching a richer search space. Furthermore, low-weight logical operators were found that inherits a cyclic logical permutation from the lift. 
} Second, it is not immediately clear how to implement syndrome extraction or logical operations quickly. \edit{ This question was partially addressed by a co-design of the original construction that yields a smaller RNAA-compatible $[[1152,580,\le12]]$ code with an $8.3$ ms error-correction (QEC) cycle, assuming four crossed AOD pairs. The fast cycle times is achieved by exploiting the circulant proto-matrix to perform syndrome extraction in parallel, and by requiring the transition affine permutations to commute with a reference APM, which ensures that the transition APM have an orbit decomposition that be implemented efficiently with AOD reconfiguration. Similarly to other hardware compatible constructions, these restrictions does make it more challenging to each better parameters, as evident by comparison of RNAA-compatible Kasai instances such as $[[1152,580,\leq12]]$ and $[[2304, 1156, \leq 14]]$ with general PP codes achieving $[[1356, 682, 14]]$. }

Several important gaps remain in making UHR codes practical for RNAA.  On construction front, it is challenging to study the algebraic structures of these codes with \edit{established} tools. Since its protograph is a circulant matrix restricted on a set of active rows, it has neither the product structure of codes such as hypergraph product~\cite{tillich2014quantum}, lifted product~\cite{panteleev2021degenerate}, and balanced product codes~\cite{breuckmann2021balanced}, whose logical observables and (fold)-transversal gates can be inherited from the structures of the classical codes~\cite{Quintavalle_2023,berthusen2025automorphismgadgetshomologicalproduct}, nor the group-algebra formalization~\cite{panteleev2021degenerate} due to the removed latent rows. As such, the availability of symmetries that lead to efficient logical operations and AOD compatibility is not guaranteed, but is instead a criterion for rejection sampling during code search. It also remains open whether one can obtain concrete instances equipped with an \emph{addressable} logical instruction set, as opposed to global permutations acting on every encoded block at once, and with a complete set of Clifford generators.

In this work, we close some of these gaps by viewing existing instances of UHR codes \edit{with a circulant protomatrix} as special cases of an algebraically transparent and customizable construction, which we call \emph{GALA}---\textit{G}roup-\textit{A}ction \textit{L}ifts with \textit{A}ctive-orthogonality. GALA is a family of two-block CSS codes lifted over a product group $G=H_k\times C_m$ (a direct product, DPG) or $G=H_k\ltimes C_m^{\,k}$ (a semi-direct product, SPG). The factors carry two distinct roles. The small non-abelian factor $H_k$ replaces APM in Kasai's construction, which supplies the controlled non-commutativity for active orthogonality and leads to UHR codes with above-check-weight distances. Unlike large APMs, however, we can choose a small $H_k$ (e.g.\ $S_3$) where the commutativity pattern can be enumerated exhaustively. The large abelian factor $C_m$ supplies the code's symmetry and makes the algebra of the code explicit, which serves both as (i) a tool to organize the code's logical operators and (ii) explicit AOD move schedule for syndrome extraction and logical operations. See Fig.~\ref{fig:idea} for a visual comparison against the original APM formulation. As a consequence, GALA codes guarantee AOD compatibility by design and admit provably fault-tolerant logical operations via code automorphisms and chain-complex homomorphisms made apparent by the group-product formulation. Hardware compatibility and the logical instruction set architecture (ISA) hence become inputs to the code search, guaranteed by choices of $H_k$, $C_m$, and the group product, instead of emergent properties to be post-selected for. 

This factorization also lets us study the properties and limitations of \edit{circulant} high rate codes; derivations for the statements below are given in the Supplemental Material, and the distance and girth of every concrete instance are certified numerically. It allows us to derive closed-form rate, distance, and girth bounds from the group data; we can construct low-weight logical bases and transversal CNOTs from the quotient groups of the lift; the syndrome-extraction schedule is prescribed exactly by the list of group elements chosen; and by imposing a fold symmetry on the group elements, we can obtain codes with ZX-dualities that admit fold-transversal Clifford gates. Crucially, a preliminary search over more than $10^5$ codes returns instances small enough for present-day hardware, while retaining rate $\ge1/2$, girth $\ge6$, and QEC cycles of a few milliseconds; we report the full set below in Fig.~\ref{fig:results} and in Table~\ref{tab:frontier-nondual}-\ref{tab:frontier-other}

\lead{Construction} Fix integers $L$, $J\le L/2$ and an \emph{active set} $\Gamma\subseteq[\tfrac L2]\times[\tfrac L2]$ containing the diagonal band $\Gamma_J=\{(i,j):i+j\equiv r \ (\mathrm{mod}\ \tfrac L2),\ r<J\}$. Choose lifts $\mathcal F=\{(F_i,f_i)\}$, $\mathcal G=\{(G_i,g_i)\}\subset H_k\times C_m$ with $F_i,G_i\in H_k$, $f_i,g_i\in C_m$, and $[F_i,G_j]=0\iff(i,j)\in\Gamma$. Writing $z_i$ for the cyclic shift on $\Z_{L/2}$, form the block-circulant parents $\hat H_X=[F\,|\,G]$ and $\hat H_Z=[G^{T}\,|\,F^{T}]$ from
\begin{equation}\label{eq:lift}
F=\sum_{i} z_i\otimes F_i\otimes f_i,\qquad
G=\sum_{i} z_i\otimes G_i\otimes g_i,
\end{equation}
and keep the first $J$ block rows as the stabilizers $H_X,H_Z$. More generally the lift may be taken in the group ring $\F_2[G]$, with each entry $F_i$ a \emph{sum} of group elements rather than a single one; the stabilizer weight is then the total number of terms, $w=\sum_i(|F_i|+|G_i|)$, reducing to $w=L$ in the \emph{monomial} case of Eq.~\eqref{eq:lift}. Both regimes matter below: the monomial family is the one that contains the co-designed Kasai codes and for which the sharpest bounds hold, while the polynomial family decouples $w$ from $L$ and yields our most compact instances. A direct computation gives $[\hat H_X\hat H_Z^{T}]_{ij}=\Psi_{j-i}$ with the anticommutator sum $\Psi_r=\sum_{u}[F_u,G_{r-u}]$. The active constraint forces $\Psi_r=0$ for $r<J$, so the $J$ active rows commute; since $\Gamma\subsetneq[\tfrac L2]^2$, $\Psi_J\neq0$ and the $\tfrac L2-J$ latent rows do not---these are the degrees of freedom that can become above-weight logicals. The commutativity condition on the product group is itself a simple algebraic test: for a direct product $[(F_i, f_i),(G_j, g_i)]=[F_i,G_j]\otimes f_i g_j$ vanishes iff the small non-abelian factors commute, so the entire orthogonality pattern is decided inside $H_k$. Taking $H_k=S_3 = \braket{\sigma, \tau | \sigma^2 = \tau^3 = e, \sigma \tau\sigma  = \tau^{-1}}$ with the ansatz $\mathcal F^{H}=\{\sigma^j,\tau^i,e,e,e,e\}$, $\mathcal G^{H}=\{e,e,\sigma^k,\tau^i,e,e\}$ and $C_m=\Z_{32}$ reproduces the co-designed $[[1152,580,\le12]]$ code of Ref.~\cite{zhao2026ultrahighratequantumerrorcorrection} exactly. More generally, requiring the transition permutations to commute with a reference affine permutation forces the lift to factor through a (semi-)direct product, so GALA contains the AOD-compatible Kasai codes of \cite{zhao2026ultrahighratequantumerrorcorrection} as a subclass while spanning a richer, more controllable search space; the semi-direct (SPG) version simply replaces the tensor product by the wreath product, which permutes the $k$ blocks by $H_k$ and $C_m$ acts within each block. We give more details on the construction Section~\ref{sec:construction} of SM.
\begin{figure*}
    \centering
    \includegraphics[width=0.99\linewidth]{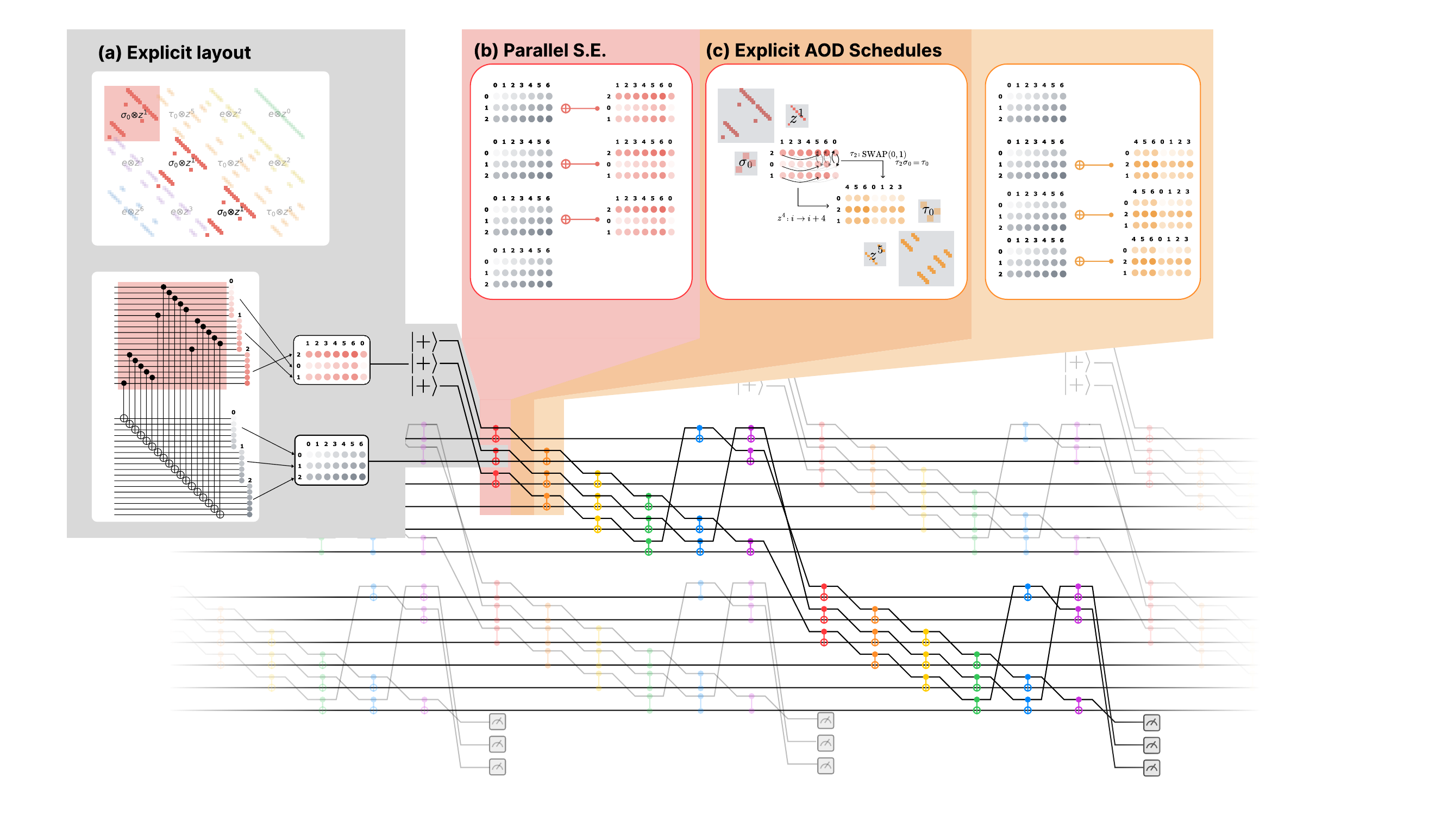}
    \caption{\edit{Syndrome extraction circuit of a $\text{GALA}_{12,3}(S_3\times \mathbb Z_7)$ code with generators $F_i = (h_i, c_i)$. The shaded circuit elements represent interleaved $Z$ checks. \emph{(a).} Layout from the group data. \emph{Top:} an example stabilizer matrix $H_X$.  Data qubits and check qubits each occupy $L$ ($J$) arrays of size $kp\times q$ respectively, writing $C_m=C_p\times C_q$. In this example, $k=3$, $C_m = \mathbb Z_7$ so $p = 1, q= 7$; hence data and ancillae groups into $12$+ $6$ arrays. \emph{Bottom:} Each wire represents a bundle of $3\times 7$ physical qubits, and each $k$-qubit gate correspond to a $21k$-qubit operation. Note control pattern of a CNOT is exactly the PCM (when in natural ordering). \emph{(b).}  Round $i$ pairs each check block $j$ with the data blocks $j +i\mod \frac L2$, performs the checks needed on the $(i+j)^{th}$ block-diagonal of $H_X$ in one step. \emph{(c).} Qubit layout along the product structure turns the transition permutation $h\times c\in H_k\times C_m$ into two independent shifts, consisting of a small column swap $h = h_ih^{-1}_{i-1}$ and a cyclic row swap $c = c_ic^{-1}_{i-1}$, which can be read off of the list of generators.}
    }
    \label{fig:move}
\end{figure*}

\lead{Bounds from the group data} Because each generator factors into two groups, we can infer bounds on the code parameters from the codes assembled from the quotient groups that are easier to analyze. Each claim in this section is backed by a proof given in Section~\ref{sec:bounds} of SM. Every rate is at least $1-2J/L$, so $J=L/4$ yields rate $\ge1/2$, and for \emph{monomial} lifts the count is sharper,
\begin{equation}\label{eq:rate}
\frac kn\ \ge\ 1-\frac{2J}{L}+\frac{2(J-1)}{n},
\end{equation}
the correction term counting the $J-1$ all-ones relations among the active rows of each sector; these relations die mod $2$ for a polynomial lift, and our $L=8$ instances sit at exactly $1/2$. For monomial lifts the distance is likewise capped by the stabilizer weight, $d\le w=L$, whenever $J>L/4$ (full orthogonality returns), and at $L<12$, so the barrier can only be broken in the window $L\ge12$ and $2<J\le L/4$. To keep qubit footprint as small as possible, this correspond exactly to the $L=12$ and $J = \frac L4 = 3$ case considered previously by \cite{kasai2026breakingorthogonalitybarrierquantum, zhao2026ultrahighratequantumerrorcorrection}.  Intuitively, once too many rows are kept active the parent check matrix commutes with itself and weight-$L$ latent rows become logicals, while at $J\le2$ every data qubit meets only two checks and the shortest Tanner-graph cycle is already a logical operator; active orthogonality survives only between these two failure modes. Polynomial lifts evade the second mechanism, which is how our $L=8$, $w=12$ instances reach $d=10$ and $d=12$ at much more compact lengths.

The product structure adds a second, complementary family of bounds. A logical operator of a quotient code, the code obtained by collapsing one group factor to a point, can be inflated back to a logical of the full code by tensoring with the all-ones vector on the collapsed coordinate, multiplying its weight by the size of that factor. Consequently
\begin{equation}\label{eq:quotient}
d\ \le\ \min\!\big[\,m\, d_{\mathcal Q^{H}},\ k\, d_{\mathcal Q_{C}}\,\big],
\end{equation}
where $\mathcal Q^{H}$ and $\mathcal Q_C$ are the top ($H_k$) and bottom ($C_m$) quotient codes; see Fig.~\ref{fig:idea}. This also yields explicit weight-$3L$ and weight-$2m$ logical representatives for the $S_3\times C_m$ family, to which the flagship co-designed Kasai instances belong \cite{zhao2026ultrahighratequantumerrorcorrection}. A covering-graph argument shows girth is inherited from the bottom quotient, and a greedy choice of bottom shifts reaches girth $\ge6$ once $m>2(J-1)L$. Thus $n,k,d$, girth are all bounded directly by $(L,J,k,m)$, turning a hyperparameter sweep into a targeted search, yielding compact hardware-compatible codes by keeping $L,k,m$ as small as possible. 
\begin{figure*}[t]
\centering
\includegraphics[width=0.48\linewidth]{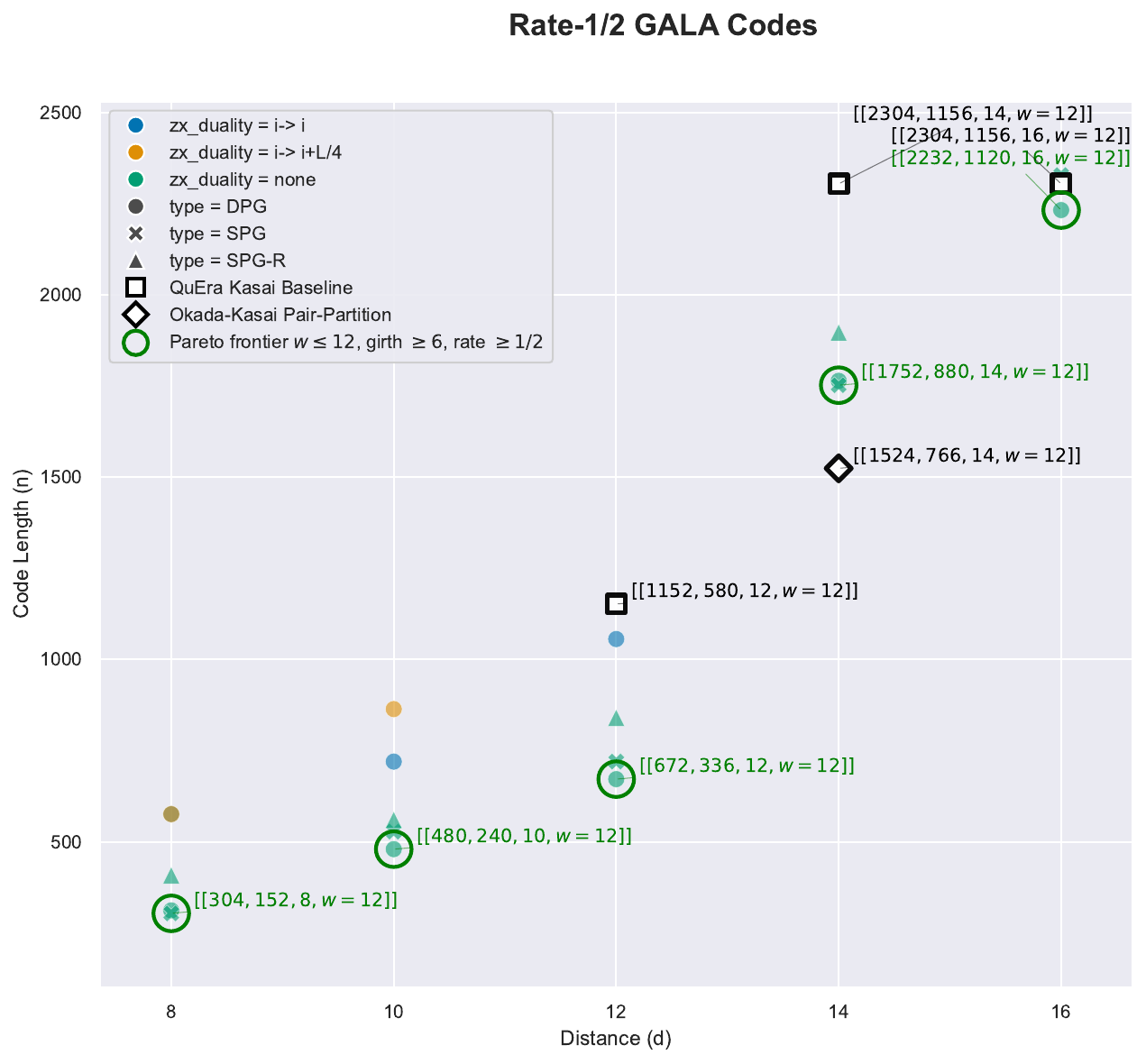}\hfill
\includegraphics[width=0.48\linewidth]{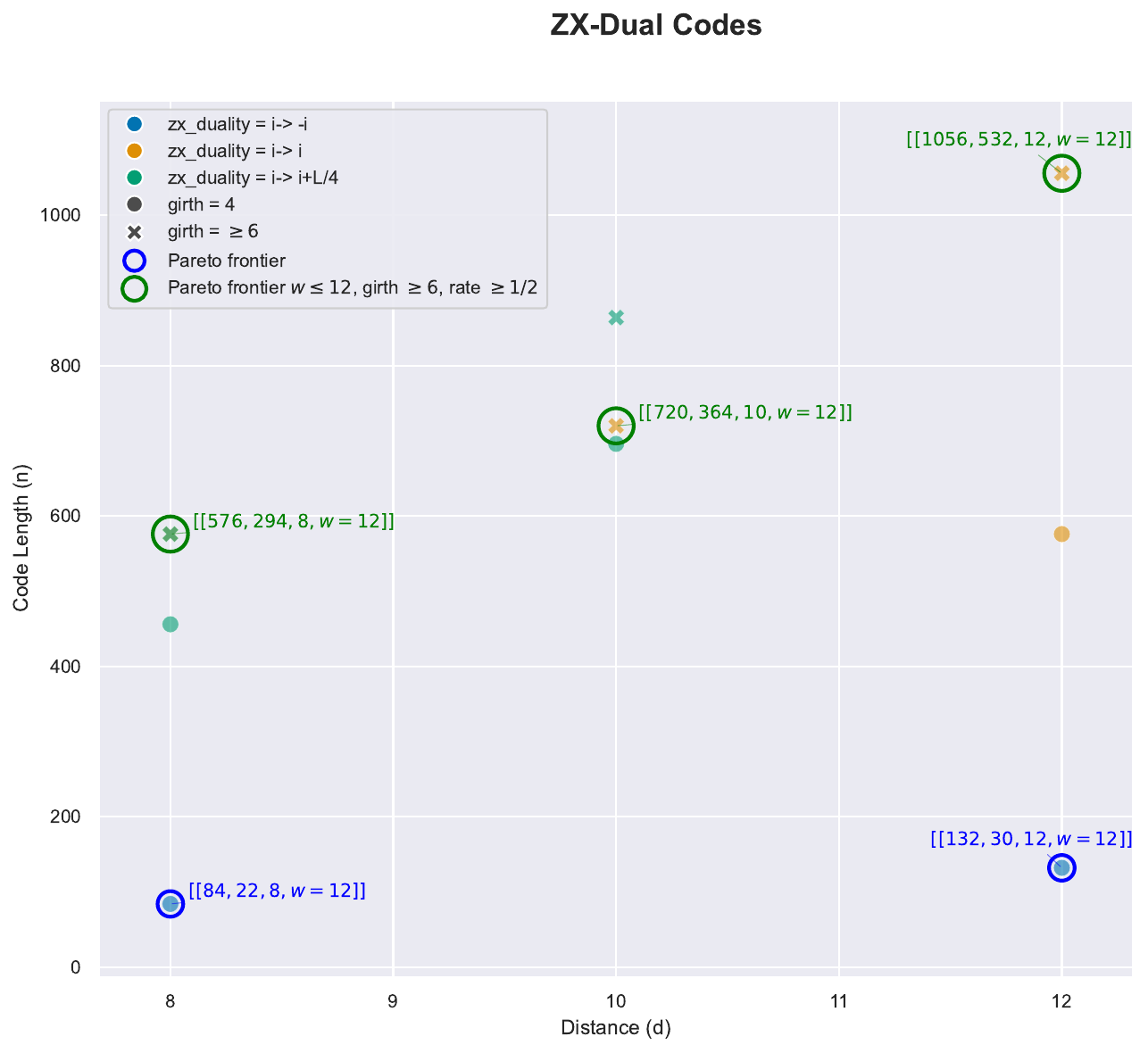}
\caption{GALA codes (lower-right is better). \textit{Left:} compact codes with $n\le2500$, rate $\ge1/2$, weight $\le12$, girth $\ge6$. Best instances are highlighted in green, including the compact $[[480,240,10]]$ and $[[672,336,12]]$ and the above-weight $[[1752,880,14]]$ and $[[2232,1120,16]]$, all shorter than the co-designed Kasai baselines (black)~\cite{zhao2026ultrahighratequantumerrorcorrection}.  Rate-$1/2$ pair-partition codes of Ref.~\cite{okada2026pairpartitionconstructionscpmbasedquantum} are shown for comparison. \textit{Right:} ZX-dual codes with fold-transversal gates, notably a girth-4 $[[132,30,12]]$ at rate $0.227$ and a girth-6, rate-$1/2$ $[[1056,532,12]]$. Every distance shown for a GALA instance is certified exactly, by complete exclusion of all lower-weight logical operators together with an explicit verified witness of weight $d$.}
\label{fig:results}
\end{figure*}

\begin{table}
    \centering
    \begin{tabular}{c|ccl}
         &  [[132,30,12]]& [[672,336,12]]& [[1152,580, 12]]\\\hline
         4 Crossed AODs&  $3.100\,\, [ms]$& $6.755\,\,[ms]$ & $ 8.317\,\,[ms]$ \\
         2 Crossed AODs&  $5.257\,\, [ms]$& $ 8.868\,\,[ms]$ & $13.267\,\, [ms]$ \\
    \end{tabular}
    \caption{QEC Cycle time compared to [[1152,580, 12]]\cite{zhao2026ultrahighratequantumerrorcorrection}. }
    \label{tab:movement}
\end{table}

\lead{Moves and automorphisms} The abelian factor describes the symmetries of the code. Arrange the data atoms on a $kp\times Lq$ grid (writing $C_m=C_p\times C_q$) and the check atoms on a $kp\times Jq$ grid. Syndrome extraction then runs in $L$ rounds, and round $i$ is a permutation element from the lift $H_k\star (C_p\times C_q)$ applied to the check atoms, which for a direct-product code factors into an independent row swap and column swap,
\begin{equation}\label{eq:moves}
        \sigma_{r,i} = h_{i}h_{i-1}^T\otimes c_ic_{i-1}^T,\quad
        \sigma_{c,i} = c'_i{c_{i-1}'}^T.
\end{equation}
i.e. rigid cyclic translations that AODs perform in parallel, as illustrated in Fig.~\ref{fig:move}, leading to ms-scale QEC cycle times as reported in Table~\ref{tab:movement}.  Cyclic permutations between the ancilla and $L$ data blocks occur in parallel with these permutations and the arrangement and movement of atoms is described in Section \ref{sec:hardware}. Taking $S_3\times\Z_{32}$ and $S_3\times\Z_2\times\Z_{32}$ reproduces the layouts and move structure reported in Ref.~\cite{zhao2026ultrahighratequantumerrorcorrection} for their $P=96$ and $P=192$ instances. 

The abelian factors are code automorphisms:
\begin{equation}\label{eq:aut}
I_{kL}\times C_m\ \le\ \Aut(\mathcal Q).
\end{equation} Here $C_m$ denotes the diagonal copy of the abelian factor, which is central in $H_k\ltimes C_m^{\,k}$. The $C_m$ translations therefore act as logical permutations, and sweeping a representative logical over its $C_m$-orbit organizes the logical space into disjoint hypercubes, as described in \ref{sec:logic}.

\lead{ZX-dual codes} Our design also has the flexibility to enforce fold symmetries between the blocks that lead to ZX-dual code, in which $H_X$ and $H_Z$ are permutation-equivalent, $H_Z= H_X\cdot\tau$. \edit{
There are two sources of such fold symmetries. Given an involution $r\in S_{\frac L2}$ that preserves the $J\times L$ proto-matrix row-wise and $t,t'\in H_k\ltimes C^k_m$, we can add enforce symmetries between the generator pairs such as 1). $t'F_i$ to $F^T_{r(i)}\cdot t$ for 2). $t'F_i = G_{r(i)}^T\cdot t$. Let $H_X = [F_L|F_R|G_L|G_R]$ and $H_Z =[G_R^T|G_L^T|F_R^T|F_L^T]$; equations of the first type establishes a permutation equivalence between $F_\star, F_\star^T$ (and similarly  $G_\star, G_\star^T$); while the equations of the second type establishes a permutation equivalence between $F_\star, G_\star^T$ (and similarly $G_\star, F_\star^T$). In particular, when a GALA code has a trivial top we will always have a ZX duality of the first type, which we call meta duality. Indeed, if all generators are cyclic, there exists the permutation $R_n$ -- the reversal on $n$ -- such we always have $rz^i r = z^{-i}$. Applying this reversal gives us $F_A\to F_A^T$. While this also caps the achievable distance to $d\leq w$, we show that the ceiling is attainable, and yields fully addressable and hardware compatible logic instruction sets on small code instances.

}

More details are given in Section~\ref{ssec:zxdual}. $\tau$ decomposes into a few rigid, AOD-compatible motion of the $kp\times Lq$ array. The cost of requiring additional symmetries is quantified in Section~\ref{ssec:zx-bounds}. In practice, we obtain girth-6 ZX-dual instances such as $[[672, 336, 12]]$ and $[[1056,532,12]]$ that are shorter than the baseline $[[1152,580,\le12]]$ (See Figure~\ref{fig:codes} and Table~\ref{tab:frontier-selfdual-good}-\ref{tab:frontier-other}). It follows from \cite{Breuckmann_2024} that the ZX-duality implements a fold-transversal Clifford gate with a single layer of physical Clifford gates and one parallel atom rearrangement---the same primitives and cost as a QEC round.  The logical action of these gates depend on the logical basis chosen, which we show explicit examples of in an end-to-end case study for the self-dual $[[132,30,12]]$ code. In particular,  For $[[132,30,12]]$ the fold is the identity, so both gates consist of a single transversal single-qubit layer with no rearrangement; at $132$ data atoms and a $3.1\,$ms QEC cycle it is realizable on present-day devices.

\lead{Code results} We searched more than $10^5$ girth-$\ge6$ codes with $n\le2500$ across the DPG, SPG-restricted, SPG-full, and polynomial families, estimating distance with \texttt{QDistRnd}~\cite{Pryadko_2022} in a cascade of $8$, $10^3$, $10^4$, and $10^5$ shots, with $2\times10^5$ further shots on the Pareto frontier. Every instance reported below is then certified \emph{exactly}: we exhaustively exclude all logical operators of weight below $d$ and verify an explicit witness of weight $d$, so no GALA parameter we quote carries a ``$\le$''. The results are plotted in Fig.~\ref{fig:results}. Beyond rediscovering the known Kasai parameters, we find, all at stabilizer weight $12$ and girth $\ge6$: (i) compact rate-$1/2$ instances $[[480,240,10]]$, $[[672,336,12]]$, and $[[720,360,12]]$, at $58\%$ the length of the co-designed $[[1152,580,\le12]]$~\cite{zhao2026ultrahighratequantumerrorcorrection}. $[[672,336,12]]$ in particular is ZX-dual and has a trivial top, showing that the abelian sector of the framework combined with a polynomial lift already suffices for compact rate-$1/2$ codes at the weight ceiling; (ii) above-weight instances $[[1752,880,14]]$ and $[[2232,1120,16]]$, both shorter than or of higher certified distance than the co-designed $[[2304,1156,\le14]]$; and (iii) the ZX-dual $[[1056,532,12]]$ with fold-transversal $H/S$, alongside a $[[132,30,12]]$ that trades rate ($0.227$) and a small number of 4-cycles for a tiny $132$-atom footprint and an identity fold. Within this weight and rate class these are, to our knowledge, the shortest instances known at $d=10$, $12$, and $16$ that reach rate 1/2; at $d=14$ the pair-partition code $[[1524,766,14]]$ of Ref.~\cite{okada2026pairpartitionconstructionscpmbasedquantum} is shorter than ours. The compact instances sit within reach of present-day atom-array processors.

\edit{

\lead{Logical Error Rates} 
We use a three stage hierarchical decoder from prior works \cite{zhao2026ultrahighratequantumerrorcorrection,lu2026cornucopia}: with two relay-BP stages from \cite{lu2026cornucopia}, and a third stage MILP decoder used in \cite{zhao2026ultrahighratequantumerrorcorrection}.

We conducted two numerical experiments. First, we performed circuit level logical error rate simulations using the full three stage decoder at $p_{phys} = 10^{-3}$ for our two smallest $d = 12$ codes We estimate the per-logical, per-round LER $\bar p_L = 1.0_{-0.53}^{+0.90}\times 10^{-9}$ for the $[[132,30,12]]$ code and $\bar p_L  = 4.3_{-0.80}^{+1.1}\times 10^{-9}$ for $[[672,336,12]]$. Following \cite{zhao2026ultrahighratequantumerrorcorrection}, we probe this regime by direct simulation rather than by extrapolating from higher physical error rates. Table~\ref{tab:ler-t3} also reports the same shots decoded with the first two stages only, and with the third stage alone, where the latter is obtained by re-decoding every failure of the first two stages with the MILP decoder. The third stage improves the $[[132,30,12]]$ code by roughly an order of magnitude but leaves $[[672,336,12]]$ essentially unchanged. This is due to the fact that the dominant source of error is from false negatives in two stages -- which the MILP decoder $T_3$ never sees. Following \cite{zhao2026ultrahighratequantumerrorcorrection}, we also obtain a $T_3$ only LER by re-decoding the false negatives, and find that the MILP decoder corrects every escalated shot for both codes, and all but one of the false negatives from the $[[672,336,12]]$ code. Under the $T_3$-only decoding pipeline, both codes reach below half-gigaquop LER, which suggests that additional hyper-parameter optimizations for escalation eagerness may further improve the performance of the hierarchical decoder.

\begin{table}
    \centering
    \renewcommand{\arraystretch}{1.25}
    \begin{tabular}{c|cc}\toprule
         &  $[[132,30,12]]$&  $[[672,336,12]]$\\\midrule
         $T_1+T_2+T_3$&  $1.0_{-0.53}^{+0.90}\times 10^{-9}$ &  $4.3_{-0.80}^{+1.1}\times 10^{-9}$\\
         $T_1+T_2$&  $9.6_{-1.6}^{+1.8}\times 10^{-9}$&  $ 5.1_{-0.86}^{+0.99}\times 10^{-9}$\\
         $T_3$ Only&  $<  3.5\times 10^{-10}$&  $4_{-3.9}^{+18}\times 10^{-11}$\\
         $T_3$ Errors&  0/48,858,000&  1/11,768,000\\ \bottomrule
    \end{tabular}
    \caption{Circuit Level Logical Error Simulations using the hierarchical  decoder.}
    \label{tab:ler-t3}
\end{table}

\begin{table}
    \centering
    \begin{tabular}{c|cc}\toprule
         &  $[[132,30,12]]$& $[[672,336,12]]$\\\midrule
         Escalations &  99 & 19 \\
         False Negatives &  11 & 102 \\
Escalations Failures & 0/99 &0/19 \\
False Negative Failures& 0/11 &1/102  \\ \bottomrule
    \end{tabular}
    \caption{Break down on MILP solves. }
    \label{tab:ler-t3-stats}
\end{table}

Secondly, we estimated the pseudo-threshold of $[[132,30,12]]$ to be $0.44\pm 0.02\%$ using a LER sweep near threshold with a two-stage decoder, where the threshold-theorem and polynomial-exponent fits agree. For $[[672,336,12]]$ the two fit forms disagree by a factor of two, so we instead bracket its pseudo-threshold directly from the measured points, between $0.4\%$ and $0.5\%$.

\begin{figure}[h]
    \centering
    \includegraphics[width=0.95\linewidth]{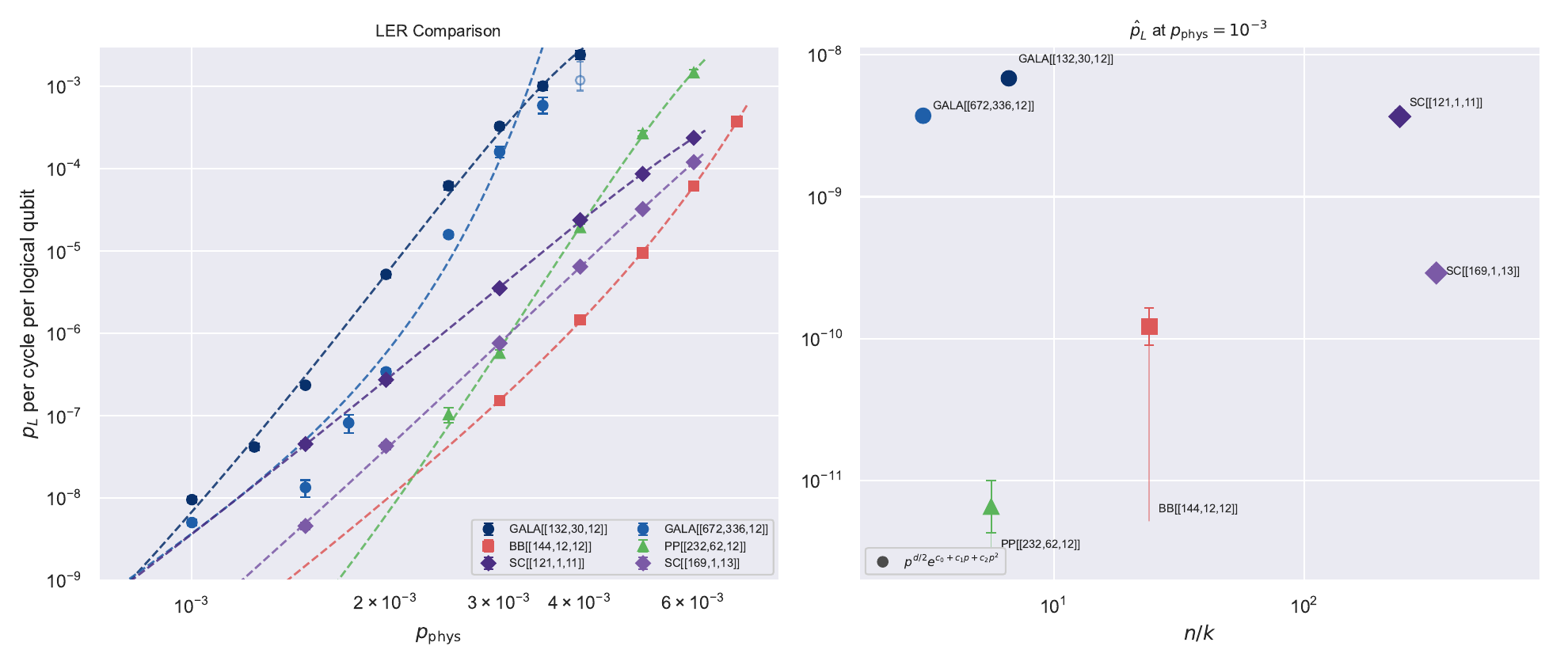}
    \caption{Circuit Level Logical Error Simulations using the $T_1+T_2$ decoder, estimating a pseudo-threshold of $0.395\pm{0.015}\%$. $T_3$ is not used for the pseudo-threshold experiment. \emph{Left:} Each solid data point correspond to a decoder experiment, where the light data points have less than 10 recorded failures and are not used in the fit. We fit the data points to the polynomial-exponent form $p^{d/2} \exp(c_0 + c_1 p + c_2 p^2)$ \cite{Bravyi_2024}. Note the [[132,30,12]] GALA codes reach $p_L = 8.6\times 10^{-9}$ at $p_{phys} = \times 10^{-3}$. \emph{Right:} Logical logical error rates at $10^{-3}$ (extrapolated for baselines). GALA achieves similar LER performances as a $d =11$ surface codes codes at much lower qubit overhead.}
    \label{fig:ler}
\end{figure}

}

\lead{Addressable and AOD Compatible Logical ISA} Combined with the $C_m$ shift-automorphisms, a self-dual GALA code carries logical $H$, $S$, and logical shifts implemented entirely by moves \edit{depth-1} Clifford gates. Using the $C_m$ shift-automorphisms, we can generate structured logical Pauli basis with logical permutations, whereas the ZX-duality $\tau$ supplies transversal $H$ and $S$. Furthermore, if the abelian factor $C_m=\Z_{p_1}\times\cdots\times\Z_{p_q}$ is a product of cyclic groups, each factor contributes an independent shift generator, and collapsing any subset of factors yields a chain homomorphism onto the corresponding quotient code. \edit{If additionally the code has a trivial top, the meta duality also gives us a fault tolerant logical reversal from physical reversals, and a logical CZ that may have fault distance $< d$. For example, we can show that the $[[672, 336, 12]]$ qubit code contains a fold-tranvsersal CZ gate with fault disatnce $d_{\bar S} = 10$.}  In principle, the available gadgets are sufficient to implement arbitrary Clifford operations, but additional work is required to ensure fault tolerance. We leave the detailed analysis of this for future work. \edit{While not in general fault tolerant, the ZX duality and meta duality can nevertheless help us when searching for low-weight logical representatives. We call the orbit class of a seed logical $X_s = \{X\cdot \sigma | \sigma \in G_{\Aut}  \}$; existence of such ZX dualities help us constraint the search space for low weight seeds to only \emph{palindromic} words fixed by a reflection, whichs live in a folded code on the $n/2$ reflection pairs, halfing the search space. Pairng with translation symmetry further reduces the symmetry, making exhaustive searches up to weight $20$ in small codes feasible.

We report the weight spectrum of the logical obtained for ZX dual codes in Table~\ref{tab:isa-stats}, along with the available addressable Clifford and the fault-distances in section \ref{sec:case-study}.}

\begin{table}[t]
    \centering
    \footnotesize
    \setlength{\tabcolsep}{3pt}
    \renewcommand{\arraystretch}{1.15}
    \begin{tabular}{lccc}\toprule
         & $[[132,30,12]]$ & $[[136,34,12]]$ & $[[672,336,12]]$\\\midrule
        $C_m$ & $\Z_{11}$ & $\Z_{17}$ & $\Z_2\times \Z_3\times \Z_{14}$\\
        Orbit Class & $2\times 11+8$ & $2\times 17$ & $4\times 84$\\
        $G_{\Aut}$ & $D_{11}$ & $D_{17}$ & $\Z_2\times C_m$\\
        Weight Spectrum & (14,14,22) & (12,14,16) & (16,18,20)\\
        \bottomrule
    \end{tabular}
    \caption{\edit{Addressable logical ISA from symmetries. \emph{$G_\Aut$} refers to the logical automorphism group of the code. $D_n$ refers to the dihedral group (cyclic shifts with a reflection).
    The \emph{weight spectrum} reports the (Min, Median, Max) weight of a complete, orthonormal, $G_{\Aut}$-covariant basis, obtained via numerical search. The basis is given in Section\ref{sec:case-study}.  
    }
    }
    \label{tab:isa-stats}
\end{table}

\lead{Outlook} GALA gives a theory to co-design qLDPC codes for hardware and application. The group product lift explicates the algebra; the non-abelian factor improves code parameters; and the abelian factor supplies the symmetries that yield parallel AOD moves and transversal Cliffords. Several future directions follow. First, there remains gaps to use them for computation, such as existence and discoveries of disjoint, minimum weight, orthonormal logical basis and fault tolerance of the logical gates in practice. It is worth noting that while the constructive quotient logical operators we give do achieve minimum weight for the majority of logical representatives for the $[[132,30,12]]$, it can in general be at least $k\times$ larger. Secondly, there is also no inherent restriction for the abelian groups to be quasi-cyclic; generalizing may yield more interesting logical automorphisms. Alternatively, we can also consider different proto-matrices. More broadly, because every logical primitive here is a parallel atom move or a depth-1 transversal gate, GALA offers a path to compiling ultra-high-rate fault-tolerant logical circuits on today's few-hundred-qubit devices whose cost is dominated by the same rearrangements that implement error correction. 

\lead{Acknowledgments}
We thank Victor Albert, Hossein Dehghani, Michael Gullans, Eric Huang, Maryam Mudassar, Chen Zhao, and Hengyun Zhou for discussions and comments. We also thank Hossein Dehghani and Ming Wang for help with decoding simulations and analysis. 

\lead{AI Contribution}
We used Claude models including Opus 4.8, Opus 5, Fable 5.0, and Fable 5.1 for software implementation and exposition, as detailed below. The authors have reviewed all content and take full responsibility for the results.

\noindent \textit{Software implementation}. The bulk of the code was written by generative AI for the code search and the logical-basis search. All reported code parameters, including block length, dimension, girth, and in particular the exactly certified distances, were verified independently by the authors.

\noindent \textit{Exposition}. All theorems, proofs, definitions, and worked examples were originally written and verified by the authors, and generative AI were used to polish the writing and improve presentation.  Most candidate figures were generated by Claude and remade in Figma.

\lead{Note} A provisional patent associated with this work was filed on 7/24/2026. After completion of this work, an independent preprint, Ref.~\cite{lu2026cornucopia} introduced ``Cornucopia'' codes, built from affine permutations on a $\Z_3\times\Z_q$ grid---the direct-product sector of the present framework with $H_k=S_3$, $C_m=\mathbb Z_q$. Our work is complementary: we give a general construction framework which also includes semidirect products, polynomial and ZX-dual generalizations, and provides closed-form bounds, and fold-transversal Cliffords. Their $[[1044,526,12]]$ and $[[1764,886,14]]$ codes are close in parameters to our self-dual $[[1056,532,12]]$, $[[1752,880,14]]$ and $[[1764,886,14]]$, but are discovered independently.

\bibliography{references}

\clearpage
\onecolumngrid

\setcounter{secnumdepth}{3}
\setcounter{section}{0}
\setcounter{equation}{0}
\setcounter{figure}{0}
\setcounter{table}{0}
\renewcommand{\thesection}{S\arabic{section}}
\renewcommand{\thesubsection}{\thesection.\arabic{subsection}}
\renewcommand{\thesubsubsection}{\thesubsection.\arabic{subsubsection}}
\makeatletter
\renewcommand{\p@subsection}{}
\renewcommand{\p@subsubsection}{}
\makeatother
\renewcommand{\theequation}{S\arabic{equation}}
\renewcommand{\thefigure}{S\arabic{figure}}
\renewcommand{\thetable}{S\arabic{table}}
\renewcommand{\theHequation}{S\arabic{equation}}
\renewcommand{\theHfigure}{S\arabic{figure}}
\renewcommand{\theHtable}{S\arabic{table}}
\setcounter{tocdepth}{2}

\newcommand{\GALAsmheader}{%
  \clearpage                       
  \ifdefined\onecolumngrid\onecolumngrid\fi   
  %
  \setcounter{section}{0}%
  \setcounter{equation}{0}%
  \setcounter{figure}{0}%
  \setcounter{table}{0}%
  \renewcommand{\thesection}{S\arabic{section}}%
  \renewcommand{\thesubsection}{\thesection.\arabic{subsection}}%
  \renewcommand{\theequation}{S\arabic{equation}}%
  \renewcommand{\thefigure}{S\arabic{figure}}%
  \renewcommand{\thetable}{S\arabic{table}}%
  %
  \providecommand{\theHsection}{}\renewcommand{\theHsection}{SM.\arabic{section}}%
  \providecommand{\theHsubsection}{}\renewcommand{\theHsubsection}{SM.\arabic{section}.\arabic{subsection}}%
  \providecommand{\theHequation}{}\renewcommand{\theHequation}{SM.\arabic{equation}}%
  \providecommand{\theHfigure}{}\renewcommand{\theHfigure}{SM.\arabic{figure}}%
  \providecommand{\theHtable}{}\renewcommand{\theHtable}{SM.\arabic{table}}%
  %
  \begin{center}
    {\large\bfseries Supplemental Material}
  \end{center}
  \vspace{10pt}
  \ifdefined\onecolumngrid\onecolumngrid\fi   
}

\GALAsmheader

\vspace{0.5em}
\begin{center}
\begin{minipage}{0.88\linewidth}
In this supplementary paper we give a detailed construction of the GALA codes, prove bounds on their rate, distance, and girth, describe the algebraic structure of their logical operators, and give fault-tolerant logical gadgets and AOD-compatible move schedules derived from symmetries that are design inputs to the construction. We then report the code search and its numerical results, and work one instance out end to end: the self-dual $[[132,30,12]]$ code, including explicit low-weight logical representatives, its cyclic shift automorphisms, and its fold-transversal Clifford gates.
\end{minipage}
\end{center}
\vspace{1em}

\tableofcontents
\clearpage

\section{Backgrounds}\label{sec:bakcgournd}
We will begin by covering some basics on quantum error correcting codes (QECC) and their logical operators, as well as describe Kasai's ultra-high rate codes. A glossary of notations is given in Section~\ref{ssec:notations}. 

\subsection{Quantum Error Correcting Codes}\label{ssec:qec-basics}

We say a quantum error correcting code (QECC) is a $[[n,k,d]]$ code if it encodes $k$ logical qubits in the entangled states of $n$ physical qubits, and no Pauli operator of weight less than $d$ maps one logical state to a different logical state. A QECC protects logical information against noise by repeatedly measuring a set of \textit{stabilizers}, which projects the joint state of the physical qubits onto a logical subspace following Pauli corrections determined by the stabilizer measurements. The \emph{encoding rate}, $\frac{k}{n}$, describes the space efficiency in terms of the number of physical qubits needed to encode a logical qubit, and $d$ describes the error correction capacity.

A Pauli operator can be represented as a binary vector $v\in \mathbb F_2^n$, where the $X$-type Pauli $X(v)$ and $Z$-type Pauli $Z(v)$ correspond to
\[
X(v) = \bigotimes_{i\in[n]} X_i^{v_i}, \qquad
Z(v) = \bigotimes_{i\in[n]} Z_i^{v_i},
\]
and the Hamming weight of $v$, $\wt{v} = |\{i: v_i = 1\}|$, corresponds to the weight of the  Pauli operator. Two Pauli operators of different type, i.e.\ $X(v), Z(u)$, commute if and only if their overlap $\rs{v,u} := \sum_i v_iu_i \bmod 2$ vanishes:
\begin{equation}\label{eq:binary-commutation}
    X(v)Z(u) = (-1)^{\rs{v,u}}\, Z(u) X(v).
\end{equation}

A Calderbank-Shor-Steane (CSS) code is a QECC whose stabilizers are each purely
of $X$- or $Z$-type; by \eqref{eq:binary-commutation}, commutativity reduces to orthogonality of the pair of $X$- and $Z$-stabilizer matrices:

\begin{definition}[CSS Code]\label{def:css}
A CSS code $\mathcal Q = \CSS{H_X, H_Z}$ on $n$ qubits is specified by a pair of binary \emph{stabilizer matrices} $H_X \in \mathbb F_2^{m_X\times n}$ and $H_Z\in \mathbb F_2^{m_Z\times n}$ satisfying the \emph{orthogonality condition}
\begin{equation}\label{eq:css-orthogonality}
    H_X H_Z^T = 0 .
\end{equation}
Each row $x$ of $H_X$ (resp.\ $z$ of $H_Z$) is a Pauli stabilizer $X(x)$ (resp.\ $Z(z)$), and they generate the stabilizer group
\[
\mathcal S = \rs{\, X(x),\, Z(z) \;|\; x = (H_X)_i,\ z = (H_Z)_j,\ 0<i\leq m_X,\ 0<j\leq m_Z\,}.
\]
The logical space of a CSS code is the joint $+1$-eigenspace
$
\mathcal Q = \{\ket \psi : S\ket\psi = \ket \psi \ \forall S \in \mathcal S\},
$
which encodes
\begin{equation}\label{eq:css-k}
    k = n - \rank(H_X) - \rank(H_Z)
\end{equation}
logical qubits.
\end{definition}

The CSS orthogonality condition~\eqref{eq:css-orthogonality} ensures that $\mathcal S$ defines a set of mutually commuting observables, which is necessary for the logical space to be well defined and non-trivial.

Properties of a CSS code can be analyzed using the \emph{Tanner graph} of its check matrices:
\begin{definition}[Tanner Graph Properties]\label{def:qldpc}
Let the \emph{Tanner graph} $\mathcal{T}(H) = G(V_H, E_H)$ of a check matrix $H\in \mathbb F_2^{m\times n}$ be the bipartite graph on vertices $V_H = [n]\sqcup[m]$ with data qubit nodes $[n]$ and check nodes $[m]$, and with edges $E_H = \{(i,j)|H_{i,j} = 1\}$; and let $\mathcal{T}(\mathcal{Q}) =G(V_\mathcal{Q}, E_\mathcal{Q}) =  \mathcal{T}(H_X)\sqcup\mathcal{T}(H_Z)$ denote the Tanner graphs of a CSS code $\mathcal{Q}(H_X,H_Z)$ as in Def.~\ref{def:css}. Then, the CSS code $\mathcal{Q}$ is \emph{qLDPC} if $\deg(v)\leq w$ for all nodes $v\in V_\mathcal{Q}$; or, equivalently, if the row and column weights of both $H_X$ and $H_Z$ are bounded above by a constant $w$. Furthermore, we define the girth of $\mathcal{Q}$ as $\gth{\mathcal Q} = \gth{\mathcal{T}_\mathcal{Q}}$.
\end{definition}
\begin{lemma}[Four-cycle Counts]\label{lemma:girth}
Tanner graphs are bipartite, so girths are even and at least $4$. The number of distinct four-cycles $t_4$ is given by:
\begin{equation}\label{eq:four-cycle-count}
    t_4 = \sum_{0\leq i< j< m_X} \binom{(H_XH_X^T)_{i,j}}{2},
\end{equation}
where the summation is taken over $\mathbb Z$. In particular $\gth{H_X}\geq 6$ if and only if $(H_XH_X^T)_{i,j}\leq 1$ for all $i\neq j$. The same statements hold for $H_Z$.
\end{lemma}
\begin{proof}
Let $c_0,\dots,c_{m_X-1}$ be the check nodes of $\mathcal T(H_X)$ and write $N(c_i) = \supp{(H_X)_i}$. A bipartite graph has no odd cycle and a simple graph has no $2$-cycle, so every cycle length is even and at least $4$. Over $\mathbb Z$,
\begin{equation}\label{eq:gram-overlap}
    (H_XH_X^T)_{i,j} = \sum_{a\in[n]}(H_X)_{i,a}(H_X)_{j,a} = |N(c_i)\cap N(c_j)|,
\end{equation}
since each summand is $1$ exactly when $a$ lies in both supports; over $\mathbb F_2$ only the parity of this overlap would survive.

Map a four-cycle to the pair consisting of its two check indices and its two qubit indices. This is well defined: consecutive vertices of a cycle lie in opposite parts of the bipartition, so a four-cycle alternates and carries exactly two checks and two qubits, say $c_i - q_a - c_j - q_b - c_i$; its vertices are distinct, so $i\neq j$ and $a\neq b$, and all four edges are present, so $a,b\in N(c_i)\cap N(c_j)$. It is injective: the image determines the vertex set $\{c_i,c_j,q_a,q_b\}$, and a four-cycle on that set is unique, because bipartiteness leaves only the four edges $c_iq_a$, $c_iq_b$, $c_jq_a$, $c_jq_b$, and a cycle through all four vertices must give each of them degree $2$ and hence use all four. It is surjective onto the pairs $(\{i,j\},\{a,b\})$ with $i\neq j$, $a\neq b$ and $\{a,b\}\subseteq N(c_i)\cap N(c_j)$, since those edges then form a four-cycle. Grouping by the check pair, each $\{i,j\}$ with $i<j$ admits exactly $\binom{|N(c_i)\cap N(c_j)|}{2}$ choices of $\{a,b\}$, and substituting \eqref{eq:gram-overlap} gives \eqref{eq:four-cycle-count}. Finally the summands are non-negative and $\binom N2 = 0$ iff $N\leq 1$, so $t_4 = 0$ iff every off-diagonal Gram entry is at most $1$.
\end{proof}

\subsubsection{Logical Basis}

Logical qubits of a QECC are usually encoded in entangled $n$-qubit stabilizer states that lack a compact description; therefore, we often use the corresponding logical Pauli operators instead to describe the logical space:

\begin{definition}[Logical Operators and Distance]\label{def:logicals}
Let $\mathcal Q = \CSS{H_X, H_Z}$. The non-trivial logical Pauli operators of $\mathcal{Q}$ are given by 
\begin{align*}
    \mathcal{L}_X = \ns{H_Z}\,/\,\rs{H_X}\;\cong\; \mathbb F_2^{k},\qquad
    \mathcal{L}_Z = \ns{H_X}\,/\,\rs{H_Z}\;\cong\; \mathbb F_2^{k}. 
\end{align*}
The distance of $\mathcal{Q}$ is $d = \min\{d_X, d_Z\}$, where
\begin{equation}\label{eq:css-distance}
    d_X = \min\{\wt{v}: v \in \ns{H_Z}\setminus \rs{H_X}\},\qquad
    d_Z = \min\{\wt{u}: u \in \ns{H_X}\setminus \rs{H_Z}\}.
\end{equation}
\end{definition}
 
It is best to use 
 
\begin{definition}[Orthonormal Logical Basis]\label{def:logical-basis}
A pair of matrices $L_X, L_Z \in \mathbb F_2^{k\times n}$ is a \emph{symplectic logical basis} for
$\mathcal Q = \CSS{H_X,H_Z}$ if
\[
H_ZL_X^T  = 0,\qquad H_XL_Z ^T = 0, \qquad L_XL_Z^T = I_k,
\]
and the rows of $L_X$ (resp.\ $L_Z$) represent a basis of $\ns{H_Z}/\rs{H_X}$ (resp. $\ns{H_X}/\rs{H_Z}$). Such a pair always exists by Gaussian Elimination.
\end{definition}

In classical codes, we can often increase the rate of a code by removing rows of the check matrix to release more logical degrees of freedom, also referred to as a code-modification \emph{augmentation}. However, for quantum codes, the immediate consequence of CSS orthogonality is that any removed stabilizer generator becomes a logical operator itself, unless it was redundant to begin with; hence, if the code we start from is LDPC, the distance after augmentation will be bounded by the same constant $w$:

\begin{fact}[Augmentation Barrier]\label{fact:augmentation}
Let $\mathcal Q = \CSS{H_X,H_Z}$ be a qLDPC CSS code with Tanner-graph degrees $\deg (v)\leq w$ for all $v\in V_{\mathcal Q}$, as in Definition~\ref{def:qldpc}, and let
$\mathcal{Q}^S = \CSS{H^S_X, H^S_Z}$ be an \emph{augmentation} of $\mathcal Q$, obtained by deleting a set $S$ of rows from $H_X$ and from $H_Z$. Suppose some deleted row is not in the row space of the retained rows of its own type; say a deleted row $x$ of $H_X$ has $x\notin \rs{H^S_X}$. Then
\begin{equation}\label{eq:augmentation-barrier}
    d_{\mathcal{Q}^S}\leq d_X(\mathcal Q^S)\leq \wt{x}\leq w.
\end{equation}
The same conclusion holds with the roles of $X$ and $Z$ exchanged.
\end{fact}
\begin{proof}
Since $\mathcal Q$ is a CSS code, $H_XH_Z^T = 0$, and $H^S_Z$ is a submatrix of $H_Z$, so $x (H^S_Z)^T = 0$, i.e.\ $x\in \ns{H^S_Z}$: the operator $X(x)$ commutes with every retained $Z$-stabilizer. By hypothesis $x\notin\rs{H^S_X}$, so $x$ represents a non-trivial class of $\ns{H^S_Z}/\rs{H^S_X} = \mathcal L_X(\mathcal Q^S)$, i.e.\ $X(x)$ is a non-trivial logical $X$ operator of $\mathcal Q^S$. Hence $d_X(\mathcal Q^S)\leq\wt x$, and $\wt x$ is the degree of the check node of $\mathcal T(H_X)$ carrying the row $x$, which is at most $w$.
\end{proof}

Note a linearly independence of the stabilizer is necessary--if the deleted row already lies in the row space of the retained rows of its own type, then $\mathcal Q^S$ and $\mathcal Q$ have the same stabilizer group and hence the same distance. We will assume this to be the case in this work unless otherwise specified.

While not inherently tied to the error correction capacity of a QECC, small cycles in the Tanner graph can lead to failures for  message passing based decoders. We therefore often require $\gth{\mathcal{Q}}\geq 6$ for practical decoder performances, which by Lemma~\ref{lemma:girth} is the condition that every off-diagonal entry of $H_XH_X^T$ and of $H_ZH_Z^T$ be at most $1$ over $\mathbb Z$.

\subsubsection{Chain Complexes}\label{sssec:chain-complex}

We can view a CSS code as length-$2$
chain complex with $H_Z^T,H_X $ as boundary maps with CSS orthogonality corresponding to the boundary map condition:

\begin{definition}[Chain Complex]\label{def:chain-complex}
A length-$\ell$ chain complex over $\mathbb F_2$ is a collection of vector spaces $\{B_i\}_{i=0}^\ell$ together with boundary maps $\partial_i : B_i \to B_{i-1}$ satisfying $\partial_{i-1}\partial_i = 0$, which we abbreviate as $\partial^2 = 0$. A classical code with check matrix $H$ is the length-$1$ complex $C_1 \xrightarrow{H} C_0$, where the bases of $C_1$ and $C_0$ are labeled by bits and checks. A CSS code $\mathcal Q = \CSS{H_X, H_Z}$ is the length-$2$ complex
\begin{equation}\label{eq:css-complex}
    Z \xrightarrow{\ H_Z^T\ } Q \xrightarrow{\ H_X\ } X,
\end{equation}
where the bases of $Z,Q,X$ are labeled by the $Z$-checks, the data qubits, and the $X$-checks respectively, and $\partial_1\partial_2 = H_XH_Z^T = 0$ is exactly~\eqref{eq:css-orthogonality}.

The co-chain complex is the dual:
\begin{equation}\label{eq:css-cocomplex}
    X \xrightarrow{\ H_X^T\ } Q \xrightarrow{\ H_Z\ } Z,
\end{equation}
with coboundary maps $\delta^{i} := \partial_{i+1}^T$; i.e.\ $\delta^0 = H_X^T$ and
$\delta^1 = H_Z$. 
\end{definition}

Logical operators are then the (co)homology of \eqref{eq:css-complex}:
\begin{equation}
    \mathcal L_Z = \ker \partial_1/\operatorname{im}\partial_2 = \ns{H_X}/\rs{H_Z},
    \qquad
    \mathcal L_X = \ker \partial_2^{T}/\operatorname{im}\partial_1^{T} = \ns{H_Z}/\rs{H_X},
\end{equation}

\subsection{Fault Tolerant Logical Operations}\label{ssec:ft-logic}
In the most general definition, a logical operation is the induced transformation of the logical space from a set of physical operations on the data qubits that preserve the logical space; that is, they map an eigenstate to another eigenstate. The logical operation is trivial if it fixes the eigenstates pointwise, and a logical operation is fault-tolerant if failures of these physical operations can be detected and corrected for up to the same distance $d$. An example of a logical operation that is both fault tolerant and trivial is the application of a Pauli stabilizer. We are interested in logical operations that are both non-trivial and fault-tolerant as means to manipulate the encoded logical information. 

There are two main approaches to enable fault tolerant Clifford operations in qLDPC codes: transversal gates, and code deformation. Transversal gates can typically be executed with $O(1)$ time overhead (at least amortized), though their availability depends on the symmetry of the underlying code. While we can always achieve trivial symmetries such as the transversal CNOTs between two identical codes, it is often challenging to attain a useful set of logical gates via happenstance symmetries alone \cite{yang2026spacetimeefficienthardwarecompatiblecomplexquantum}. On the other hand, while code deformation allows us to obtain arbitrary Pauli product measurements by a targeted addition of fault tolerant checks, it typically incurs space-time overheads on the order of $d$ to ensure reliable measurements of the added stabilizers if not carefully optimized. In this work, we mainly investigate gates of the first flavor to prioritize minimizing time overheads, but the symmetries we describe are not specific to that choice: they pay off for surgery- and extractor-based architectures too.

\subsubsection{Virtual Logical Clifford from Code Automorphisms}
A permutation transversal gate is induced by a permutation of the data qubits that preserves the logical space, also called a code automorphism. They are particularly relevant for RNAA devices, where reconfiguration can be performed with AODs in parallel and with high fidelity. Formally:

\begin{definition}[Code Automorphism]\label{def:code-automorphism}
Let $\mathcal{Q}$ be an $[[n,k,d]]$ CSS code specified by the stabilizers $H_X, H_Z$ and symplectic logical basis $L_X, L_Z$. An $n$-dimensional permutation $\sigma$ is said to be a code automorphism of $\mathcal{Q}$ if its right action on the qubit labels preserves the stabilizer group of $\mathcal Q$ sector-wise; that is,
\[
\rs{H_X\sigma} = \rs{H_X},\qquad \rs{H_Z\sigma} = \rs{H_Z},
\]
which we abbreviate as $\mathcal Q\cdot\sigma = \mathcal Q$. Since a permutation preserves overlaps, such a $\sigma$ also maps $\ns{H_Z}$ onto itself, and therefore descends to a linear automorphism of $\mathcal L_X = \ns{H_Z}/\rs{H_X}$ and of $\mathcal L_Z = \ns{H_X}/\rs{H_Z}$. We record the induced logical action as the pair $(g_X, g_Z)$ of $k\times k$ invertible matrices determined by
\[
g_X L_X \equiv L_X \sigma \pmod{\rs{H_X}}, \qquad
g_Z L_Z \equiv L_Z \sigma \pmod{\rs{H_Z}};
\]
the congruences are equalities whenever $\sigma$ happens to fix the chosen representatives, and we then write $g_XL_X = L_X\sigma$.
\end{definition}

We have the following facts:
\begin{fact}[Automorphism Facts]\label{fact:aut-facts}
Let $\mathcal{Q}$ be an $[[n,k,d]]$ CSS code specified by the stabilizers $H_X, H_Z$ and symplectic logical basis $L_X, L_Z$. Let $\mathrm{Aut}(\mathcal{Q}) = \{\sigma | \mathcal{Q}\cdot \sigma = \mathcal{Q}\}$ be the set of code automorphisms and let $\mathrm{G}(\mathcal{Q})$ be the set of logical actions $(g_X,g_Z)$ induced by the elements of $\mathrm{Aut}(\mathcal Q)$, as in Definition~\ref{def:code-automorphism}. Then:

\begin{enumerate}
    \item \textbf{Group Homomorphism \cite{berthusen2025automorphismgadgetshomologicalproduct}:} The set of all code automorphism $\mathrm{Aut}(\mathcal{Q})$ and induced logical actions $\mathrm{G}(\mathcal{Q})$ both form a group, and
    \[
    \phi : \mathrm{Aut}(\mathcal{Q}) \to \mathrm{G}(\mathcal{Q}), \qquad \sigma \mapsto (g_X, g_Z)
    \]
    is a group homomorphism.
    \item \textbf{Virtuality:} Suppose $\sigma\in\mathrm{Aut}(\mathcal Q)$ preserves the two stabilizer \emph{matrices} up to a permutation of their rows; that is, suppose there exist row permutations $\rho_X\in S_{m_X}$, $\rho_Z\in S_{m_Z}$ with
    \[
    \rho_X H_X = H_X \sigma, \qquad \rho_Z H_Z = H_Z \sigma.
    \]
    Then the induced logical operation $\phi(\sigma) = (g_X,g_Z)$ is realized by a \emph{virtual relabeling} of the physical qubits: after renaming qubit $i$ as $\sigma(i)$ and renaming the $X$- and $Z$-checks by $\rho_X$ and $\rho_Z$, one recovers the original code together with its original syndrome-extraction schedule, generator by generator. No physical gate is therefore required to implement $\phi(\sigma)$ beyond the rearrangement itself: the gate is a relabeling that can be absorbed into the classical bookkeeping of the decoder, or, on hardware where a fixed physical layout must be restored, into a single atom-rearrangement step.
\end{enumerate}
\end{fact}

Furthermore, we assume that qubit permutation do not spread errors.

\subsubsection{Fold-transversal Logical Clifford from ZX-Duality}
ZX duality also results from a permutation symmetry of the check matrices, but rather than mapping $X$-type stabilizers and $Z$-type stabilizers back to themselves, it maps $X$-type checks to $Z$ and vice versa. Formally,
\begin{definition}[ZX Duality]\label{def:zx-duality}
Let $\mathcal{Q}$ be an $[[n,k,d]]$ CSS code specified by the stabilizers $H_X, H_Z$ and symplectic logical basis $L_X, L_Z$, with $m_X = m_Z$. An $n$-dimensional permutation $\tau$ is said to be a ZX duality of $\mathcal{Q}$ if $\tau$ exchanges $H_X$ and $H_Z$ by right action; that is,

\[
H_X\cdot \tau = H_Z,  \quad H_Z\cdot\tau  = H_X.
\]
 $\tau$ is an involution:
\[
H_X\tau^2 = H_Z\tau = H_X,\qquad H_Z\tau^2 = H_X\tau = H_Z.
\]
\end{definition}

Codes with a ZX duality can implement a fold-transversal $H$ gate and phase-type gate, as given in \cite{Breuckmann_2024}:

\begin{fact}[Fold-transversal Gates \cite{Breuckmann_2024}]\label{fact:fold-transversal}
Let $\mathcal Q$ be a CSS code with a ZX duality $\tau$. The following physical operations, up to a Pauli frame update, induce fold-transversal logical operations on $\mathcal Q$:
\begin{align}
    H_\tau &= \bigotimes_{i \in [n],i\leq \tau(i)} \mathrm{SWAP}_{i, \tau(i)} \times \bigotimes_{i \in [n]} H_i,\\
    S_\tau &= \bigotimes_{i \in [n],i\leq \tau(i)} \mathrm{CZ}_{i, \tau(i)} \times \bigotimes_{i \in [n], i = \tau(i)} S_i^\dagger.
\end{align}

\end{fact}

\subsubsection{Homomorphic CNOTs/Measurements from Chain Homomorphisms}
 
The two symmetries above relate a code to itself. Symmetries relating \emph{two} codes also give
rise to transversal gates, and the chain complex formalism of Section~\ref{sssec:chain-complex} is
the natural language for them: a structure-preserving map between two codes is a map between their
complexes that commutes with the boundary maps.
 
\begin{definition}[Chain Homomorphism]\label{def:chain-map}
Let $\mathcal Q, \mathcal Q'$ be CSS codes with boundary maps $\partial_i, \partial'_i$ as in Definition~\ref{def:chain-complex}. A chain homomorphism
$\mathbf \gamma = \{\gamma_z, \gamma_q,\gamma_x\}$ is a triple of linear maps such the following diagram commute:
\begin{equation}\label{eq:chain-map}
    \begin{tikzcd}[ampersand replacement=\&]
        {Z} \& {Q} \& {X} \\
        {Z'} \& {Q'} \& {X'}
        \arrow["{H_Z^T}", from=1-1, to=1-2]
        \arrow["{H_X}", from=1-2, to=1-3]
        \arrow["{H_Z'^T}", from=2-1, to=2-2]
        \arrow["{H_X'}", from=2-2, to=2-3]
        \arrow["{\gamma_z}", from=1-1, to=2-1]
        \arrow["{\gamma_q}", from=1-2, to=2-2]
        \arrow["{\gamma_x}", from=1-3, to=2-3]
    \end{tikzcd}
\end{equation}
That is, $\gamma_z: Z\to Z'$, $\gamma_q: Q\to Q'$ and $\gamma_x: X\to X'$ are named after the space they map \emph{out of}, and commutativity of the two squares reads
\begin{equation}\label{eq:chain-map-squares}
    \gamma_q H_Z^T = H_Z'^T\gamma_z,\qquad
    \gamma_x H_X = H_X'\gamma_q .
\end{equation}
\end{definition}
 
A chain homomorphism preserves the (co)homology of the complex. That is, $\mathbf\gamma$ induces a map $\mathcal L'_Z \to \mathcal L_Z$ on the $Z$-sector and $\mathcal L_X \to \mathcal L'_X$ in the opposite direction. Note that the previous two subsubsections are the special case $\mathcal Q = \mathcal Q'$ with $\gamma_q$ a permutation: a code automorphism is a chain endomorphism of the code, the right-hand square of \eqref{eq:chain-map-squares} becoming $\gamma_x H_X = H_X\gamma_q$, which is the virtuality condition $\rho_XH_X = H_X\sigma$ of Fact~\ref{fact:aut-facts} under $\gamma_x = \rho_X^{-1}$, $\gamma_q = \sigma^{-1}$ (permutations act on the right of check matrices but on the left of chain groups, so the two conventions differ by an inverse, which for a permutation matrix is a transpose). A ZX duality is instead an isomorphism onto the cochain complex \eqref{eq:css-cocomplex}.
 
These homomorphisms identify correspondences between logical qubits of different code blocks, which results in transversal CNOT operators:
 
\begin{definition}[Homomorphic CNOT~\cite{huang2022homomorphiclogicalmeasurements}]\label{def:homomorphic-cnot}
Let $\mathbf \gamma: \mathcal Q \to \mathcal Q'$ be a chain homomorphism as in Definition~\ref{def:chain-map}. The physical $\mathcal Q$-controlled CNOTs specified by $\gamma_q$, i.e.\ a CNOT controlled on qubit $i$ of $\mathcal Q$ and targeting qubit $j$ of $\mathcal Q'$ applied whenever $\gamma_q[i,j] = 1$, implement a set of $\mathcal Q$-controlled logical CNOTs between the two blocks. Here and below we write $\gamma_q$ as the $n\times n'$ incidence matrix of this CNOT pattern, with rows indexed by the qubits of $\mathcal Q$ and columns by those of $\mathcal Q'$, so that it acts on the right of Pauli row vectors; it is the transpose of the vertical arrow of~\eqref{eq:chain-map}. In particular, let $L_X, L_Z$ and $L'_X, L'_Z$ be symplectic logical bases for
$\mathcal Q, \mathcal Q'$. The logical circuit is described by
\[
    \prod_{(i,j)\,:\,M_{i,j} = 1} \mathrm{CNOT}(\bar Q_i, \bar Q'_j),
    \qquad\text{where}\qquad
    M = L_X\,\gamma_q\, L'^T_Z \in \mathbb F_2^{k\times k'}.
\]
\end{definition}

The pattern matrix $M$ follows from propagation, since a logical $\bar X_i$ with representative $(L_X)_i$ picks up $(L_X)_i\gamma_q$ on $\mathcal Q'$, whose class is named by pairing against $L'_Z$.
 
We can use these transversal CNOTs to implement logical Pauli product measurements (PPMs). To measure $Z$-type products on a data code $\mathcal Q$, we prepare the logical qubits of an ancilla code $\mathcal Q'$ in the $Z$ basis, apply the $\mathcal Q$-controlled homomorphic CNOTs, and measure $\mathcal Q'$ transversally in the $Z$ basis; the $X$-type version reverses the control and both bases. What gets measured is the set of products named by the \emph{columns} of $M$ --- measuring ancilla logical $j$ returns $\prod_i\bar Z_i^{\,M_{ij}}$ --- so the gadget performs $k'$ measurements in parallel, with single logical qubit measurements being the special case where each column of $M$ has weight one. The scheme is selective, in the sense that the choice of ancilla determines which products are measured.
 
Furthermore, if either block has non-trivial code automorphisms or ZX dualities, the same physical CNOT pattern gives access to a larger family of measurement patterns by composing $\mathbf\gamma$ with the permutations $\sigma$ or $\tau$. Designing the lift so that these symmetries are available by construction is therefore doing double duty, and we
return to it in Section~\ref{sec:logic}.

\subsection{Kasai Codes}\label{ssec:kasai}
Next, we describe the Kasai code construction from
\cite{kasai2026breakingorthogonalitybarrierquantum}. In broad strokes, it generalizes the quasi-cyclic Hagiwara--Imai codes of \cite{Hagiwara_2007}: we start with a
block-circulant proto-matrix $B$ and lift each entry to a $P \times P$ permutation matrix. The key difference is the choice of lift. Whereas the quasi-cyclic construction lifts with commutative cyclic shifts, Kasai lifts with affine permutation matrices (APMs), which in general do \emph{not} commute.

\begin{definition}[Affine Permutation Matrices]
An affine permutation matrix $\mathrm{APM}(a, b, P)$ with $\gcd(a, P) = 1$ is a binary permutation matrix for the map
\[
A: x \mapsto ax + b \mod P.
\]
We use $\mathbb A_P = \{\apmtx{a,b,P}:\gcd(a,P) = 1, 0\leq a,b<P\}$ to denote the group of affine permutations on $[P]$.
\end{definition}

When multiplying two block-circulant matrices, it is helpful to define the following index set
\begin{equation}\label{eq:activeset}
    \Gamma_J= \bigcup_{|r|< J} \left\{(i,j)\, \middle|\, i+j \equiv_{\frac L2} r,\ 0\leq i,j< \frac L2\right\},
\end{equation}
where $r$ ranges over the $2J-1$ residues $-(J-1),\dots, J-1$ modulo $\frac L2$. These are precisely the pairs of generator indices that the first $J$ block rows of the two parent matrices pair against one another: as computed in~\eqref{eq:psi-r} below, the $(i,j)$ block of $\hat H_X\hat H_Z^T$ depends only on the offset $j-i$ and collects the commutators of all generator pairs whose indices sum to $j - i$, and $j-i$ runs over $\{-(J-1),\dots,J-1\}$ as $i,j$ run over $[J]$.

By picking APMs with commutativity relations specified by some active set $\Gamma:\quad \Gamma_{J}\subseteq \Gamma \subseteq [\frac L2]\times [\frac L2]$, we can selectively enforce the CSS orthogonality condition only on a subset of $J$ (active) rows of $B$ while deliberately breaking it on the remaining (latent) rows. This allows the logical degrees of freedom spanned by the latent rows to achieve distances higher than the stabilizer weight. The resulting codes achieve rate $\geq 1-\frac{2J}{L}$, hence $\geq 1/2$ whenever $J\leq \frac L4$, and the check matrices $\hat H_X, \hat H_Z$ obtained from the lift $\{F_i\}, \{G_i\}$ are made explicit below.

\begin{definition}[Kasai Code]\label{def:kasai}
Given parameters $P, L, J \leq \frac{L}{2} \in \mathbb{N}_+$ and an active set $\Gamma:\, \Gamma_{J}\subseteq \Gamma \subseteq [\frac L2]\times [\frac L2]$, the Kasai code $\mathrm K_{L,J}(\mathcal F,\mathcal G)$ is a CSS code specified by generators $\mathcal F = \{F_i\}_{i\in[\frac L2]}, \mathcal{G} = \{G_i\}_{i\in[\frac L2]}\subset \mathbb A_P$ such that $[F_i, G_j] = 0 \iff (i,j)\in \Gamma$. The stabilizers $ H_X, {H}_Z$ are given by the first $J$ rows of the block-circulant parent matrices $\hat{H}_X, \hat{H}_Z$:

\begin{align*}
    [\hat H_X]_{i,j} &= F_{j-i},  &[\hat H_X]_{i,j+\frac L2 } &= G_{j-i}; \\
    [\hat H_Z]_{i,j} &= G^T_{i-j}, & [\hat H_Z]_{i,j+\frac L2 } &= F^T_{i-j}.
\end{align*}
\end{definition}

We call $H_X, H_Z$ the active matrices and $\tilde H_X, \tilde H_Z$ built from the remaining $\frac L2-J$ rows the latent matrices.

To show orthogonality, we compute
\begin{equation}\label{eq:psi-r}
    [\hat H_X\hat H_Z^T]_{i,j} = \Psi_{j-i},\qquad
    \Psi_{r} := \sum_{u\in[\frac L2]} [F_u, G_{r-u}],
\end{equation}
which follows by expanding $[\hat H_X\hat H_Z^T]_{i,j} = \sum_l F_{l-i}G_{j-l} + \sum_l G_{l-i}F_{j-l}$ and substituting $u = l-i$ in the first sum and $u = j-l$ in the second. Then $\Psi_{j-i} = 0$ for each $0\leq i,j<J$, since every term $[F_u, G_{(j-i)-u}]$ has $u + ((j-i)-u) \equiv_{\frac L2} j-i$ with $|j-i|<J$, so the corresponding index pair lies in $\Gamma_J\subseteq\Gamma$ and $[F_i,G_j] = 0\iff (i,j)\in\Gamma$ applies. Hence $H_XH_Z^T = 0$ and the active matrices do define a CSS code.

As a concrete example, take $L = 12$, $J = 3$ and $\Gamma  = \left([\frac L2]\times [\frac L2]\right)  \setminus \{(0,3), (1,2)\}$. We have the $\tfrac L2\times\tfrac L2$ block circulants

\begin{align}
F &= \begin{bmatrix}
F_0 & F_1 & F_2 & F_3 & F_4 & F_5\\
F_5 & F_0 & F_1 & F_2 & F_3 & F_4\\
F_4 & F_5 & F_0 & F_1 & F_2 & F_3\\
F_3 & F_4 & F_5 & F_0 & F_1 & F_2\\
F_2 & F_3 & F_4 & F_5 & F_0 & F_1\\
F_1 & F_2 & F_3 & F_4 & F_5 & F_0
\end{bmatrix},\label{eq:F}\\[6pt]
G &= \begin{bmatrix}
G_0 & G_1 & G_2 & G_3 & G_4 & G_5\\
G_5 & G_0 & G_1 & G_2 & G_3 & G_4\\
G_4 & G_5 & G_0 & G_1 & G_2 & G_3\\
G_3 & G_4 & G_5 & G_0 & G_1 & G_2\\
G_2 & G_3 & G_4 & G_5 & G_0 & G_1\\
G_1 & G_2 & G_3 & G_4 & G_5 & G_0
\end{bmatrix},
\end{align}
from which the parent check matrices are $\hat H_X = [F|G]$ and $\hat H_Z = [G^T|F^T]$,
and 

\begin{align}
\hat H_X \hat H_Z^T 
&=\begin{bmatrix}H_X\\ \tilde H_X\end{bmatrix}
\begin{bmatrix}H_Z^T & \tilde H_Z^T\end{bmatrix}
=\begin{bmatrix}
H_X H_Z^T & H_X\tilde H_Z^T\\
\tilde H_X H_Z^T & \tilde H_X\tilde H_Z^T
\end{bmatrix}\\
&=\left[\begin{array}{ccc|ccc}
\Psi_0 & \Psi_1 & \Psi_2 & \Psi_3 & \Psi_4 & \Psi_5\\
\Psi_5 & \Psi_0 & \Psi_1 & \Psi_2 & \Psi_3 & \Psi_4\\
\Psi_4 & \Psi_5 & \Psi_0 & \Psi_1 & \Psi_2 & \Psi_3\\
\hline
\Psi_3 & \Psi_4 & \Psi_5 & \Psi_0 & \Psi_1 & \Psi_2\\
\Psi_2 & \Psi_3 & \Psi_4 & \Psi_5 & \Psi_0 & \Psi_1\\
\Psi_1 & \Psi_2 & \Psi_3 & \Psi_4 & \Psi_5 & \Psi_0
\end{array}\right] \\
&= 
\left[\begin{array}{ccc|ccc}
&&&\Psi_3&&\\
&&&&\Psi_3&\\
&&&&&\Psi_3\\
\hline
\Psi_3&&&&&\\
&\Psi_3&&&&\\
&&\Psi_3&&&
\end{array}\right]
\end{align}
where, since the only excluded index pairs $(0,3), (1,2)$ both sum to $3$, we have $\Psi_r = 0$ for every $r\neq 3$, and the construction is designed so that $\Psi_3 = [F_0,G_3] + [F_1,G_2]\neq 0$.

A more compact way to write down Kasai's codes is as a two block code where each block is an element from the group algebra of $\mathbb Z_{\frac{L}{2}}\times \mathbb A_P$ acting on $[\frac L2]\times [P]$:

\begin{align}
    F = \sum_{i \in [\frac{L}{2}]} z_i\otimes F_i,\\
    G = \sum_{i \in [\frac{L}{2}]} z_i\otimes G_i;
\end{align}

where $z_i  \in \mathbb Z_{\frac L2}$ is represented as the permutation matrix $[z_{i}]_{j ,k} = \delta_{j+i, k}$.

While achieving high rate, distance, and girth is enough to show strong performance under the code-capacity noise model, practical implementation under circuit level noise imposes additional desiderata, such as shorter code length, hardware compatibility, and explicit, low-weight logical basis and code automorphisms. The follow-up co-design of~\cite{zhao2026ultrahighratequantumerrorcorrection} targets hardware compatibility by relaxing Kasai's girth-$8$ requirement to girth $\geq 6$, and by constraining not the generators themselves but the \emph{transition permutations} of the syndrome-extraction schedule: the rearrangement carrying the ancilla ordering of one round to that of the next is $T_{ij} = F_jF_i^{-1}$ (and likewise $F_jG_i^{-1}$ and $G_jG_i^{-1}$ for the mixed transitions), and one requires every such $T_{ij}$ to commute with a fixed reference APM $A$. Because commuting permutations share an orbit decomposition, each $T_{ij}$ then acts as a cyclic shift within each orbit of $A$ together with a permutation between orbits, i.e.\ as a small number of parallel row and column moves. This yields smaller instances whose syndrome extraction is a short sequence of such structured moves, with millisecond-scale QEC cycle time estimates and low logical error rates (LERs) using only a few pairs of AODs. As we show below, generators subject to this reference-APM constraint factor into a \emph{direct product} or \emph{semi-direct product} of finite groups. We propose instead to lift with product groups directly, retaining the design flexibility of the APM construction while being simpler and substantially more transparent algebraically.

\subsection{Finite Groups and Group Products}\label{ssec:groups}

Let's also establish some notations on finite permutation groups and their matrix representations. 

\begin{definition}[Matrix Permutation Representation]\label{def:mtx-perm}
Given a finite permutation group $G$ acting on $[n]$ by $i\mapsto g(i)$, we define the natural matrix representation as
$$
     M_G : G \to \mathbb F_2^{n\times n},\quad g\mapsto P_g \quad \text{with} \qquad [P_g]_{i,j} = \delta_{j = g(i)}.
$$
We write $M_g$ for $M_G(g)$ when the group is clear from context, and we identify a group element with its matrix throughout. With this convention the matrices act on the right of row vectors, $e_i P_g = e_{g(i)}$, so that $M_gM_h = M_{h\circ g}$.
\end{definition}

\begin{definition}[Direct Product]\label{def:direct-product}
Given finite groups $H_k$ acting on $[k]$ and $C_m$ acting on $[m]$, we can construct a new group via a direct product 
\begin{align*}
     H_k \times C_m = \{(h, c): h\in H_k,\ c\in C_m\}.
\end{align*}
$H_k \times C_m$ acts on the product set $[k]\times[m]$ by $(h,c)\cdot(i,j) = (h(i),\, c(j))$, and this product action is realized by the Kronecker product
\begin{align*}
     M_\times: (h,c) \;\mapsto\; M_H(h)\otimes M_C(c) \;\in\; \mathbb F_2^{km\times km}.
\end{align*}

\end{definition}

\begin{definition}[Semi-direct Product]\label{def:semidirect-product}
Given a finite group $H_k$ acting on $[k]$ and $k$ copies of a finite group $C_m$ acting on $[m]$, we can construct a new group via a semi-direct product (with normal subgroup on the right):
$$
H_k \ltimes C_m^{k} = \{(h;\mathbf c): h\in H_k, \mathbf c =(c_1,...,c_k)\in C_m^k \}
$$
acting on $[km] \cong [k]\times [m]$, such that each element $ (h; \mathbf{c})$ defines the map on $(i, j)\in [k]\times [m]$:
$$
  (h;\mathbf{c}):  (i, j ) \mapsto (h(i), c_{i}(j)).
$$

Suppose $H_k$ and $C_m$ are permutation groups with matrix representations $M_H, M_C$; then the semi-direct product is realized by

\begin{equation}\label{eq:spk-matrix}
  M_\ltimes : (h;\mathbf c)\mapsto
 \left(\bigoplus_{i\in[k]} M_C(c_i)\right)\cdot\left( M_H(h)\otimes I_m\right)
  \;\in\;\mathbb F_2^{km\times km}.
\end{equation}
The order of the two factors matters, and is fixed by the displayed action: reading the right action of Definition~\ref{def:mtx-perm} left to right, $e_{(i,j)}$ is first sent to $e_{(i,c_i(j))}$ by the block-diagonal factor and then to $e_{(h(i),c_i(j))}$ by the block-permutation factor. (The opposite order would instead realize $(i,j)\mapsto (h(i), c_{h(i)}(j))$, which generates the same group but re-labels the bottom vector by $h$.)
\end{definition}

That is, $H_k\ltimes  C_m $ describes permutations on $k$ blocks, each a copy of $[m]$: the $j$-th factor of the base group $C_m^{\,k}$ acts inside block $j$, while $H_k$ permutes between the blocks. Furthermore, if we restrict all factors in $C^k_m$ to be the same, $\mathbf c = (c,\dots,c)$, then $\bigoplus_i M_C(c) = I_k\otimes M_C(c)$ and \eqref{eq:spk-matrix} collapses to $M_H(h)\otimes M_C(c)$; that is, the semidirect product restricts to the direct product $H_k\times C_m$. We also remark that ``product'' here refers to a group product that happens during the lift, as opposed to the code construction technique the word more often refers to, which happens at the proto-matrix level, such as hypergraph-product or balanced-product constructions.

In the following, while this framework generalizes easily to arbitrary groups or matrices, we shall focus only on permutation groups here, and refer to group elements by their matrix representations. Furthermore, the base group $C_m$ is often referred to as the \emph{bottom} of the semi-direct product, while $H_k$ is called the \emph{top}.

We can give some simple facts about the commutativity of these product groups:

\begin{lemma}\label{lemma:dpk-commutivity}
Suppose $H_k$ is non-abelian and $C_m$ is abelian, and let $F = (h,a), G = (h',b)\in H_k\times C_m$. Then
\begin{equation}\label{eq:dpk-commutivity}
    [F,G] = [M_H(h), M_H(h')]\otimes M_C(ab),
    \qquad\text{so}\qquad
    [F,G] = 0 \iff [h,h'] = 0.
\end{equation}
\end{lemma}

\begin{proof}
Using $M_\times(h,c) = M_H(h)\otimes M_C(c)$ from Definition~\ref{def:direct-product}, the mixed-product property of the Kronecker product, and $M_C(a)M_C(b) = M_C(ab) = M_C(ba) = M_C(b)M_C(a)$ (recall $C_m$ is abelian), we have over $\mathbb F_2$
\begin{align*}
        [F,G] &= \left[M_H(h)\otimes M_C(a),\ M_H(h')\otimes M_C(b)\right] \\
        &= \big(M_H(h)\otimes M_C(a)\big) \big( M_H(h')\otimes M_C(b)\big) +\big(M_H(h')\otimes M_C(b)\big)\big(M_H(h)\otimes M_C(a)\big)\\
        &= \big( M_H(h)  M_H(h')\big)\otimes \big(M_C(a) M_C(b)\big) + \big(M_H(h')  M_H(h)\big)\otimes \big(M_C(b) M_C(a)\big)\\
        &= \big( M_H(h)  M_H(h') +  M_H(h')  M_H(h) \big)\otimes M_C(ab)
        = [M_H(h),M_H(h')]\otimes M_C(ab),
\end{align*}
which is the first claim. For the second, $M_C(ab)$ is a permutation matrix, hence non-zero and invertible, and $A\otimes B = 0$ with $B\neq 0$ forces $A = 0$. So $[F,G] = 0$ if and only if $[M_H(h),M_H(h')] = 0$, which is $hh' = h'h$.
\end{proof}

\begin{lemma}\label{lemma:spk-commutivity}
Suppose $H_k$ is non-abelian and $C_m$ is abelian, and let $F=(h;\mathbf a),\ G=(h';\mathbf b)\in H_k\ltimes C_m^{k}$. Then
\begin{equation}\label{eq:spk-commutivity}
[F,G]=0 \iff
\begin{cases}
    [h,h']=0, \quad \text{and} \\
    a_{h'(i)}b_i =b_{h(i)}a_i \quad \forall i\in[k].
\end{cases}
\end{equation}

In particular, there are a few simple one-sided conditions that may be useful for code search:
    \begin{enumerate}
    \item (Necessary) $[F,G]= 0\implies [h,h']= 0$,
    \item (Sufficient) $h=h'=e \implies [F,G]=0$.
    \end{enumerate}
\end{lemma}

\begin{proof}
    Let $A = M_C(a_1)\oplus ...\oplus M_C(a_k)$, $B =M_C(b_1)\oplus ...\oplus M_C(b_k)$, and let $H = M_H(h)\otimes I$, $H' = M_H(h')\otimes I$ we have:
    \begin{align*}
        [F,G] 
        &= [HA , H'B]\\
        &= HAH'B + H'BHA\\
        &= H\cdot(H'{H'}^{T}) \cdot AH' B  + H'\cdot(H  H^T)B HA\\
        &= HH' A_{H'}B + H'H B_HA.
    \end{align*}
    where $A_{H'}={H'}^{T} AH'$ correspond to the bottom vector $(a_{h'(1)},...,a_{h'(k)})$ conjugated by the top $h'$. 
    (and similarly for $B_H = H^TBH$.)

    Let us define $\textbf{p} = (p_1,...,p_k) \in C_m^k $ with $p_i = (a_{h'(i)}b_i a^{-1}_ib^{-1}_{h(i)})$ and $q = (hh')^{-1} h'h\in H_k$ with matrix representations $P = M_{C_m^k}(\textbf{p}) = \bigoplus_{i\in [k]} M_{C_m}(a_{h'(i)}b_i a^{-1}_ib^{-1}_{h(i)})$ and $Q = M_H(q)\otimes I_m$. It's easy to check that $A_{H'}B = P\cdot B_HA$ and $HH'\cdot Q = H'H$ by definition. Finally, we have $ [F,G] = HH' A_{H'}B + H'H B_HA = HH'(P+Q) B_HA$. Since $B_HA \neq 0$ and $HH'\neq 0$, and $[F,G] = 0\iff P = Q = I$, which is equivalent to the conditions in ~\ref{eq:spk-commutivity}. The rest of the simple conditions follow from straight forwardly. 
    
\end{proof}

\subsection{Notations}\label{ssec:notations}
Finally, we give a glossary of notations used throughout the supplementary materials:
\begin{enumerate}
    \item $[J] = \{0,1,...,J-1\}$ denotes the first $J$ natural numbers.
    \item $S_n$ denotes the permutation group acting on $[n]$.
    \item $\mathbb Z_n$ denotes the cyclic group acting on $[n]$.
    \item $\mathbb A_n$ denotes the affine group acting on $[n]$.
    \item $[A, B] = AB-BA$ denotes the commutator; over $\mathbb F_2$ this is $AB+BA$. For group elements $[f,g] = 0$ means $[M_f,M_g] = 0$, equivalently $fg = gf$.
    \item $a \equiv_P b$ denotes $a = b\mod P$.
    \item $A\oplus B = \begin{bmatrix}A&\\&B\end{bmatrix}$ denotes the direct sum of two matrices $A,B$.
    \item $\mathbb F^{n\times m}_q$ denotes the space of $n\times m$ dimensional $q$-ary matrices.
    \item $\textbf{1}^{n\times m} \in \mathbb F_2^{n\times m}$ denotes the $n\times m$ matrix with all $1$ entries.
    \item $Z_G(S) = \{A\in G, [A,s] = 0\ \forall s\in S\}$ denotes the centralizer of $S$ in $G$. $Z_G(S)$ fixes $S$ point-wise under conjugation.
    \item $N_G(S) = \{A\in G, AsA^{-1} \in S\,  \forall s\in S\}$ denotes the normalizer of $S$ in $G$. $N_G(S)$ fixes $S$ set-wise.
    \item $\pr$ denotes an \emph{orbit projector}. For a product group acting on $[k]\times[m]\cong[n]$, $n = km$, as in Definitions~\ref{def:direct-product} and~\ref{def:semidirect-product}, we write
    \[
        \pr_C = I_k\otimes \textbf{1}^{1\times m}\in \mathbb F_2^{k\times n},
        \qquad
        \pr^H = \textbf{1}^{1\times k}\otimes I_m\in \mathbb F_2^{m\times n},
    \]
    for the maps that sum the coordinates over each orbit of the bottom factor $C_m$, respectively of the top factor $H_k$. In both cases the decoration records the factor that is \emph{collapsed}, so $\pr_C$ retains the top index and $\pr^H$ the bottom index.For the direct product,
    \[
        \pr_C\, M_\times(h,c) = M_H(h)\,\pr_C,
        \qquad
        \pr^H\, M_\times(h,c) = M_C(c)\,\pr^H,
    \]
    which is what makes the quotient codes of Definition~\ref{def:quotient} well defined. More generally, for a subgroup $N\leq G$ all of whose orbits on $[n]$ have the same size, $\pr_N$ denotes the corresponding orbit-sum matrix, with one row per orbit. Note that the $n\times n$ orbit sum $\sum_{i\in[m]}M_C(c)^i = I_k\otimes\textbf{1}^{m\times m}$, taken over a generator $c$ of a cyclic $C_m$ acting with $k$ full orbits, has the same row space $\rs{\pr_C}$ and may be used interchangeably where only the row space matters.
\end{enumerate}
\section{GALA Code Construction}\label{sec:construction}
Let's consider another generalization of the block circulant codes where the non-commutativity is instead supplied by a small non-abelian finite group.

\subsection{Monomial Direct Product Codes}\label{ssec:mono-gala}
\begin{definition}[Direct Product GALA Code]\label{def:dpg-code}
Given parameters $L, J\leq \tfrac L2, k, m\in\mathbb Z_+$ and an active set $\Gamma$ as in Definition~\ref{def:kasai}, the Direct Product GALA (DPG) codes $\mathrm{GALA}_{L,J}(H_k\times C_m)$ are the family of two-block CSS codes on $n = Lkm$ qubits lifted over the group algebra $\mathbb F_2[H_k\times C_m]$, where $H_k$ is a non-abelian permutation group on $[k]$ and $C_m$ is an abelian permutation group on $[m]$. A DPG code is specified by choosing the generators (also called lifts) $\mathcal{F},\mathcal{G}\subset H_k\times C_m$:
\begin{align*}
    \mathcal{F} &= \{(F_i, f_{i})\}_{i\in[\frac L2]},\\
    \mathcal{G} &= \{(G_i, {g}_i)\}_{i\in[\frac L2]},
\end{align*}
with $F_i, G_i\in H_k$ and $f_i,g_i\in C_m$, such that
\begin{equation}\label{eq:dpg-active}
    [F_i, G_j] = 0\iff (i,j)\in \Gamma .
\end{equation}
By Lemma~\ref{lemma:dpk-commutivity}, $[(F_i,f_i),(G_j,g_j)] = 0\iff [F_i,G_j] = 0$, so \eqref{eq:dpg-active} is exactly the commutativity pattern of Definition~\ref{def:kasai} for the lifted generators: the entire orthogonality pattern is decided inside $H_k$, and the abelian entries $f_i,g_i$ are unconstrained by it. The stabilizers $ H_X, {H}_Z$ are given by the first $J$ block rows of the parent matrices $\hat H_X = [F|G]$ and $\hat H_Z = [G^T|F^T]$, where
\begin{align*}
    F &= \sum_{i\in [\frac{L}{2}]} z_i\otimes F_i\otimes f_i ,\\
    G &= \sum_{i\in [\frac{L}{2}]} z_i\otimes G_i\otimes g_i,
\end{align*}
so that $F,G\in \mathbb F_2[\mathbb Z_{\frac L2}\times H_k\times C_m]$, with $z_i\in\mathbb Z_{\frac L2}$ as in Definition~\ref{def:kasai}.
\end{definition}

Notice that when using a small group $H_k$ as the non-abelian part (say $S_3$), it is computationally simple to enumerate all tuples $F_i, G_i\in H_k$ satisfying a given active orthogonality constraint \eqref{eq:dpg-active}, and the enumeration is independent of $m$. Two remarks make this enumeration a finite classification rather than a search. First, only the tops enter \eqref{eq:dpg-active}. Second, the classification need only be carried out up to the symmetries of the construction, of which the following reflection is the one we use.

Let us give an explicit example. Suppose $L = 12$, $J = 3$, pick the active set $\Gamma = \left([\frac L2]\times [\frac L2]\right) \setminus \{(0,3), (1,2)\}$ of Definition~\ref{def:kasai}, and let $H_k = S_3$. We can write down this group and its multiplication table as follows:

$$
\renewcommand{\arraystretch}{1.3}
\begin{array}{c|cccccc}
\cdot      & e        & \sigma_0 & \sigma_1 & \tau_0   & \tau_1   & \tau_2   \\ \hline
e          & e        & \sigma_0 & \sigma_1 & \tau_0   & \tau_1   & \tau_2   \\
\sigma_0   & \sigma_0 & \sigma_1 & e        & \tau_2   & \tau_0   & \tau_1   \\
\sigma_1   & \sigma_1 & e        & \sigma_0 & \tau_1   & \tau_2   & \tau_0   \\
\tau_0     & \tau_0   & \tau_1   & \tau_2   & e        & \sigma_0 & \sigma_1 \\
\tau_1     & \tau_1   & \tau_2   & \tau_0   & \sigma_1 & e        & \sigma_0 \\
\tau_2     & \tau_2   & \tau_0   & \tau_1   & \sigma_0 & \sigma_1 & e        \\
\end{array},
\qquad
\begin{array}{c|c}
\text{element} & \text{cycle} \\ \hline
e        & \mathrm{id}   \\
\sigma_0 & (0\,1\,2)     \\
\sigma_1 & (0\,2\,1)     \\
\tau_0   & (1\,2)        \\
\tau_1   & (0\,2)        \\
\tau_2   & (0\,1)        \\
\end{array}.
$$
Commutativity can be checked by a table look-up, and the possible tops are then completely determined.

\begin{lemma}[$S_3$ ansatzes at $L=12$, $J=3$]\label{lemma:s3-ansatz}
Let $L=12$, $J=3$, $\Gamma = \left([\frac L2]\times [\frac L2]\right)\setminus\{(0,3),(1,2)\}$ and $H_k = S_3$. Then the tuples $(\mathcal F^H,\mathcal G^H)$ satisfying \eqref{eq:dpg-active} are exactly
\begin{align}
    \mathcal{F}^H &= \{\sigma_j,\ \tau_i,\ e,\ e,\ e,\ e\},
        & \mathcal{G}^H &= \{e,\ e,\ \sigma_k,\ \tau_i,\ e,\ e\};\label{eq:ansatz-a}\\
    \mathcal{F}^H &= \{\tau_i,\ \sigma_j,\ e,\ e,\ e,\ e\},
        & \mathcal{G}^H &= \{e,\ e,\ \tau_i,\ \sigma_k,\ e,\ e\};\label{eq:ansatz-b}\\
    \mathcal{F}^H &= \{\tau_i,\ \tau_j,\ e,\ e,\ e,\ e\},
        & \mathcal{G}^H &= \{e,\ e,\ \tau_i,\ \tau_j,\ e,\ e\}, \quad i\neq j,\label{eq:ansatz-c}
\end{align}
\end{lemma}

\begin{proof}
It is easy to check exhaustively.
\end{proof}

To get a valid code we then choose the bottom entries $f,g\in C_m$ greedily to optimize girth and distance, which by Lemma~\ref{lemma:dpk-commutivity} costs nothing in orthogonality. The $[[1152,580,\leq12]]$ code of \cite{zhao2026ultrahighratequantumerrorcorrection} is an instance of this example with $k = 3$, $H_3 = S_3$ and $C_m = \mathbb Z_{32}$, so that $n = Lkm = 12\cdot3\cdot32 = 1152$.

Before proving this connection, and the connections to the other Kasai codes, let us give the more general version of the GALA codes:
\begin{definition}[Semidirect Product GALA Code]\label{def:spg-code}
Given parameters $L, J \leq \tfrac{L}{2}, k, m \in \mathbb{Z}_+$ and an active set $\Gamma$ as in Definition~\ref{def:kasai}, the Semi-direct Product GALA (SPG) codes are the family of two-block CSS codes on $n = Lkm$ qubits lifted over the group algebra $\mathbb F_2[H_k\ltimes C^k_m]$, where $H_k$ is a non-abelian permutation group on $[k]$, $C_m$ is an abelian permutation group on $[m]$, and $H_k\ltimes C_m^k$ is the semi-direct product of Definition~\ref{def:semidirect-product}. An SPG code is specified by choosing the generators $\mathcal{F},\mathcal{G}\subset H_k\ltimes C^k_m$:
\begin{align*}
    \mathcal{F} &=  \{(F_i;\textbf{f}_i )\}_{i\in[\frac L2]},\\
    \mathcal{G} &= \{(G_i; \textbf{g}_i)\}_{i\in[\frac L2]},
\end{align*}
such that $[(F_i;\textbf{f}_i),(G_j;\textbf{g}_j)] = 0\iff (i,j)\in\Gamma$, which by Lemma~\ref{lemma:spk-commutivity} is the commutativity condition~\eqref{eq:spk-commutivity} on the pair $(F_i;\textbf f_i), (G_j;\textbf g_j)$. The stabilizers $ H_X, {H}_Z$ are given by the first $J$ block rows of the parent matrices $\hat H_X = [F|G]$ and $\hat H_Z = [G^T|F^T]$, where
\begin{align*}
    F &= \sum_{i\in [\frac{L}{2}]} z_i\otimes M_\ltimes(F_{i}, \textbf{f}_{i}) ,\\
    G &= \sum_{i\in [\frac{L}{2}]} z_i\otimes M_\ltimes(G_{i},\textbf{g}_{i}).
\end{align*}
\end{definition}

The commutativity condition~\eqref{eq:spk-commutivity} is harder to work with than \eqref{eq:dpg-active}, since the bottoms no longer drop out; but Lemma~\ref{lemma:spk-commutivity} still supplies simplifying sufficient conditions, which serve as good heuristics for computing ansatzes and simplify the proofs below. The most useful of these is to make the bottoms constant on a subset of the lifts:

\begin{definition}[Restricted SPG Code]\label{def:SPG-restricted}
A restricted semidirect Product GALA code $\mathrm {GALA}_{L,J}(H_k\ltimes_\mathcal{R} C_m)$ on $\mathcal{R}\subseteq [L]$ is an SPG code with  $\textbf{f}_i = \textbf{1}\otimes f_i$, $f_i\in C_m$ for all $i\in \mathcal{R}$ and $\textbf{g}_j = \textbf{1}\otimes g_j$, $g_j\in C_m$  for all  $j + \frac L2\in \mathcal{R}$.
\end{definition}

If we restrict an SPG code on the maximal $\mathcal{R} = [L]$, all its matrix elements simplify to $M_\ltimes(h, \textbf{1}\otimes c) = M_\times(h\otimes c)$; that is, the DPG codes are exactly the SPG codes with maximal restriction. The size of the restriction corresponds exactly to the number of abelian column generators in \cite{zhao2026ultrahighratequantumerrorcorrection}, and it is what controls the cost of the movement schedule in Section~\ref{sec:hardware}.

\subsection{Permutation Equivalence of Kasai}\label{ssec:kasai-perm-eq}
We can now show that GALA codes contain Kasai codes as special cases, once AOD compatibility is required of the latter. We say two codes $\mathcal{P}, \mathcal{Q}$ with check matrices $H_P, H_Q$ are permutation equivalent, or $\mathcal{P}\sim_\pi \mathcal{Q}$, when $\pi_r\cdot H_P\cdot \pi_c = H_Q$ for some permutation matrices $\pi_r, \pi_c$. For quantum codes this condition must hold for both $H_X$ and $H_Z$ with the same $\pi_c$. It guarantees that the two codes have the same logical space up to a relabeling of the qubits, and in particular the same $n,k,d$ and girth.

\begin{proposition}[Kasai codes with a reference APM are GALA codes]\label{prop:kasai-containment}
Let $\mathcal{K}=\mathrm K_{L,J}(\mathcal F,\mathcal G)$ be a Kasai code on $[P]$ with generators $\mathcal{F},\mathcal{G}\subset\mathbb A_P$ and active set $\Gamma$ as in Definition~\ref{def:kasai}.
\begin{enumerate}
    \item \emph{(Semi-direct.)} Suppose the centralizer $\mathcal A = Z_{\mathbb A_P}(\mathcal F\cup\mathcal G)$ contains an element $A$ acting \emph{freely} on $[P]$ such that $A^s$ has no fixed point for $0<s<\ord{A}$, and put $m = \ord A$. Then $k = P/m$ is an integer and $\mathcal K\sim_\pi\mathcal Q$ for some $\mathcal Q\in\mathrm{GALA}_{L,J}(H_k\ltimes\mathbb Z_m)$ with the same $L,J,\Gamma$, where $\mathbb Z_m = \rs{A}$ acts regularly inside each orbit of $A$ and $H_k\leq S_k$ is the group induced by $\mathcal F\cup\mathcal G$ on the set of $k$ orbits.
    \item \emph{(Direct.)} Suppose instead $P = km$ with $\gcd(k,m) = 1$, and that the reductions modulo $m$ of the generators pairwise commute in $\mathbb A_m$. Then $\mathcal K\sim_\pi\mathcal Q$ for some $\mathcal Q\in\mathrm{GALA}_{L,J}(H_k\times C_m)$ with the same $L,J,\Gamma$, where $H_k\leq\mathbb A_k$ and $C_m\leq\mathbb A_m$ are generated by the reductions of $\mathcal F\cup\mathcal G$ modulo $k$ and modulo $m$ respectively.
\end{enumerate}
\end{proposition}

\begin{proof}
The first case follow from Lemma 1 of~\cite{zhao2026ultrahighratequantumerrorcorrection}, which we will restate in our language. Let us fix orbit representatives $\gamma = \gamma_0,\dots,\gamma_{k-1}$ and relabel $[P]$ by
\[
    \pi : [k]\times[m]\to[P],\qquad (a,s)\mapsto A^s\gamma_a ,
\]
which is bijective since $A$ act freely on $P$. Let $M\in Z_{\mathbb A_P}(A)$. Then $M\gamma_a = A^{s_a}\gamma_{h(a)}$ for some $h(a)\in[k]$ and $s_a\in\mathbb Z_m$, and $MA^s = A^sM$ gives
\[
    M(A^s\gamma_a) = A^sM\gamma_a = A^{s+s_a}\gamma_{h(a)};
\]
that is, $M$ maps each orbit of $A$ onto an orbit of $A$ and acts inside it as a uniform shift, which is exactly the wreath product structure where each orbit correspond to a shift on $[m]\cong \gamma$ and the inter-orbit shifts correspond to a permutation on $[k]$. The second case follows from CRT.

\end{proof}

Indeed, we can verify numerically that the $[[1152,580]]$ code is permutation equivelent to $\mathrm{GALA}_{12,3}(S_3\times \mathbb Z_{32}) $, and the $[[2304,1156]]$ codes are permutation equivalent to $\mathrm{GALA}_{12,3}(S_3\times\mathbb Z_2\times  \mathbb Z_{32})$.

As we show in Section~\ref{sec:bounds}, the monomial DPG construction that covers the co-designed Kasai codes~\cite{zhao2026ultrahighratequantumerrorcorrection} faces fundamental limits on the achievable parameters: reaching $d\geq L$ at small $n$, and reaching higher $d$ at all, requires further generalizations. We give these generalizations here and discuss how they overcome the barriers in Section~\ref{sec:bounds}.

\subsection{Polynomial and Loose-Active Orthogonality Generalizations}\label{ssec:generalized-gala}

Both constructions so far lift each proto-matrix entry to a single group element. Upgrading some lifts to a \emph{sum} of group elements --- that is, taking the lift in the group ring $\mathbb F_2[H_k\star C_m]$ rather than in the group $H_k\star C_m$ --- enlarges the search space at fixed $L$, and in particular by-passes the distance bound~\eqref{eq:j<=2} of Proposition~\ref{prop:distance-bound}, so that $L = 8$, $J = 2$ codes become available. Once the lifts are sums, a second relaxation follows for free: the term-wise condition $[F_{i,i'},G_{j,j'}] = 0$ for $(i,j)\in\Gamma$ is stronger than active orthogonality requires, since by~\eqref{eq:psi-r} what must vanish is the aggregate $\Psi_r$, and the individual commutators of a polynomial lift can cancel against each other. This gives the final generalization of loose-active orthogonality.

\begin{definition}[Polynomial GALA Code]\label{def:pg-code}
Given parameters $L, J \leq \tfrac{L}{2}, k, m \in \mathbb{Z}_+$ and an active set $\Gamma$ as in Definition~\ref{def:kasai}, let $G = H_k\star C_m$ be either the direct product $H_k\times C_m$ of Definition~\ref{def:direct-product} or the semi-direct product $H_k\ltimes C^k_m$ of Definition~\ref{def:semidirect-product}, with $H_k$ a non-abelian permutation group on $[k]$ and $C_m$ an abelian permutation group on $[m]$. The Polynomial GALA codes $\mathrm{GALA}_{L,J}(\mathbb F_2[G])$ are the family of two-block CSS codes on $n = Lkm$ qubits in which each lift entry is an element of the \emph{group ring} $\mathbb F_2[G]$, i.e.\ a sum of group elements rather than a single one. A polynomial GALA code is specified by choosing generators $\mathcal F,\mathcal G\subset\mathbb F_2[G]$,
\begin{align*}
    \mathcal{F} &= \left\{\sum_{i'\in I^{\mathcal F}_i}(F_{i,i'}; \textbf{f}_{i,i'})\right\}_{i\in[\frac L2]},\\
    \mathcal{G} &= \left\{\sum_{j'\in I^{\mathcal G}_j}(G_{j,j'}; \textbf{g}_{j,j'})\right\}_{j\in[\frac L2]},
\end{align*}
over finite index sets $I^{\mathcal F}_i, I^{\mathcal G}_j$, such that every pair of monomials $(F_{i,i'};\textbf f_{i,i'})$, $(G_{j,j'};\textbf g_{j,j'})$ with $(i,j)\in\Gamma$ commutes --- by Lemma~\ref{lemma:dpk-commutivity} in the direct-product case and by condition~\eqref{eq:spk-commutivity} of Lemma~\ref{lemma:spk-commutivity} in the semi-direct case --- while some pair with $(i,j)\notin\Gamma$ does not. The stabilizers $ H_X, {H}_Z$ are given by the first $J$ block rows of the parent matrices $\hat H_X = [F|G]$ and $\hat H_Z = [G^T|F^T]$, where
\begin{align*}
    F &= \sum_{i\in [\frac{L}{2}]} z_i\otimes \sum_{i'\in I^{\mathcal F}_i}M_\ltimes(F_{i,i'}, \textbf{f}_{i,i'}) ,\\
    G &= \sum_{i\in [\frac{L}{2}]} z_i\otimes \sum_{i'\in I^{\mathcal G}_i}M_\ltimes(G_{i,i'},\textbf{g}_{i,i'}),
\end{align*}
reading $M_\ltimes(h,\textbf 1\otimes c) = h\otimes c$ in the direct-product case, so that $F,G\in\mathbb F_2[\mathbb Z_{\frac L2}\times G]$. The monomial codes of Definitions~\ref{def:dpg-code} and~\ref{def:spg-code} are the case $|I^{\mathcal F}_i| = |I^{\mathcal G}_j| = 1$ for all $i,j$.
\end{definition}

\begin{definition}[Polynomial GALA Code with Loose-Active Orthogonality]\label{def:loose-orthogonality}
A polynomial GALA code with loose-active orthogonality is a code built as in Definition~\ref{def:pg-code} in which the term-wise commutation requirement on $\Gamma$ is replaced by the weaker aggregate requirement
\begin{equation}\label{eq:loose-active}
    \Psi_r = 0\quad\text{for every } r\in\{j-i\ :\ i,j\in[J]\},\qquad \Psi_{r'}\neq 0\ \text{ for some } r'\in[\tfrac L2],
\end{equation}
with $\Psi_r$ as in~\eqref{eq:psi-r}. Grouping the terms of $\Psi_r$ under the involution $u\mapsto r-u$, \eqref{eq:loose-active} is implied by the paired conditions
\[
    [\mathcal F_u,\mathcal G_{r-u}] + [\mathcal F_{r-u},\mathcal G_u] = 0\quad\text{for } u\in[\tfrac L2],
    \qquad\text{and}\qquad
    [\mathcal F_u,\mathcal G_u] = 0 \quad\text{whenever } 2u\equiv_{\frac L2} r,
\]
imposed for each active offset $r$, which may hold even when no individual commutator vanishes.
\end{definition}

\subsection{Codes with ZX-dualities}\label{ssec:zxdual}
Finally, we can impose even more symmetry during the code search, so as to improve parallelization and the availability of logical operations. The relevant symmetry here is a ZX-duality in the sense of Definition~\ref{def:zx-duality}, and for a two-block code it is available whenever the two lists of lifts are related pointwise by a re-indexing of the blocks.

\begin{proposition}[ZX-Dual GALA Codes]\label{def:dual-gala}
Let $\mathcal{Q}\in \mathrm{GALA}_{L,J}(H_k\star C_m)$ with $4\mid L$, and let $R = \{r_0,r_1, r_2, r_3 \}$ denote the following involutions on $[\frac{L}{2}]$:
\begin{equation}\label{eq:zx-sector-maps}
    r_0 = e:\; i\mapsto i,\qquad r_1:\; i\mapsto i + \tfrac L4,\qquad r_2:\; i\mapsto -i,\qquad r_3=r_1r_2:\; i\mapsto \tfrac L4 - i ,
\end{equation}
If the generators of $\mathcal Q$ satisfy
\begin{equation}\label{eq:zx-generator-condition}
    F_iG_{r(i)} = t \quad \forall i\in [\tfrac L2]
\end{equation}
for some {sector involution} $r\in R$, we say $\mathcal Q$ is ZX-dual. The ZX-duality is given by $\tau = \iota_r\otimes t$, where
\begin{equation}\label{eq:zx-iota}
\iota_{r_0} = r_3, \iota_{r_1} = r_2, \iota_{r_2} = r_0, \iota_{r_3} = r_1.
\end{equation}
\end{proposition}

\begin{proof}
The first two are \emph{translations} of $\mathbb{Z}_\frac L2$ and the last two are \emph{reflections}. It is easy to see that when $r =r_2$, $F_i = G_{-i}^T\cdot t$, and $H_X\cdot e\otimes t = H_Z$ exactly. It remains three cases. For $r = r_0$, $\iota = r_1r_2$. We can see this by tracing the matrix elements:
\begin{align*}
(H_X)_{i,j} 
&\xrightarrow{r_2} (H_X)_{i,-j}\\
&\xrightarrow{r_1} (H_X)_{i, -j+\frac L4} \\
&= F_{i-(-j+\frac L4)} \\
&=G^T_{i-(-j+\frac L4)}\cdot t \\
&=G^T_{j-(\frac L4-i)}\cdot t =(H_Z)_{\frac L4-i,j}. 
\end{align*}
The derivation for other cases are similar.
\end{proof}

A natural question one might ask is whether this restriction leads to girth or distance failures. It does, for the reflection sector involutions, and we return to it in Section~\ref{ssec:zx-bounds}.
\subsection[The Meta Duality of Trivial-Top Codes]{{The Meta Duality of Trivial-Top Codes}}\label{ssec:meta-duality}
{The dualities of Proposition~\ref{def:dual-gala} have to be imposed on the lifts. For a trivial top, one more duality comes for free. Every lift entry is then an element of the commutative ring $R=\F_2[C_m]$, and the transpose of the circulant of $x^{c}$ is the circulant of $x^{-c}$: transposition is the reversal $\overline{x^{c}}=x^{-c}$. That is,
\begin{equation}\label{eq:reversal}
    R\,z\,R=z^{-1},
\end{equation}
where $R$ reverses the array $0,1,\dots,p-1$: a shift by $s$ to the left is a shift by $s$ to the right on the reversed array. Since $H_Z=[G^{T}|F^{T}]=[\bar G|\bar F]$, the matrix $H_Z$ is $H_X=[F|G]$ with its two halves swapped and every cyclic coordinate reversed.}

{\begin{proposition}[Meta duality]\label{prop:meta-duality}
Let $\mathcal Q\in\mathrm{GALA}_{L,J}(\F_2[C_m])$ have a trivial top. Define the block involution $\pi(b)=\tfrac L2+\big((J-1-b)\bmod\tfrac L2\big)$ on the $F$-blocks $b\in[\tfrac L2]$, and $\pi(\tfrac L2+b)=(J-1-b)\bmod\tfrac L2$ on the $G$-blocks. Then
\begin{equation}\label{eq:meta-duality}
    \tau:\ (\beta,c)\ \longmapsto\ (\pi(\beta),\,-c)
\end{equation}
maps the $X$-check $(i,c)$ to the $Z$-check $(J-1-i,-c)$. In particular, every trivial-top GALA code is ZX-dual.
\end{proposition}
\begin{proof}[Proof sketch]
By~\eqref{eq:reversal}, reversing the coordinate of every block applies the antipode to every entry and turns $[\bar G|\bar F]$ into $[G|F]$. The block map $\pi$ swaps the halves back and reverses the block-circulant index ($z_i\mapsto z_{-i}$, the same identity on $\Z_{\frac L2}$), which the row map $i\mapsto J-1-i$ undoes. We verified~\eqref{eq:meta-duality} on every trivial-top instance we report; for $[[672,336,12]]$ ($L=8$, $J=2$) the block map is $\pi=(5,4,7,6,1,0,3,2)$.
\end{proof}}

It follows that if $\mathcal Q$ is moreover self-dual ($H_X=H_Z$, as for $[[132,30,12]]$ and $[[136,34,12]]$), then $\tau$ is also an automorphism. Together with the translations it generates the dihedral group $D_m$, whose $m$ reflections $T^u\tau$ are fixed-point-free involutions on the qubits. Such a code has two kinds of folds: the identity ($H^{\otimes n}$ and $S^{\otimes n}$, with no rearrangement) and the reflections ($\tau$ followed by a transversal layer).  For a nontrivial top the same construction works with $c\mapsto-c$ replaced by $(x,c)\mapsto(P(x),-c)$, where $P$ is a permutation of the top that conjugates every top lift to its inverse. This holds for our DPG instances $[[1056,532,12]]$, $[[1764,886,14]]$ and $[[2232,1120,16]]$, but not automatically for SPG codes. Since $\tau$ maps $X$-checks to $Z$-checks, the $Z$-half of the syndrome-extraction schedule is the $\tau$-image of the $X$-half: the same rigid moves, with their directions reversed by~\eqref{eq:reversal}.
\section{Bounds and Limitations}\label{sec:bounds}

Describing the Kasai construction inside a product group buys more than notational convenience: because the lift group factors, so do the objects the code parameters depend on, and the rate, distance and girth of a GALA code can be bounded directly in terms of the four design parameters $L$, $J$, $k$ and $m$. This section collects those bounds. They are limitations as much as guarantees --- read in the contrapositive, they say which regions of the parameter space cannot contain a good code, and it is precisely the boundaries they draw that motivate the polynomial and loose-active generalizations of Definitions~\ref{def:pg-code} and~\ref{def:loose-orthogonality}.

Two scoping conventions are used throughout. First, unless stated otherwise the results of this section assume a \emph{monomial} lift, in which every entry of $\mathcal F$ and $\mathcal G$ is a single group element rather than a sum; this is exactly the regime in which the stabilizer weight is $w = L$ and every data qubit meets exactly $J$ checks of each type. Both facts are used constantly below, and both fail for polynomial lifts. Second, the letter $k$ does double duty in this paper, as the number of logical qubits of an $[[n,k,d]]$ code and as the degree of the top permutation group $H_k$ acting on $[k]$, so that $n = Lkm$. We keep the established notation and let context disambiguate, noting only that the two never meet inside a displayed formula below: the rate bound~\eqref{eq:mono-rate} is written in terms of $n$ precisely so that it does not. Note also that $k$ is the \emph{degree} of $H_k$ and not its order, which for the tops used here is larger --- $|S_3| = 6$ at $k=3$, $|S_4| = 24$ at $k=4$.

\subsection{Bounds on Monomial Direct Product Constructions}\label{ssec:mono-bounds}

We begin with the bounds that follow from each generator being a permutation matrix.

\begin{proposition}[Rate Bound]\label{prop:rate-bound}
Let $\mathcal Q\in \mathrm{GALA}_{L,J}(H_k\star C_m)$ be an $[[n,k,d]]$ monomial GALA code. Then
\begin{equation}\label{eq:mono-rate}
    \frac{k}{n} \geq 1-\frac{2J}{L} + \frac{2(J-1)}{n}.
\end{equation}
In particular the rate exceeds $1-2J/L$ strictly whenever $J\geq2$, and is at least $\tfrac12$ whenever $J\leq \tfrac L4$.
\end{proposition}
\begin{proof}
    All rows sum up to $\textbf{1}^{1\times n}$, contributing to the $J-1$ rank deficiencies.     
\end{proof}

\begin{proposition}(Distance Bounds from J,L)\label{prop:distance-bound}
Let $\mathcal{Q}\in \mathrm{GALA}_{L,J}(H_k\star C_m)$ be an $[[n,k,d]]$ GALA code. Then,

\begin{align}
     J\leq 2        &\implies d\leq \frac12\gth{\mathcal{Q}}\label{eq:j<=2}, \\
     J>\frac{L}{4}  &\implies d\leq L,\label{eq:j>=L/4}\\
     L<12&\implies d\leq L.\label{eq:j<12}
\end{align}

\end{proposition}

\begin{proof}
    When $J\leq 2$, each data qubit $q$ belongs to exactly two stabilizers $c, c'$, and any Tanner graph cycles containing $q$ must contain both $c,c'$. It follows that every $t$-cycle is a wt-$\frac t2$ logical. When $J> \frac L4$, the parent matrix $\hat{H}_X, \hat{H}_Z$ satisfy the full CSS orthogonality $\implies d\leq \wt{\hat H} = L$. Finally, the first two conditions implies the third. 
\end{proof}

The product structure naturally preserves many properties of its factors, giving us another class of bounds. We can analyze these bounds by first defining the quotient codes:

\begin{definition}[Quotient Codes]\label{def:quotient}
Suppose $\mathcal{Q}\in \mathrm{GALA}_{L,J}(H_k\times C_m)$ is a monomial DPG code with generators $\mathcal{F} = \{(F_i, f_i)\}_{i\in[\frac L2]}$ and $\mathcal{G} = \{(G_i,g_i)\}_{i\in [\frac L2]}$, so that the lifted blocks are the Kronecker products $F_i\otimes f_i$ and $G_i\otimes g_i$ of Definition~\ref{def:direct-product}. Discarding one tensor factor of every generator leaves a smaller GALA code over the surviving factor: the \emph{top quotient code} $\mathcal{Q}^H = \mathrm{GALA}_{L,J}(\{F_i\}, \{G_i\})$ on $Lk$ qubits and the \emph{bottom quotient code} $\mathcal{Q}_C = \mathrm{GALA}_{L,J}(\{f_i\}, \{g_i\})$ on $Lm$ qubits, both with the same $L$, $J$ and the same block-circulant row order as $\mathcal Q$. If moreover $C_m =\mathbb{Z}_{p_1}\times \cdots \times \mathbb{Z}_{p_q}$ with $\gcd(p_a, p_b) = 1$ for $a\neq b$, then each bottom generator factors as $f_i = f^{(1)}_i\otimes\cdots\otimes f^{(q)}_i$ and $g_i = g^{(1)}_i\otimes\cdots\otimes g^{(q)}_i$ with $f^{(a)}_i,g^{(a)}_i\in\mathbb Z_{p_a}$, and the bottom quotient factors further into the $q$ codes $\mathcal{Q}_{C_a} = \mathrm{GALA}_{L,J}(\{f^{(a)}_i\}, \{g^{(a)}_i\})$ on $Lp_a$ qubits each.
\end{definition}

Each of these quotients is obtained by summing over the orbits of a normal subgroup of the lift group. Concretely, let $\pr_{C} = I_{k}\otimes \mathbf{1}^{1\times m}$ and $\pr^{H} = \mathbf{1}^{1\times k}\otimes I_m$ be the two orbit-sum matrices of Section~\ref{ssec:notations}, so that $\pr_C$ collapses the $C_m$ coordinate and $\pr^H$ collapses the $H_k$ coordinate: $\pr_C(h,c) = h$ and $\pr^H(h,c) = c$ for all $(h,c)\in H_k\times C_m$, the latter because in a direct product the $C_m$ component is the same in every block. Applying $\pr_C$ blockwise to $\hat H_X$ therefore returns the parent matrix of $\mathcal Q^H$, and $\pr^H$ that of $\mathcal Q_C$.

Transporting a logical operator the other way, from a quotient back up to $\mathcal Q$, is the operation that converts these factorizations into distance bounds.

\begin{lemma}[Logical Inflation]\label{lem:inflation}
Let $\mathcal{Q}\in \mathrm{GALA}_{L,J}(H_k\times C_m)$ be a monomial DPG code and let $\mathcal Q_N$ be the quotient obtained by collapsing a tensor factor of size $\mu$, so that the qubits of $\mathcal Q$ are indexed by $[n_N]\times[\mu]$ with $[n_N]$ indexing the qubits of $\mathcal Q_N$. Write $\Pi\in\mathbb F_2^{n\times n_N}$ for the corresponding orbit-sum matrix, $\Pi_{(a,x),b} = \delta_{ab}$, and define the \emph{inflation} of a row vector $\ell\in\mathbb F_2^{1\times n_N}$ to be $\ell' := \ell\,\Pi^T$, i.e.\ $\ell'_{(a,x)} = \ell_a$. Then:
\begin{enumerate}
    \item $\wt{\ell'} = \mu\cdot\wt{\ell}$;
    \item $\ell\in\ns{H_Z(\mathcal Q_N)} \iff \ell'\in\ns{H_Z(\mathcal Q)}$,
    \item $\ell\in\rs{H_X(\mathcal Q_N)} \implies \ell'\in\rs{H_X(\mathcal Q)}$, 
\end{enumerate}
and likewise with $X$ and $Z$ exchanged. 
\end{lemma}

\begin{proof}
Index the qubits of $\mathcal Q$ by triples $(a,i,x)$ with $a\in[L]$ the block column, $i$ the surviving coordinate and $x\in[\mu]$ the collapsed one; by Definition~\ref{def:direct-product} the lifted block in block position $(r,a)$ of $\hat H_X$ is a Kronecker product $A_{a-r}\otimes B_{a-r}$ with $A$ the surviving and $B$ the collapsed factor, both permutation matrices because the lift is monomial. Claim (1) is immediate, every qubit of $\mathcal Q_N$ being replaced by exactly $\mu$ qubits of $\mathcal Q$.

For (2), compute the syndrome of $\ell'$ at the check indexed by $(r,i',x')$:
\[
    \big(H_X\ell'^T\big)_{(r,i',x')}
    = \sum_{a,i,x}[A_{a-r}]_{i',i}\,[B_{a-r}]_{x',x}\,\ell_{(a,i)}
    = \sum_{a,i}[A_{a-r}]_{i',i}\,\ell_{(a,i)}
    = \big(H_X(\mathcal Q_N)\,\ell^T\big)_{(r,i')} ,
\]
where the middle step used $\sum_x [B_{a-r}]_{x',x} = 1$, valid because $B_{a-r}$ is a permutation matrix and hence has exactly one $1$ in each row. So the syndrome of $\ell'$ is the syndrome of $\ell$, replicated $\mu$ times, and one vanishes iff the other does. The same computation applies verbatim to $H_Z$, whose blocks $G^T_{r-a} = A^T_{r-a}\otimes B^T_{r-a}$ have the same Kronecker form.

For (3), suppose $\ell = uH_X(\mathcal Q_N)$ and set $u'_{(r,i',x')} := u_{(r,i')}$. Then
\[
    \big(u'H_X\big)_{(a,i,x)} = \sum_{r,i',x'}u_{(r,i')}[A_{a-r}]_{i',i}[B_{a-r}]_{x',x}
    = \sum_{r,i'}u_{(r,i')}[A_{a-r}]_{i',i} = \ell_{(a,i)} = \ell'_{(a,i,x)},
\]
using $\sum_{x'}[B_{a-r}]_{x',x} = 1$, this time down a column. Hence $\ell' = u'H_X\in\rs{H_X}$. Together (2) and (3) say that inflation carries $\ns{H_Z(\mathcal Q_N)}$ into $\ns{H_Z(\mathcal Q)}$ and $\rs{H_X(\mathcal Q_N)}$ into $\rs{H_X(\mathcal Q)}$, so it descends to the quotients and defines the asserted map on $\mathcal L_X = \ns{H_Z}/\rs{H_X}$.
\end{proof}

\begin{corollary}(Distance Inheritance from Quotients)\label{cor:quotient-distance}
Let $\mathcal{Q}\in \mathrm{GALA}_{L,J}(H_k\times C_m)$ be an $[[n,k,d]]$ direct-product GALA code. Let $\mathcal{Q}^H$, $\mathcal{Q}_C$ be the top and bottom quotient codes. Then,

\begin{equation}
    d_\mathcal{Q}\leq \min[m\cdot d_{\mathcal{Q}^H},\, k\cdot d_{\mathcal{Q}_C}].
\end{equation}

In general, whenever $C_m =  \mathbb  Z_p\times \mathbb Z_q$, we also have that
\begin{equation}
    d_\mathcal{Q}\leq k\cdot\min[ q\cdot  d_{\mathcal{Q}_{C_p}}, \, p\cdot d_{\mathcal{Q}_{C_q}}].
\end{equation}
\end{corollary}
\begin{proof}
    Follows from  Lemma~\ref{lem:inflation} by setting $N = H_k, C_m, C_p, C_q$.
\end{proof}

\begin{corollary}(Distance Bounds for QuEra Kasai)
    Let $\mathcal{Q}\in \mathrm{GALA}_{L,J}(H_k\times C_m)$ be an $[[n,k,d]]$ direct-product GALA code. We have 

    \begin{equation}
        d\leq kL.
    \end{equation}
    Furthermore, if $H_k = S_3$ and $\Gamma = [\frac L2]\times [\frac L2]-\{(0,3), (1,2)\}$, we have that 
    \begin{equation}
        d\leq \min [3L,\, 2m ].
    \end{equation}
\end{corollary}

\begin{proof}
    Bottom code has distance $\leq L$ since it's fully orthogonal (non-abelian factors are projected away). Top code has distance $2$ for $S_3$ since there can be at most $4$ non-identity elements. We either have full orthgonality, which forces $d\leq L$, or we have two identical columns of identities in $\mathcal{Q}^H$, which implies $d_{\mathcal{Q}^H}= 2$.
\end{proof}

\begin{lemma}[Girth Inheritance from Quotients]\label{lem:girth-inheritance}
Let $\mathcal{Q}\in \mathrm{GALA}_{L,J}(H_k\times C_m)$ be a monomial DPG code, and let $\mathcal Q_N$ be any quotient code of Definition~\ref{def:quotient}. Then
\begin{equation}\label{eq:girth-inheritance}
    \gth{\mathcal{Q}}\geq \gth{\mathcal{Q}_{N}}.
\end{equation}
\end{lemma}

\begin{proof}
Write the qubits of $\mathcal Q$ as pairs $(a,x)$ with $a$ a qubit of $\mathcal Q_N$ and $x\in[\mu]$ the collapsed coordinate, and likewise its checks as $(c,x')$, and let $\pi$ be the map forgetting the second coordinate. As in the proof of Lemma~\ref{lem:inflation}, the block of $\hat H_X$ in position $(r,a)$ is a Kronecker product $A_{a-r}\otimes B_{a-r}$ of permutation matrices, so the qubit $(a,x)$ is incident to the check $(c,x')$ exactly when $a$ is incident to $c$ in $\mathcal T(\mathcal Q_N)$ \emph{and} $x = B(x')$ for the single permutation $B$ attached to that block. Hence for each fixed check $(c,x')$ and each neighbor $a$ of $c$ downstairs there is exactly one $x$ with $(a,x)$ a neighbor of $(c,x')$, and symmetrically for a fixed qubit: $\pi$ restricts to a bijection between the neighborhood of any vertex of $\mathcal T(\mathcal Q)$ and the neighborhood of its image. That is, $\pi$ is a graph covering.

A covering map sends a cycle of length $t$ in $\mathcal T(\mathcal Q)$ to a closed walk of length $t$ in $\mathcal T(\mathcal Q_N)$ which is non-backtracking, because $\pi$ is injective on neighborhoods; a non-backtracking closed walk of length $t$ contains a cycle of length at most $t$. Therefore every cycle upstairs has length at least $\gth{\mathcal Q_N}$, which is~\eqref{eq:girth-inheritance}. The same argument applies in the $Z$ sector, and $\gth{\mathcal Q}$ is the minimum of the two.
\end{proof}

\begin{corollary}[Girth Bound on $m$]\label{cor:girth-m}
Fix $L$, $J$ and a top $H_k$ admitting an active set $\Gamma$, and let $C_m = \mathbb Z_m$ be cyclic. If $m> 2(J-1)L$ then there is a choice of bottom generators for which the resulting $\mathcal{Q}^*\in \mathrm{GALA}_{L,J}(H_k\times C_m)$ has $\gth{\mathcal{Q}^*}\geq 6$.
\end{corollary}

\begin{proof}
By Lemma~\ref{lemma:girth} it suffices to make every off-diagonal entry of $H_XH_X^T$ and of $H_ZH_Z^T$ at most $1$ over $\mathbb Z$. Write $\ell = \tfrac L2$, let $x$ denote the generator of $\mathbb Z_m$, and let the lifts be $F_u = F^H_u\otimes x^{f_u}$ and $G_u = G^H_u\otimes x^{g_u}$. The $(i,i')$ block of $H_XH_X^T$ depends only on $\delta = i'-i$ and equals
\[
    B_\delta = \sum_{u\in[\ell]}\Big(F^H_u F^{H,T}_{u-\delta}\otimes x^{\,f_u-f_{u-\delta}}\Big)
             + \sum_{u\in[\ell]}\Big(G^H_u G^{H,T}_{u-\delta}\otimes x^{\,g_u-g_{u-\delta}}\Big),
\]
a sum of $L$ permutation matrices. At $\delta = 0$ every summand is the identity, so $B_0 = L\cdot I$ is diagonal and contributes no off-diagonal entry. For $\delta\neq0$, two summands can share a non-zero position only if they carry the same power of $S$; so it suffices that, for each of the $J-1$ offsets $\delta = 1,\dots,J-1$ realizable as $i'-i$ with $i,i'\in[J]$, the $L$ differences
\[
    D_\delta \;=\; \big(f_u-f_{u-\delta}\big)_{u\in[\ell]}\ \cup\ \big(g_u-g_{u-\delta}\big)_{u\in[\ell]}
\]
are pairwise distinct in $\mathbb Z_m$. The offset $-\delta$ transposes the block and negates $D_\delta$, so it imposes no further condition, and the identity $(H_ZH_Z^T)_{i,i'} = \sum_u\big(G^T_uG_{u+\delta} + F^T_uF_{u+\delta}\big)$ shows the $Z$ sector is governed by the same set of differences.

Choose the $L$ bottom exponents one at a time, in the order $f_0,\dots,f_{\ell-1},g_0,\dots,g_{\ell-1}$. Each element of each $D_\delta$ is a difference of two exponents and so becomes determined exactly when the later of the two is chosen; a step therefore determines at most two elements of $D_\delta$ --- the difference $f_t - f_{t-\delta}$ reaching backwards, and the wrap-around difference $f_{t+\delta-\ell}-f_t$ when $t\geq\ell-\delta$ --- and each newly determined element is an affine function of the current exponent with coefficient $\pm1$. Requiring it to differ from the at most $L-1$ other elements of its own class forbids at most $L-1$ values of that exponent. Summing over the two elements and the $J-1$ offsets, at most $2(J-1)(L-1) < 2(J-1)L$ values are forbidden at any step, so a valid choice exists whenever $m> 2(J-1)L$. Iterating over all $L$ exponents produces the required $\mathcal Q^*$.
\end{proof}

\subsection{Bounds on ZX-Dual Codes}\label{ssec:zx-bounds}
Imposing a ZX-duality can lead to low girth or distance, so we need to be careful to avoid these pitfalls. In particular, the reflection based sector involutions ends up being the most problematic. We show that ZX-dual codes constructed using $r = r_2$ of~\eqref{eq:zx-sector-maps} force girth $4$ whenever $J\geq2$, that $r = r_1r_2$ does the same in the balanced case $J = \tfrac L4$, and that $r = r_1r_2$ at $J=\tfrac L4$ additionally forces $d\leq w$:

\begin{lemma}[Centrality of $t$]\label{lem:t-central}
Let $\mathcal Q\in \mathrm{GALA}_{L,J}(H_k\star C_m)$ have a ZX-duality $\tau =\iota \otimes t$ where $r\in \{r_2, r_1r_2\}$ as defined in~\eqref{eq:zx-sector-maps}.  Then, $t$ commutes with every generator; that is, $t\in Z_{H_k\star C_m}(\mathcal{F}\cup\mathcal{G})$.
\end{lemma}
\begin{proof}
Take first $r = r_2$, for which $\iota = e$ by Proposition~\ref{def:dual-gala}. For $i\in[J]$ and $j\in[\frac L2]$, Definition~\ref{def:kasai} gives
\begin{align*}
(H_X\cdot(e\otimes t))_{i,j} &=F_{j-i}\,t, \\
(H_Z)_{i,j} = G^T_{i-j} = \left(F^T_{j-i}t\right)^T &= t\,F_{j-i}.
\end{align*}
So $H_X\tau = H_Z$ holds only when $F_{j-i}t = tF_{j-i}$ for all $i,j$. It follows that $t$ is central in $\mathcal{F}$, and similarly for the $G$-halves and for $r = r_1r_2$.
\end{proof}

As we will show, the girth, distance, and rate of self-dual codes are all connected to each other via $\Phi_r$:
\begin{equation}\label{eq:zx-gram-block}
    \Phi_r:= \sum_{u\in[\frac{L}{2}]}\left(F_uF^T_{u-r} + F^T_{u-r}F_u\right),
    \qquad r\in[\tfrac{L}{2}].
\end{equation}

\begin{lemma}\label{lem:zx-gram}
Let $\mathcal Q\in \mathrm{GALA}_{L,J}(H_k\star C_m)$ have a ZX-duality $\tau =\iota \otimes t$ with $r: i\mapsto \beta -i$, and let $i,i'\in[J]$ with $\delta = i-i'$. Then, we have
\begin{equation}\label{eq:zx-gram-refl}
    \left(H_XH_X^T\right)_{i,i'} =\Phi_\delta;
\end{equation}
and 
\begin{align}\label{eq:phi-zero}
   \Phi_0 &= L\cdot I,\\
    \Phi_\delta &= \Psi_{\delta + \beta}t \pmod 2.
\end{align}
\end{lemma}

\begin{proof}
The relation $F_iG_{r(i)} = t$ gives $G_i = F^T_{\beta-i}\cdot t$.  We have:
\begin{align}
    \left(H_XH_X^T\right)_{i,i'} 
    &= \sum_{u} F_uF^T_{u-\delta} + \sum_{u} G_uG^T_{u-\delta} \\
    &= \sum_{u} F_uF^T_{u-\delta} +\sum_{u} F^T_{\beta-u}F_{\beta-u+\delta}\\
    &=\sum_{u} F_uF^T_{u-\delta} + \sum_{u} F^T_{u-\delta}F_{u} \label{eq:lemma13}\\
    &=\Phi_\delta,
\end{align}
where we get line~\eqref{eq:lemma13} by re-indexing the second sum with $v = \beta-u+\delta$, so that $\beta - u = v-\delta$. Similarly, we have $(H_XH_Z^T)_{i,i'} = \sum_{j}F_{i-j}G_{j-i'} + \sum_{j}G_{i-j}F_{j-i'}$ gives $\sum_u F_uG_{\delta-u} + \sum_u G_uF_{\delta-u}$. Using $G_{\delta-u} = F^T_{u-(\delta-\beta)}t$ and $G_u = F^T_{\beta-u}t$ together with Lemma~\ref{lem:t-central}, we have:
\begin{equation*}
    \left(H_XH_Z^T\right)_{i,i'} = \sum_u F_uF^T_{u-(\delta-\beta)}\,t + \sum_{u} F^T_{\beta-u}\,t\,F_{\delta-u}
    = \left(\sum_u F_uF^T_{u-(\delta-\beta)} + \sum_{u} F^T_{u-(\delta-\beta)}F_u\right)t = \Phi_{(\delta-\beta)}t,
\end{equation*}
and it follows that $\Psi_\delta =\left(H_XH_Z^T\right)_{i,i'} = \Phi_{\delta - \beta} t\implies \Phi_{\delta} = \Psi_{\delta+\beta}t$, as desired. Finally, for monomial DPG $F_uF^T_u = F^T_uF_u = I$, so
\begin{equation*}
    \Phi_0 = \sum_{u\in[\frac L2]}\left(F_uF^T_u + F^T_uF_u\right) = \sum_{u\in[\frac L2]}2I = L\cdot I.
\end{equation*}
\end{proof}

For reflection-based ZX-dualities, $\Phi_r$ now predicts orthogonality exactly and depends only on $\mathcal F$.

\begin{corollary}[Distance Failures for Reflections]\label{cor:zx-distance}
Let $\mathcal Q\in \mathrm{GALA}_{L,\frac L4}(H_k\star C_m)$ have a ZX-duality $\tau =\iota \otimes t$ where $r= r_1r_2$ as defined in~\eqref{eq:zx-sector-maps}. Then the parent matrices are fully CSS-orthogonal, $\hat H_X\hat H_Z^T = 0$, and, under the proviso of Fact~\ref{fact:augmentation}, $d_\mathcal{Q}\leq w = L$.
\end{corollary}
\begin{proof}
Here $\beta = \frac L4 = J$, so \eqref{eq:zx-gram-refl} at $\delta = 0$ together with \eqref{eq:phi-zero} gives $\Psi_J\, t = \Psi_{0+\beta}\,t = \Phi_0 = wI$, that is $\Psi_J = w\,t \equiv 0 \bmod 2$, since $w = L$ is even. Active orthogonality already supplies $\Psi_\delta = 0$ for every $\delta \in D := \{-(J-1),\dots,J-1\}$, and $D\cup\{J\} = \mathbb Z_{\frac L2}$ because $\frac L2 = 2J$; hence $\Psi_r = 0$ for \emph{every} offset $r$ and $\hat H_X\hat H_Z^T = 0$. So the barrier cannot be broken: every row of $\hat H_X$ is a weight-$L$ operator commuting with all $Z$-checks, so any latent row outside $\rs{H_X}$ is a logical representative of weight $L$ and $d\leq w$ by Fact~\ref{fact:augmentation}. As in Proposition~\ref{prop:distance-bound}, the proviso can fail only if every latent row is already in the active row space, in which case $\mathcal Q$ is its own fully orthogonal parent.
\end{proof}

\begin{proposition}[Girth-4 Dual Codes]\label{prop:zx-girth4}
Let $\mathcal Q\in \mathrm{GALA}_{L,J}(H_k\star C_m)$ with $J\geq 2$ have a ZX-duality $\tau =\iota \otimes t$ where $r$ is a reflection, i.e.\ either $r = r_2$, or $r = r_1r_2$ with $J = \frac L4$, as defined in~\eqref{eq:zx-sector-maps}. Then $\gth{\mathcal{Q}} = 4$, and the number of four-cycles in the Tanner graph of $H_X$ is at least
\begin{equation}\label{eq:zx-four-cycle-count}
    \binom{J}{2}\cdot \frac L2\cdot mk = \frac{J(J-1)L\,mk}{4},
\end{equation}
which for $J = \frac L4$ equals $J^2(J-1)mk$.  The same bound holds for $H_Z$.
\end{proposition}
\begin{proof}
By Lemma~\ref{lemma:girth}, we have $t_4 = \sum_{a<b}\binom{N_{ab}}{2}$ with $N_{ab}$ the integer check-check overlap. For a fixed unordered pair $\{i,i'\}$ of distinct active block rows with $\delta = i - i'$, the pairs with one check in each are indexed bijectively by $(x,y)\in[km]^2$ with overlap $\Phi_\delta[x,y]$. Every row of $\Phi_\delta$ sums to $L$, since by~\eqref{eq:zx-gram-block} it is a sum of $L$ permutation matrices, so $\sum_{x,y}\Phi_\delta[x,y] = L\,mk$. Since $\binom{N}{2}\geq \frac N2$ for every even $N\geq 0$, and $\Phi_\delta$ has even entries because $\Phi_\delta \equiv \Psi_{\delta+\beta}t\equiv 0 \bmod 2$ by~\eqref{eq:phi-zero} and active orthogonality,
\begin{equation*}
    \sum_{x,y}\binom{\Phi_\delta[x,y]}{2}\geq \frac12\sum_{x,y}\Phi_\delta[x,y] =\frac L2mk.
\end{equation*}
Summing over the $\binom J2$ pairs gives \eqref{eq:zx-four-cycle-count}; for $J = \frac L4$ we have $\frac L2 = 2J$ and the bound reads $\frac{J(J-1)}2\cdot 2J\cdot mk = J^2(J-1)mk$. In particular $t_4>0$ as soon as $J\geq2$, so $\gth{\mathcal Q} = 4$. Finally $\tau$ is a permutation with $H_X\tau = H_Z$, so it carries the Tanner graph of $H_X$ isomorphically onto that of $H_Z$ and the same count holds on the $Z$ side.
\end{proof}

The two translation-type ZX dualities escape because they misalign the two index pairings.  Repeating the computation of Lemma~\ref{lem:zx-gram} for $r(i) = i+\beta$, where $G_v = F^T_{v-\beta}\cdot t$, gives instead
\begin{equation}\label{eq:zx-translation}
    \left(H_XH_X^T\right)_{i,i'} = \sum_u \left(F_uF^T_{u-\delta} + F^T_uF_{u-\delta}\right),
    \qquad
    \left(H_XH_Z^T\right)_{i,i'} = \left(\sum_u \left(F_uF^T_{(\delta-\beta)-u} + F^T_uF_{(\delta-\beta)-u}\right)\right)t .
\end{equation}
The two sums are no longer forced to agree term by term. Within the range covered by Proposition~\ref{prop:zx-girth4} --- $J\geq2$, and $J=\tfrac L4$ for $r_1r_2$ --- this leaves the translations $r\in\{e, r_1\}$ as the only ZX-dualities compatible with $\gth{\mathcal Q}\geq 6$. Indeed, certified dual instances of Table~\ref{tab:frontier-selfdual-good} confirms this, with every girth-$6$ entry carrying a translation.

\section{Logical Operations on GALA Codes}\label{sec:logic}

With the GALA construction, a class of automorphisms becomes very transparent: 

\begin{proposition}[GALA Automorphisms]
    Given a GALA code $\mathcal{Q} \in \mathrm{GALA}_{L,J}(H_k\ltimes C^k_m)$, and let $C_N\leq  H_k$ be the central subgroup of $H_k$. We have that

    \begin{equation}
        I_{kL}\otimes C_m \leq \mathrm{Aut}(\mathcal{Q}).
    \end{equation}
\end{proposition}

\begin{algorithm}[b]
\DontPrintSemicolon
\SetKwInOut{Input}{Input}
\SetKwInOut{Output}{Output}
\caption{Hypercube Basis Generation}\label{alg:hypercube-basis}
\Input{ $\mathcal{Q}\in \mathrm{GALA}_{L,J} (H_k\star (C_{m_1}\times ...\times C_{m_q}))$ with logicals $\mathcal{L}'_X = \{\bar X_i\},\ \mathcal{L}'_Z= \{\bar Z_i\}$.} 
\Output{
        $\tilde{\mathcal{L}}_X = \bigsqcup_i \tilde{\mathcal{L}}^i_X$,\
        $\tilde{\mathcal{L}}_Z = \bigsqcup_i \tilde{\mathcal{L}}^i_Z$} 
\BlankLine
$\tilde{\mathcal{L}}_X,\ \tilde{\mathcal{L}}_Z \leftarrow \emptyset$;\quad $i \leftarrow 0$\;
\While{$\mathcal{L}'_X \neq \emptyset$}{
    $(\bar X_i,\ \bar Z_i) \leftarrow \big(\mathcal{L}'_X.\texttt{pop()},\ \mathcal{L}'_Z.\texttt{pop()}\big)$;\quad
    $\tilde{\mathcal{L}}^i_X,\ \tilde{\mathcal{L}}^i_Z \leftarrow \emptyset$\;
    \ForEach(\tcp*[f]{sweep the $C_m$-orbit}){$\mathbf{j} = (j_1, \dots, j_q) \in [m_1] \times \cdots \times [m_q]$}{
        $A_{\mathbf{j}} \leftarrow I_{Lk} \otimes c_{m_1}^{j_1} \otimes \cdots \otimes c_{m_q}^{j_q}$\;
        $(\bar X_{i,\mathbf{j}},\ \bar Z_{i,\mathbf{j}}) \leftarrow
            (\bar X_i A_{\mathbf{j}},\ \bar Z_i A_{\mathbf{j}})$\;
    }
    $\tilde{\mathcal{L}}_X \leftarrow \tilde{\mathcal{L}}_X \sqcup \tilde{\mathcal{L}}^i_X$;\quad
    $\tilde{\mathcal{L}}_Z \leftarrow \tilde{\mathcal{L}}_Z \sqcup \tilde{\mathcal{L}}^i_Z$;\quad
    $i \leftarrow i + 1$\;
    \ForEach(\tcp*[f]{prune seeds absorbed by the new block}){pair $(X', Z') \in (\mathcal{L}'_X, \mathcal{L}'_Z)$}{
        \If{$X' \in \rs{\tilde{\mathcal{L}}_X \cup H_X}$ \textbf{and} $Z' \in \rs{\tilde{\mathcal{L}}_Z \cup H_Z}$}{
            $\mathcal{L}'_X \leftarrow \mathcal{L}'_X \setminus \{X'\}$;\quad
            $\mathcal{L}'_Z \leftarrow \mathcal{L}'_Z \setminus \{Z'\}$\;
        }
    }
}
\Return{$\tilde{\mathcal{L}}_X,\ \tilde{\mathcal{L}}_Z$}
\end{algorithm}


Since these automorphism will preserve the space of logical operators, we can use them to discover new logical representatives on which the logical action have a clean description

\begin{proposition}[Hypercube Organization]\label{prop:hypercube-logical}
    Suppose $\mathcal{Q}\in \mathrm{GALA}_{J,L}(H_k\star C_m)$ is a GALA code with logical operators $L'_X, L'_Z$, and suppose $C_m = \mathbb Z_{p_1}\times ...\times\mathbb Z_{p_{q}}$ with $\gcd(p_i, p_j) = 1$ for $0<i< j\leq q$. We can efficiently find a set of over-complete logical operators $\tilde L_X, \tilde L_Z$ using algorithm~\ref{alg:hypercube-basis} that can be organized in disjoint sets of $q$-dimensional hypercubes:

    \begin{align}
        \tilde L_X &= \bigsqcup_{i = 1}^s \tilde L_X^i, \\
        \tilde L_Z &= \bigsqcup_{i = 1}^s \tilde L_Z^i;
    \end{align}

    such that the coordinate permutation $I_{kL}\otimes c\in C_m$ induces a logical permutation $c$ on each $\tilde L^i_X, \tilde L^i_Z$. 

    Furthermore, the constructed logical operators $\tilde L_X, \tilde L_Z$ preserve the weight of $L_X, L_Z$; that is:
    \begin{align}
        \min_{x\in L_X}(\wt x) &= \min_{\tilde x\in \tilde L_X}(\wt{\tilde {x}}),\\
        \min_{z\in L_Z}(\wt z) &= \min_{\tilde z\in \tilde L_Z}(\wt{\tilde{z}});
    \end{align}
\end{proposition}

\begin{proof}
    We can find this set of basis using algorithm~\ref{alg:hypercube-basis}. Since the basis is only transformed using permutation, its weight is preserved. 
\end{proof}

\begin{corollary}(Bottom Logicals)
Given GALA code $\mathcal{Q}\in \mathrm{GALA}_{L,J}(H_k\times C_m)$. We can find weight $kL$ sector-disjoint logicals
\begin{align}
    L_X^{C_m} &= \bigsqcup_{p = 1} ^{\frac{L}{2} - J} L^p_X, \\
    L_Z^{C_m} &= \bigsqcup_{p = 1} ^{\frac{L}{2} - J} L^p_Z,
\end{align}
such that $\rank(L_X^{C_m}\backslash H_X) = \rank(L_Z^{C_m}\backslash H_Z) = Jm - k +1$ and $\supp{X_i} \cap \supp{X_j} = 0$ for all $X_i, X_j\in L^k_X$ if $\mathcal{Q}$ is monomial (and similarly for $L_Z^{C_m}$). 
\end{corollary}
\begin{proof}
    This follows exactly from Lemma~\ref{lem:inflation}, restated in the hypercube language, where each sector correspond to a latent row of cyclic matrices that each disjoint rows.
\end{proof}
\subsection{Dirty Shifts via Code Automorphisms}
The first class of logical operations we get are dirty cyclic shifts arising from the cyclic code automorphism group. These cyclic shifts are "dirty" because the hypercube logical basis is typically over complete.

\begin{proposition}[Dirty Cyclic Shifts]
Let $\mathcal{Q}\in \mathrm{GALA}_{L,J}(H_k\star C_m)$ and let \begin{align*}
            \tilde L_X &= \bigsqcup_{i = 1}^s \tilde L_X^i, \\
        \tilde L_Z &= \bigsqcup_{i = 1}^s \tilde L_Z^i,
\end{align*}
be the hypercube logical basis as in ~\ref{prop:hypercube-logical}. Then, the physical action of $c\in I_L\times I_k\times C_m$ induces the logical cyclic shift $\bar c \in C_m$ on each sector of the hypercube logicals. That is,
\begin{equation}
     L_X^i\cdot I_{kL}\times c = c\cdot  L_X^i;\quad 
     L_Z^i\cdot I_{kL}\times c = c\cdot  L_Z^i.
\end{equation}
for all $i\in [s]$.
\end{proposition}
\begin{proof}
By algorithm~\ref{alg:hypercube-basis}, each hypercube are generated from the orbit of a representative basis, which gives the desired cyclic shift structure.
\end{proof}
\subsection{Transversal CNOT Gates via Chain Homomorphisms}
\begin{proposition}[Quotient Homomorphisms]\label{prop:q-homo}
Let $\mathcal{Q}\in \mathrm{GALA}_{L,J}(H_k\star C_m)$ with normal subgroup $N$ and let $\mathcal{Q}_N$ its quotient code by collapsing $N$. Let $\mathcal{Q}: Z\xrightarrow{H_Z^T} Q\xrightarrow{H_X} X$ and $\mathcal{Q}_N: Z\xrightarrow{(H'_Z)^T} Q\xrightarrow{H'_X} X$ be the associated chain complexes. We have that

\begin{equation}
    \begin{tikzcd}[ampersand replacement=\&]
        {Z} \& {Q} \& {X} \\
        {Z'} \& {Q'} \& {X'}
        \arrow["{H_Z^T}", from=1-1, to=1-2]
        \arrow["{H_X}", from=1-2, to=1-3]
        \arrow["{H_Z'^T}", from=2-1, to=2-2]
        \arrow["{H_X'}", from=2-2, to=2-3]
        \arrow["{\gamma_Z}", from=1-1, to=2-1]
        \arrow["{\gamma_Q}", from=1-2, to=2-2]
        \arrow["{\gamma_X}", from=1-3, to=2-3]
    \end{tikzcd}
\end{equation}
commutes for maps $\gamma_Q = I_L\otimes\pr_N$ and $\gamma_X = \gamma_Z = I_J\otimes\pr_N$, where $\pr_N$ is the orbit-sum matrix of Section~\ref{ssec:notations}, with one row per $N$-orbit.
\end{proposition}

\begin{proof}
Everything reduces to one identity. Because $N$ is normal, every $g$ in the lift group carries $N$-orbits to $N$-orbits, $g(Nx) = N g(x)$, so it induces a permutation $\bar g$ of the $km/\mu$ orbits. Reading $[\pr_N]_{O,x} = 1$ exactly when $x\in O$, and using the right-action convention $e_xM_g = e_{g(x)}$ of Definition~\ref{def:mtx-perm},
\[
    \big[\pr_N M_g\big]_{O,y} = \sum_{x\in O}\delta_{y = g(x)}
    = \begin{cases}1, & y\in g(O),\\ 0,&\text{else,}\end{cases}
    \qquad
    \big[M_{\bar g}\pr_N\big]_{O,y} = [\pr_N]_{\bar g(O),y},
\]
and these agree since $\bar g(O) = g(O)$. Hence
\begin{equation}\label{eq:orbit-intertwine}
    \pr_N M_g = M_{\bar g}\,\pr_N \qquad\text{for every } g,
\end{equation}
and by linearity the same holds for any $\mathbb F_2$-sum of group elements, so the argument covers polynomial lifts as well as monomial ones.
\end{proof}
\begin{proposition}[tCNOTs]

Let $\mathcal{Q}\in \mathrm{GALA}_{L,J}(H_k\star C_m)$ and let \begin{align*}
            \tilde L_X &= \bigsqcup_{i = 1}^s \tilde L_X^i, \\
        \tilde L_Z &= \bigsqcup_{i = 1}^s \tilde L_Z^i,
\end{align*}
be the hypercube logical basis as in ~\ref{prop:hypercube-logical} and suppose each logical $X^i_{j_1,...,j_q}\in  L^i_X$ is indexed by its $q$ coordinates in the hypercube. Then,  let $N \leq C_m = \prod_{i\in \mathcal{N}}\mathbb Z_{i}$ be a subgroup of $C_m$ consisting of the factors $\mathcal{N}\subset \{p_1,...,p_q\}$; there exists a homomorphic measurement gadget that measures 
\begin{align*}
    \bigsqcup_{i = 1}^s \bigsqcup_{\textbf{p}\in \mathbb F_2^{[q]-\mathcal{N}}}\prod_{p_N\in \mathbb F_2^{\mathcal{N}}} X_{\textbf{p}, \textbf{p}_N}
\end{align*}
in parallel. The fault distance of this operation is $d_{\mathcal{Q_N}}$.
\end{proposition}
\begin{proof}
    Follows from the quotient homomorphisms~\ref{prop:q-homo} of factors of the bottom code.
\end{proof}
\section{Physical Implementations on RNAA}\label{sec:hardware}
In this section, we make a case for GALA codes in comparison to the original construction using APMs. Firstly, similar to APMs, product-groups also allow us to design active orthogonality, and with more simplicity afforded by compressing commutativity conditions to a small $H_k$ factor. Secondly, similar to the reference APM approach taken in the followup co-design work, GALA also restrict the search space to AOD compatible ones, and more controllably so: while the reference APMs with desirable orbit structures is neither prescriptive nor does it always ensure all transitions are rigid global cyclic shifts, such properties for GALA codes (especially DPG and SPG-restricted codes) follows directly from the abelian group $C_m$, with group elements that translate directly to movement schedules.

\subsection{Hardware Compatibility by Construction}\label{ssec:matching-hardware}

 We ground our claims that all GALA codes make efficient use of crossed-AODs by constructing an atom layout and move sequence for their syndrome extraction circuit, building off of atom layouts for previous codes \cite{zhao2026ultrahighratequantumerrorcorrection}. In GALA codes, the order of qubits as specified in the check matrices maps directly to the locations of neutral atoms in a 2D layout in the quantum computer and block permutations correspond directly to row and column permutations within those blocks of atoms. We obtain an overall fast elapsed time for syndrome extraction because these required permutations correspond to ``easy'' atom moves.

To understand what an ``easy'' atom move is on neutral atom quantum computers we need to understand AODs. An Acousto-optic deflector is a dynamically controllable diffraction grating that allows us to copy one laser beam into multiple beams with arbitrarily-controllable positions along a line. By combining two AOD devices at 90 degrees, one oriented to the x-axis and to the y-axis, we obtain a crossed-AOD which copies a laser beam to a set of x-coordinates and then copies each of those to a set of y-coordinates, producing a grid of dynamically movable beams with the restriction that they stay in an irregular grid \textit{and} those grid sites cannot cross during moves. Each beam is a trap which can pick up and move atoms to arbitrary locations in a 2D grid. Because we cannot collide atoms during moves, when we do cyclic permutations we must either move in two sequential steps or use two sets of crossed-AOD devices to avoid conflicts. We use either two or four crossed-AODs for our estimates here.

We now describe the layout in 2D of the data and ancilla atoms during each step of syndrome extraction. We refer to \textit{blocks} of $P$ data or ancilla atoms as groups of atoms that we arrange, move, and permute as a group. For a GALA code with parameters $L$ and $J$, where $L$ is typically 8 or 12 and $J\ge L/2$, there are $L$ blocks of data qubits, $J$ blocks of X ancilla and $J$ blocks of Z ancilla. We use a 2D arrangement of sites identical to \cite{zhao2026ultrahighratequantumerrorcorrection} with placement of $L+J$ pairs of blocks vertically but we do some moves more optimally. 

A key choice when laying out is how to split $P$ into rows and columns. For example, if the GALA code uses group $S2\times S3\times S14$ with $P=2*3*14=84$, multiple layouts are valid but will have slightly different distances needed to move atoms which impacts the total time.  We have the following:

\begin{proposition}[Regular Movements]\label{prop:regular-moves}
    Let $\mathcal Q\in \mathrm{GALA}_{L,J, \star}(H_k\ltimes_{\mathcal{R}} C_m)$ be a restricted semi-direct product GALA code, and suppose $C_p\times C_q = C_m$. There is a natural layout where the data qubits are arranged in a $kp\times Lq$ array and the check qubits are arranged in a $kp\times Jq$ arrays such that syndrome extraction is implementable over $L$ rounds of independent row swaps $\sigma_r$ and column swaps $\sigma_c$ on the check qubits. In particular, if $i\in \mathcal{R}$, $F_i = (h_i, c_i\otimes c'_i)\in H_k\times C_p\times C_q$ ($\mathcal{R} = [L]$ for DPK constructions), the shifts are given by 

    \begin{align*}
        \sigma_{r,i} &=  h_{i}h_{i-1}^T\otimes c_ic_{i-1}^T,\\
        \sigma_{c,i} &=  c'_i{c_{i-1}'}^T.
    \end{align*}

    which can be efficiently implemented using AOD compatible moves. Otherwise, for $i\neq \mathcal{R}$, $F_i = (h_i; c_{i0}\otimes c'_{i0}, ..., c_{ik}\otimes c'_{ik})\in H_k\star(C_p\times C_q)$, the moves are
    \begin{align*}
        \sigma_{r,i} &= h_{i}h_{i-1}^T\otimes I_p \cdot c_{i0} c_{i-1,0}^T\oplus ...\oplus c_{ik} c_{i-1,k}^T,\\
        \sigma_{c,i} &= c'_{i0} {c'_{i-1,0}}^T\oplus ...\oplus c'_{ik} {c'_{i-1,k}}^T.\\
    \end{align*}
    \end{proposition}
    Proof follows from the CRT. Note picking $H_k\star C_m = S_3\times \mathbb Z_{32}$ and $S_3\times \mathbb Z_2\times \mathbb Z_{32} $ recovers the optimal moves for the $P=96$ and $P=192$ codes from \cite{zhao2026ultrahighratequantumerrorcorrection}.

In this factorization, any non-commuting factor $H_k$ becomes (part of) the row permutation to reduce the chance for conflicting AOD moves. All factorizations of the groups into row and column groups are valid so for $[[672,336,12]]$ we choose 2 rows and 42 columns and for $[[132,30,12]]$ we choose 1 row and 11 columns per block.

The purpose of atom permutations is to rearrange the qubits so that we can apply a layer of parallel CZ gates between all paired atoms. $F_i$ and $G_i$ specify how to arrange the ancilla relative to the data atoms therefore $F_j F^{-1}_i$, $G_j F^{-1}_i$, $G_j G^{-1}_i$, and $F_j G^{-1}_i$ specify the changes to the previous permutation to obtain the permutation for the next CZ layer. A key insight that improves our move parallelism is the observation that the Z check matrix contains the inverse block permutations of the X check matrix so to do a permutation for the next X check and the next Z check in parallel we permute the ancilla needed for the non-inverse permutation but permute \textit{half of the data} the same amount and by permuting the data instead of the ancilla, we obtain the inverse \textit{relative permutation} between the Z ancilla and the data.

With this arrangement, we choose the permutation sequence that minimizes move time (sequence $F_0$, $F_3$, $F_4$, $F_5$, $F_2$, $F_1$, $G_1$, $G_4$, $G_2$, $G_5$, $G_3$, $G_0$ for $[[132,30,12]]$ and sequence $F_0$, $F_2$, $F_{3a}$, $F_{3b}$, $F_{1a}$, $F_{1b}$, $G_{1b}$, $G_{1a}$, $G_{3b}$, $G_{3a}$, $G_{2}$, $G_{0}$ for $[[672,336,12]]$). A design feature of GALA codes is the sum of multiple permutation matrices in some blocks of the check matrix. This enables high rate at smaller code block size and also reduces the number of times we have to rearrange the blocks of atoms. The a and b subscripts refer to the two summed permutations in the check matrix for $[[672,336,12]]$. From the positions determined by this sequence of permutations, assuming a distance between atom pairs of 12um, vertical and horizontal, and assuming movement acceleration of $5500 m/s^2$ \cite{zhou2025resource} we estimate the total time for syndrome extraction. $[[672,336,12]]$ takes 8.868ms with two crossed-AODs or 6.755ms parallelizing block and column permutations with four crossed-AODs. $[[132,30,12]]$ takes 5.257ms with two crossed-AODs or 3.100ms by parallelizing block and column permutations with four crossed-AODs.

%

\begin{table*}
\centering
\begin{tabular}{lllllll}
\toprule
Desideratum &  Quasi-Cyclic \cite{Hagiwara_2007}&Kasai \cite{kasai2026breakingorthogonalitybarrierquantum}
& QuEra
\cite{zhao2026ultrahighratequantumerrorcorrection} & Paired Partition\cite{okada2026pairpartitionconstructionscpmbasedquantum}&SPG~(Sec.~\ref{sec:construction})&DPG~(Sec.~\ref{sec:construction})\\\midrule
Rate& \color{purple}$?$&\color{teal} $\geq 50\%$& \color{teal}$\geq 50\%$&\color{teal}$\geq 50\%$& \color{teal}$\geq 50\%$&\color{teal}$\geq 50\%$\\
Length& \color{teal} $10^2 - 10^3$& \color{purple} $\geq 5\times 10^3$& \color{olive} $10^3-5\times 10^3$& \color{teal} $10^2 - 10^3$& \color{teal} $10^2 - 10^3$&\color{teal} $10^2 - 10^3$\\
Girth $\geq 6$&  {\color{teal}Yes} &{\color{teal}Yes} & {\color{teal}Yes} &{\color{teal}Yes} & {\color{teal}Yes} &{\color{teal}Yes} \\
AOD Compatibility&  {\color{teal}Parallel Cyclic}&\color{purple}?& \color{teal}Parallel Cyclic &\color{olive}Serial Cyclic& \color{olive}Serial Cyclic&{\color{teal}Parallel Cyclic}\\
Low-weight Basis&  \color{purple}?&\color{purple}?& \color{purple}?&{\color{teal}Yes}& {\color{teal}Yes} &{\color{teal} Yes}\\
 Automorphisms& \color{olive} Cyclic Shift & \color{purple}?&\color{purple}?& \color{olive}Cyclic Shift& \multicolumn{2}{c}{\color{teal} Customizable Shifts; Clifford if ZX-dual}\\
\bottomrule
\end{tabular}
\caption{Summary of desiderata of related constructions.}
\label{tab:comparison}
\end{table*}

\subsection{Automorphism Covariant Logical Operators}\label{ssec:case-search}

{The hypercube basis of Proposition~\ref{prop:hypercube-logical} is over-complete in general, which is why its shifts are dirty. When the logical module is free over $R=\F_2[C_m]$, we can instead ask for a complete symplectic basis with the same symmetry.}

{\begin{definition}[Covariant frame]\label{def:covariant-frame}
A covariant frame of $\mathcal Q$ is a set of seed pairs $\{(x_s,z_s)\}_{s\in[r]}$ such that the operators $\bar X_{s,h}=X(h\cdot x_s)$ and $\bar Z_{s,h}=Z(h\cdot z_s)$, $h\in C_m$, satisfy $\braket{\bar X_{s,h},\bar Z_{s',h'}}=\delta_{ss'}\delta_{hh'}$. The frame may be completed by a residual set of translation-invariant logical qubits.
\end{definition}

\begin{proposition}[Clean shifts and matching folds]\label{prop:clean-frame}
On a covariant frame, 
\begin{enumerate}
    \item every translation $g\in C_m$ acts as the logical permutation $(s,h)\mapsto(s,h+g)$.
    \item If a ZX-duality $\tau$ sends every frame class to a frame class, $Z(\tau\,\bar x_{s,h})\equiv\bar Z_{\sigma(s,h)}$ for a fixed-point-free involution $\sigma$, then $H_\tau=\tau\circ H^{\otimes n}$ acts as $\bar H^{\otimes k}$ followed by SWAPs along $\sigma$, and $S_\tau=Z(w)\prod_q CZ(q,\tau q)$ acts as $\overline{CZ}$ along $\sigma$, up to logical Pauli corrections.
\end{enumerate}
\end{proposition}
\begin{proof}
(i) follows from the definition. (ii) $H_\tau$ maps $X(v)\mapsto Z(\tau v)$, and $S_\tau$ maps $X(v)\mapsto\pm X(v)Z(\tau v)$; naming the class of $Z(\tau v)$ in the frame gives $\bar Z_{\sigma(s,h)}$.
\end{proof}}

For a self-dual code we also impose a reflection $\rho=T^u\tau$ from Proposition~\ref{prop:meta-duality}: the seeds satisfy $\rho a\equiv a$ and $\rho b\equiv b$, and the frame is $\bar X_{1,t}=X(T^ta)$, $\bar Z_{1,t}=Z(T^tb)$, $\bar X_{2,t}=X(T^tb)$, $\bar Z_{2,t}=Z(T^ta)$. Such a frame is translation covariant, reflection covariant, and self-dual ($L_X=L_Z$ as sets), and it is orthonormal exactly when $\braket{a,T^tb}=\delta_{t,0}$ and $\braket{a,T^ta}=\braket{b,T^tb}=0$ for all $t$. Palindromic seeds are, up to stabilizers, lifts of words of the folded code $H_f=HP^{T}$ on the $n/2$ reflection pairs, where $P$ is the pair-incidence matrix. For $[[132,30,12]]$ the folded code is a self-orthogonal $[[66,16]]$ code, and the search reduces to the following loop: for each logical $a$ with $\rho a=a$, for each logical $b$ with $\rho b=b$, check that $\braket{a,T^tb}$ is nonzero for exactly one $t$. The loop is exhaustive over a census of low-weight logicals. It yields the $(14,14)$ frame of $[[132,30,12]]$: $22$ weight-$14$ covariant qubits, an $8$-qubit residual of weight $20$--$22$, and a D$_{11}$-covariant frame on which both $T$ and the reflections act as logical permutations.

We use three following reductions make the frame searches exhaustive rather than heuristic. 
\begin{enumerate}
    \item \emph{Translations.} Every $C_m$-orbit of kernel vectors has a representative whose smallest index is a block start, so a branch-and-bound enumeration of all logicals of weight $\le W$ only needs to root at the $L$ block starts. The pairing of two orbits is a single ring element, which transforms covariantly under translations, so all zero/unit tests run on orbit pairs. A frame is fixed by its $r$ seeds rather than by $k$ vectors. 
    \item  \emph{Reflections.} For self-dual codes, palindromic seeds live in the folded code on $n/2$ qubits, as above. 
    \item \emph{Dualities.} Conjugating by a translation $h$ sends $g\tau$ to $(g+2h)\tau$, so the $m$ dualities fall into $|C_m/2C_m|$ classes: four for $[[672,336,12]]$. With these reductions, the complete census of the $985{,}385$ logical orbits of $[[672,336,12]]$ up to weight $18$ takes about four hours on $17$ cores. That census certifies the optimality and fault-distance.
\end{enumerate}

\section{Numerical Results}
\subsection{Code Search}
Now, we describe the code searching process. First, we will pick $L, J, \Gamma$, a group product $\star$, and two groups: one small non-abelian group $H_k$ and one large abelian group $C_m$. The sweep over these hyperparameters can be restricted by the bounds of Section~\ref{sec:bounds}. Furthermore, the group structures and products used should be customized to match physical constraints and logical needs, which we covered in Section~\ref{sec:hardware} and Section~\ref{sec:logic}.

The sampling of the code proceeds in two stages: we first enumerate permutation-inequivalent top ansatzes $\mathcal{F}^H, \mathcal{G}^H$ that satisfy the anti-commutativity pattern given by the active set $\Gamma$. For small $k$, this can be done efficiently and exhaustively. To sample a DPG, we then randomly sample $L$ bottom elements $(c_1,...,c_L)\in C^L_m$, where $m$ is odd. We can estimate the distance of this bottom code quickly, and, whenever $\deg H_k$ is odd as it is for the $S_3$ tops we use, reject samples with $(\deg H_k)\,d_{\mathrm{bottom}}\leq L$, which by Corollary~\ref{cor:quotient-distance} already forces $d\leq L$. To sample a SPG, we need to pay additional care to ensure commutativity. We can either restrict our ansatzes to ones that meet the sufficient condition listed in Lemma~\ref{lemma:spk-commutivity} to allow for the maximum number of AOD compatible elements (it suffices to leave a $2/L$ fraction of the generators non-compatible in the most restricted case); or we can inherit the DPG ansatzes, and restrict the $c\in C_m$ sampled to only the ones satisfying Equation~\eqref{eq:spk-commutivity}. We label the former SPG-\emph{restricted} and the latter SPG-\emph{full}.

Finally, we post-select on girth $\geq 6$ and use \texttt{QDistRnd} to estimate the code distance~\cite{Pryadko_2022}. A \texttt{QDistRnd} run reports the lowest logical weight it happens to hit, and is therefore an \emph{upper} bound on $d$; we consequently spend trials adaptively, escalating only on codes that are still candidates for a large distance. Every sampled code receives an initial pass of $8$ random trials, and is promoted to the next tier --- $1024$, then $10^4$, then $10^5$ trials --- only while the smallest logical weight found so far is at least $6$, then $8$, then $10$. Codes that survive the last tier and lie on the Pareto frontier receive a final pass of $2\times10^5$ trials. We performed distance estimation in this way on more than $2\times10^5$ randomly sampled codes with girth $\geq 6$ and $n\leq 3000$, across the DPG, SPG-restricted, and SPG-full constructions and predominantly with polynomial lifts, with results reported in Figure~\ref{fig:codes}. Two reference families are plotted in black there for comparison: the co-designed Kasai baselines of~\cite{zhao2026ultrahighratequantumerrorcorrection} and the pair-partition CPM codes of~\cite{okada2026pairpartitionconstructionscpmbasedquantum}.

Because an estimate of this kind can only ever be an upper bound, every instance we report below is additionally \emph{certified}: we determine $d$ exactly by exhaustive enumeration, which is feasible in the regime we target, $n\leq 2500$ and $d\leq 16$. This is the certification standard established for this code family in~\cite{okada2026pairpartitionconstructionscpmbasedquantum}; the enumeration used here is our own implementation. It has two independent parts, and both are carried out separately in the $X$ and in the $Z$ sector. The first is \emph{exclusion}: we rule out every non-trivial logical operator of weight strictly below $d$, i.e.\ a complete search of $\ns{H_Z}\setminus\rs{H_X}$ (for the $X$ sector) and of $\ns{H_X}\setminus\rs{H_Z}$ (for the $Z$ sector) through weight $d-1$. In practice the search only has to reach weight $d-2$: every column weight of $H_X$ and $H_Z$ is odd for the instances reported here, so $\textbf{1}^{1\times n}$ lies in each row space, every logical operator has even weight, and the odd weights can be skipped. The second is a \emph{witness}: we exhibit one explicit operator of weight exactly $d$ in each sector and re-verify directly that it is a non-trivial logical. Together the two parts give $d_X$ and $d_Z$ exactly, hence $d=\min(d_X,d_Z)$. Explicit group-ring generators for the certified frontier instances are tabulated in Tables~\ref{tab:frontier-nondual}--\ref{tab:frontier-other}.
\begin{figure*}
    \centering
        \includegraphics[width=.46\linewidth]{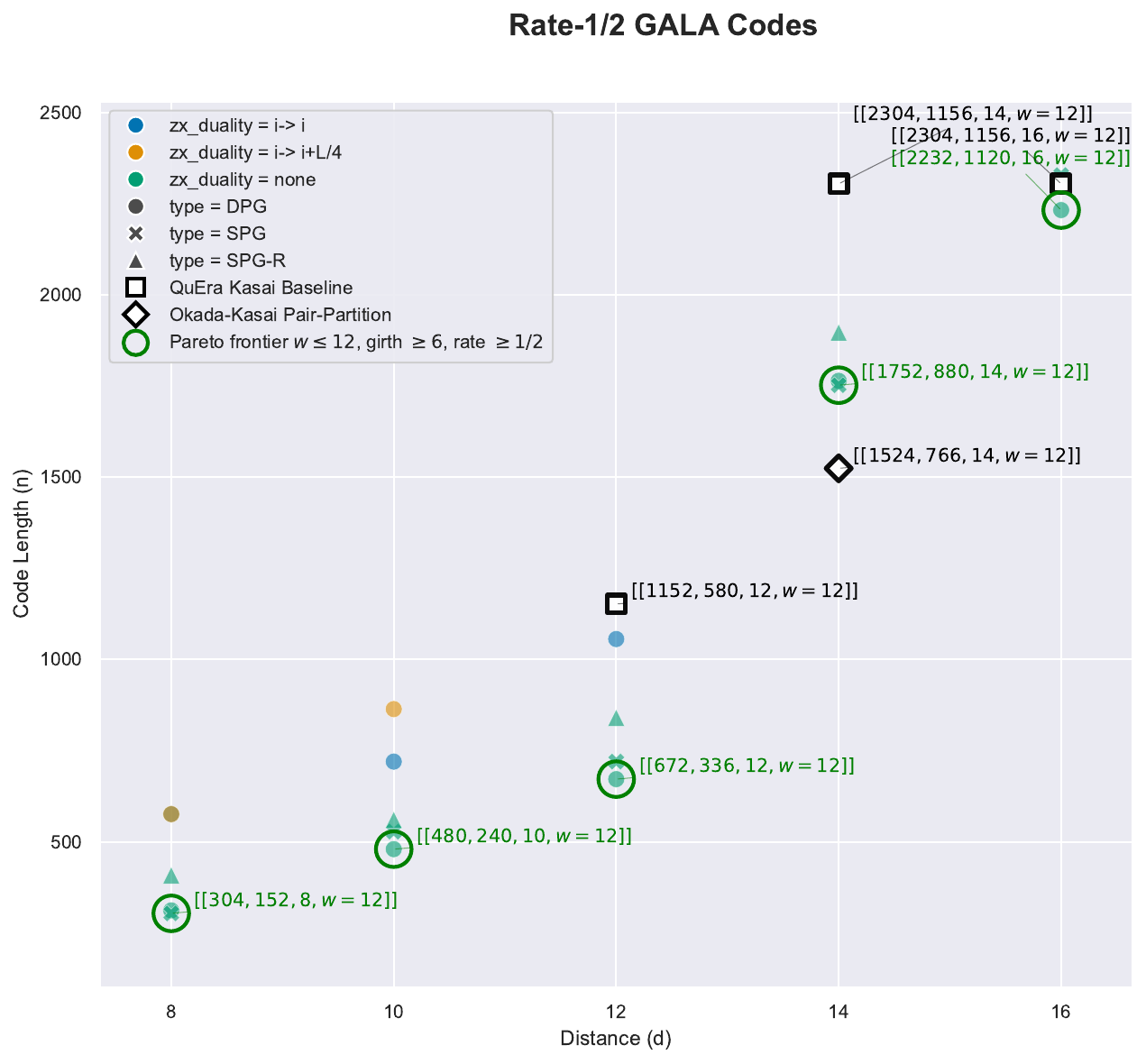}
        \includegraphics[width=0.46\linewidth]{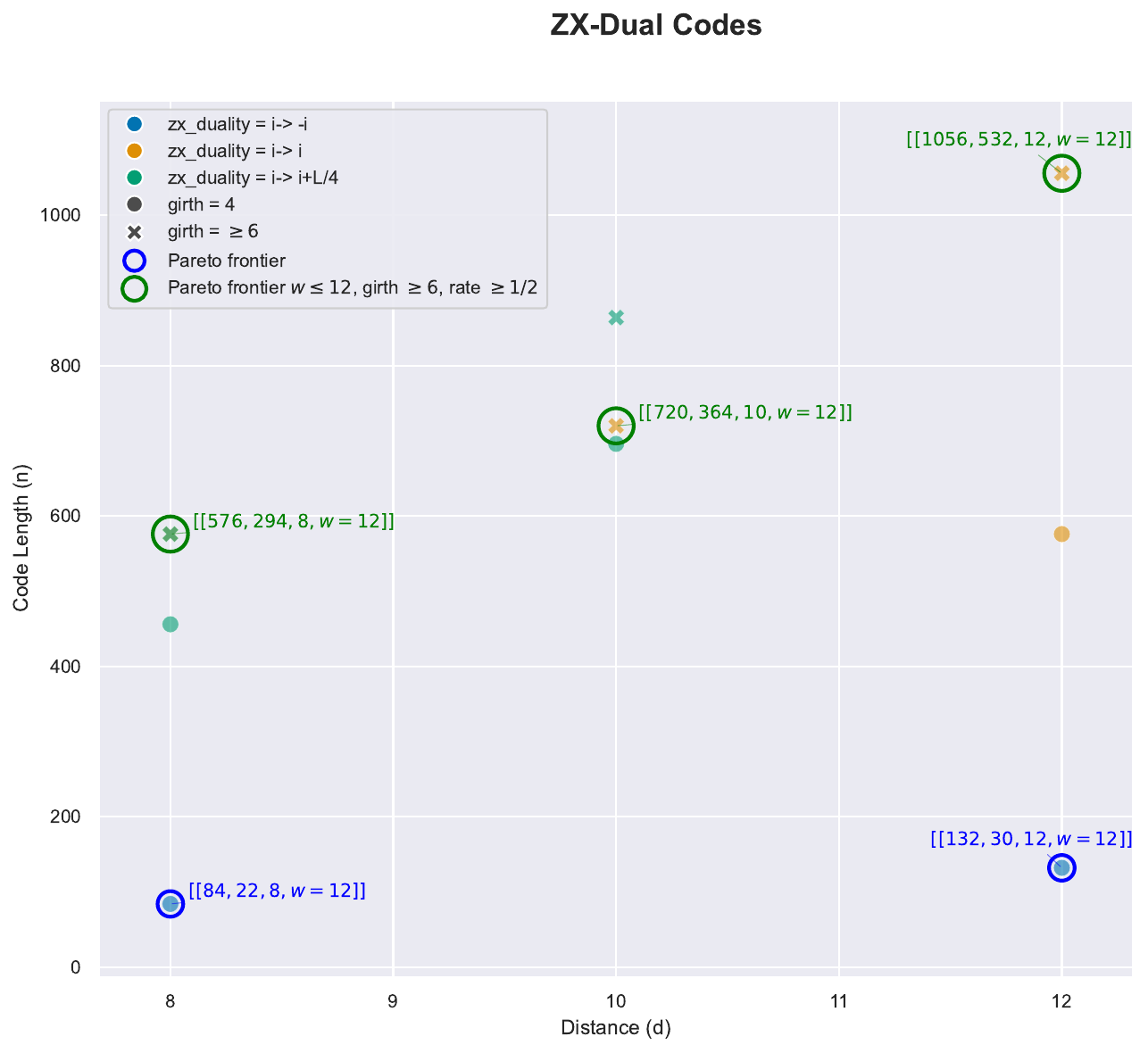}
        \caption{Codes discovered in the preliminary search. Both panels plot certified distance $d$ horizontally against code length $n$ vertically, so towards the lower right --- more distance for fewer qubits --- is better.
        (Left): the practical regime, $n\leq2500$, rate $\geq1/2$, $w\leq12$ and girth $\geq6$; one representative, the best instance, is shown for each (girth, $L$, construction-type) class, and the annotated instances are certified frontier codes, listed with their row weight $w$. Marker shape and color distinguish the construction type (DPG, SPG-restricted, SPG-full) and the ZX-duality sector involution of Eq.~\eqref{eq:zx-sector-maps}, if any.
        (Right): ZX-dual GALA codes. Marker shape and color record the duality class (which sector involution of Eq.~\eqref{eq:zx-sector-maps} the instance realizes, $i\mapsto i$, $i\mapsto i+L/4$, or the reflection $i\mapsto -i$) and  girth. Consistently with Proposition~\ref{prop:zx-girth4}, the reflection class is girth $4$. The highlighted subset is the Pareto frontier of those instances that additionally satisfy $w=12$, girth $\geq6$ and rate $\geq1/2$; it contains the self-dual $[[1056,532,12]]$.}
    \label{fig:codes}
\end{figure*}

In addition to rediscovering codes matching known Kasai parameters, we also find:

\begin{enumerate}
    \item barrier-breaking instances with $d>w=12$: the $[[1752, 880, 14]]$, at $76\%$ of the length of the previously known $[[2304, 1156, \leq 14]]$~\cite{zhao2026ultrahighratequantumerrorcorrection}, and the $[[2232, 1120, 16]]$, which is shorter still than that same $n=2304$ baseline while reaching a higher and exactly certified distance;
    \item a $[[676,336,12]]$ at $58\%$ of the length of the previously known $[[1152, 580, \leq 12]]$~\cite{zhao2026ultrahighratequantumerrorcorrection}; here $d = w = 12$, so this instance does not itself break the weight barrier, but it attains the baseline distance on $5/8$ of the qubits;
    \item short codes with rate-$1/2$ reaching $d = 10$: the direct-product $[[480,240,10]]$, the shortest certified $d=10$ instance we found, and the semi-direct-product $[[304, 152, 8]]$, the shortest such instance among the SPG families that reach $d = 8$;
    \item self-dual codes $[[132,30,12]]$ and $[[1056, 532,12]]$ with fold-transversal Clifford gates. The two sit in different regimes: $[[1056,532,12]]$ has girth $6$ and rate $532/1056 > 1/2$, whereas $[[132,30,12]]$ has girth $4$ and rate $30/132 \approx 0.227$, so the latter is not covered by the girth-$\geq6$, rate-$\geq1/2$ claims made for the other instances and is reported purely for its compactness and its transversal gates.
\end{enumerate}

One further direct-product instance, on $n=2160$ qubits with $k=1084$, has a completed $X$-sector enumeration establishing $d_X=16$ exactly; its $Z$-sector enumeration is not yet complete, so we do not quote it as a certified $[[n,k,d]]$ instance above. Every instance that we do quote carries an exactly certified distance in the sense of the previous paragraph --- all logical operators below $d$ excluded, plus a verified weight-$d$ witness on each side --- so no ``$\leq$'' appears on any GALA parameters that we report as results. The only upper bounds that appear anywhere in this work are the \texttt{QDistRnd} estimates of the uncertified instances in the right panel of Figure~\ref{fig:codes} and the distances quoted for the baselines.

\subsection{Circuit level logical error rates}\label{ssec:ler}

We estimate the logical error rate of our two smallest $d=12$ codes and compare them against the Gross code (BB[[144,12,12]])\cite{Bravyi_2024}, the smallest Pair Partition code (PP[[232, 62,12]]) \cite{okada2026pairpartitionconstructionscpmbasedquantum} and the Surface Code at distance 11 and 13\cite{Fowler_2012}. We use a standard error model where measurement, readout, and all gate errors are parametrized by a single phyiscal error $p$. We do not account for idle errors as consistent with prior work \cite{zhao2026ultrahighratequantumerrorcorrection}; indeed, with coherence time on the scale of 10s of seconds, the idle error during gate layer are many orders of magnitude smaller.

We report the probability that any logical qubit fails in a QEC cycle, or  $p_L = 1 - (1-p_{fail})^{1/kN_c}$, using the same decoder family on all baselines to have an apples-to-apples comparison. Sampling with Stim \cite{Gidney_2021} and decoding with a two-stage Relay-BP cascade\cite{lu2026cornucopia} on a greedy coloration syndrome extraction circuit, we accumulate $2.9 \times 10^9$ shots and $1.0 \times 10^4 $ logical failures across [[132,30,12]] and [672,336,12]] and four baselines, with the surface code decoded by minimum-weight matching. The result is reported in Figure~\ref{fig:ler}. We fit the data to the conventional polynomial-exponent form $p^{d/2} \exp(c_0 + c_1 p + c_2 p^2)$ \cite{Bravyi_2024}. We find the pseudo-thresholds of the GALA codes is at around $0.38–0.41\%$.

\edit{Both GALA codes are measured directly at $p = 10^{-3}$ rather than extrapolated there. With the two-stage decoder, $[[132,30,12]]$ reaches $9.6_{-1.6}^{+1.8}\times 10^{-9}$ from $134$ logical failures, and $[[672,336,12]]$ reaches $5.1_{-0.86}^{+0.99}\times 10^{-9}$ from $120$ logical failures. Quoted intervals are $95\%$ confidence interval. We measure these points directly because the extrapolation is not reliable at this distance from the data: an earlier fit to our $p \geq 1.3\times10^{-3}$ points predicted $2\times 10^{-11}$ for $[[672,336,12]]$, which the measurement exceeds by more than two orders of magnitude. The same concern is raised in \cite{zhao2026ultrahighratequantumerrorcorrection}, where the large number of logical qubits in a high-rate block steepens the logical error scaling near threshold and makes extrapolation across the waterfall unreliable.}

\edit{The measured curve of $[[672,336,12]]$ flattens below $p = 1.5\times 10^{-3}$. Its block error probability falls by only a factor of $2.7$ between $p = 1.5\times10^{-3}$ and $p=10^{-3}$, where distance-limited scaling would give about two orders of magnitude. This is a property of the decoder rather than of the code. At $p = 10^{-3}$, $101$ of the $120$ failures of $[[672,336,12]]$ are converged-but-wrong, that is, the cascade returns a confident but incorrect logical class, whereas $95$ of the $105$ failures of $[[132,30,12]]$ are non-convergences. A converged-but-wrong failure is not heralded and does not escalate, so it is not repaired by a more accurate fallback decoder and it does not fall off with $p$ at the rate the distance would suggest.}

\edit{At $p = 10^{-3}$, $[[672,336,12]]$ therefore attains $5.1\times 10^{-9}$ at an overhead of $3.0$ physical qubits per logical qubit, which is comparable to the $3.7\times10^{-9}$ of the distance-11 surface code at an overhead of $241$, an improvement of roughly eighty-fold in overhead at matched logical error rate. Against the qLDPC baselines the comparison is less favorable: the Gross code extrapolates to $1\times 10^{-10}$ at an overhead of $24$ and the pair-partition code to $7\times 10^{-12}$ at an overhead of $5.6$. We stress that these two baseline numbers are extrapolations, not measurements. Our lowest data point for the Gross code is at $p = 3\times10^{-3}$ and for the pair-partition code at $p = 2.5\times 10^{-3}$, so their quoted values at $10^{-3}$ sit three and four orders of magnitude below any physical error rate at which we observed failures. An extrapolation of that length cannot reveal a decoder floor of the kind we measure for $[[672,336,12]]$, and we do not claim that these baselines are free of one.}

\edit{Fitting the polynomial-exponent form over the measured range gives effective distances of $d_{\mathit{eff}} = 9.1$ for $[[132,30,12]]$ and $d_{\mathit{eff}} = 7.3$ for $[[672,336,12]]$, both well below the code distance of $12$, consistent with the near-threshold waterfall behavior reported in \cite{kasai2026breakingorthogonalitybarrierquantum,zhao2026ultrahighratequantumerrorcorrection}. Adding the third MILP stage removes most of the non-convergence failures and lowers $[[132,30,12]]$ to $1.0_{-0.53}^{+0.90}\times 10^{-9}$, while $[[672,336,12]]$ improves only to $4.3_{-0.80}^{+1.1}\times 10^{-9}$ because its failures are dominated by the unheralded stratum.}

\section{Case Study}\label{sec:case-study}

The propositions of Section~\ref{sec:logic} are stated for a general GALA code, where the group data is the only input.  This section works two instances out end to end, so that every object in that section (the hypercube organization, the dirty cyclic shifts, the homomorphic CNOTs, and the fold-transversal Cliffords) are given explicitly.  We focus on two small codes $[[132, 30,12]]$ and $[[672,336,12]]$, so that logical action can be illustrated compactly.  The $[[132, 30,12]]$ instance is moreover \emph{strictly} self-dual, with the identity permutation as its ZX-duality, so its fold-transversal Clifford gates require no atom rearrangement at all, whereas a generic ZX-dual GALA code must first implement the fold of Prop.~\ref{def:dual-gala}.

\subsection{Self-Dual $[[132, 30,12]]$ Code}
Let use first consider the $[[132, 30,12]]$ code $\mathcal{Q}\in\mathrm{GALA}_{12,5}(e\times C_{11})$ with $L=12$, $J=5$, a trivial top, and cyclic bottom $C_m=C_{11}$, so that $n=12\times1\times11=132$.  Every lift is a single group element, so the monomial distance bound of Prop.~\ref{prop:distance-bound} apply: $J=5>L/4=3$, and the $J>L/4$ branch caps the distance at $d\leq L$.  An exhaustive search certifies $d=L=12$, attaining the cap.  

\begin{table}[h]
\centering
\begin{tabular}{lccc}
\toprule
basis & $\min$ & median & $\max$ \\
\midrule
$\mathcal{L}_X^{\min}$ & $12$ & $12$ & $12$ \\
$\mathcal{L}_Z^{\min}$ & $12$ & $12$ & $12$ \\
\midrule
 $\mathcal{L}^t$, $\mathcal{L}^b$ (Structured, see \ref{fig:case-structure}) & $12$ & $12$ & $22$ \\
${\mathcal{L}^t}', {\mathcal{L}^b}'$ (Greedy orthonormal, see \ref{fig:basisorthgo}) & $12$ & $22$ & $24$ \\

 $\mathcal{L}(a)$, $\mathcal{L}(b), \mathcal{L}_{res}$ ($D_{11}$-covariant, see \ref{fig:case-covar-basis}) & $14$ & $14$ & $22$ \\
\bottomrule
\end{tabular}
\caption{Weight spectra of the logical bases of the $[[132,30,12]]$ code. $\mathcal{L}_X^{\min}$ and $\mathcal{L}_Z^{\min}$ form a complete, minimum-weight logical basis, but is not orthonormal. The structured chains $\mathcal{L}^t$, $\mathcal{L}^b$ and ${\mathcal{L}^t}'$, ${\mathcal{L}^b}'$ is obtained via top and bottom codes. It has logical shifts and fold-transversal Cliffords (See Figures \ref{fig:case-structure}. The $\mathcal{L}(a)$, $\mathcal{L}(b), \mathcal{L}_{res}$ is obtained via exhustive search for symplectic seeds $a,b$. They are palindromic and translation covariant-covariant, see \ref{fig:case-covar-basis})\ref{fig:gatesortho}).}
\label{tab:case-bases}
\end{table}

\paragraph{Greedy Weight Reduction} We can find a min-weight logical basis and its partner logicals via Gaussian elimination followed by a greedy weight reduction, as summarized in Table~\ref{tab:case-bases}.  Both bases attain the code distance, $\wt{\bar X_i}=\wt{\bar Z_i}=12$ for every $i$, but $\mathcal{L}_X^{\min}$ and $\mathcal{L}_Z^{\min}$ are \emph{not} orthonormal partners: their pairing matrix $M_{ij}=\braket{\bar Z_i,\bar X_j}$ has full rank $30$ but is not the identity.  That failure is forced.  Since $H_X=H_Z$, the $X$- and $Z$-logical classes are represented inside the same subspace $\ns H\subset\mathbb{F}_2^{132}$, and every logical has even weight, so $\braket{v,v}=0$ for every representative: a $Z$-logical sharing its partner's support would commute with it, and no minimum-weight basis can be its own symplectic partner.  Orthonormality can be restored by the change of basis $\bar Z'=M^{-1}\bar Z$, which gives $\braket{\bar X_i,\bar Z'_j}=\delta_{ij}$, but only at the cost of weight: the conjugate basis $\mathcal{L}_Z'$ has min/median/max logical weight $18/22/24$. This greedily generated basis is unstructured: the shift $\sigma$ acts on it as an arbitrary Clifford circuit.  Algorithm~\ref{alg:hypercube-basis} is meant to fix that, but run on the orthonormal pair $(\mathcal{L}_X^{\min},\mathcal{L}_Z')$ that its input hypothesis requires, it emits a massively overcomplete family of $121=11\times11$ logicals spanning a $30$-dimensional space; after the second sector every further orbit of size $11$ adds exactly one new dimension.

\subsubsection{Low weight logical Basis Exposing Logical Cyclic Shifts}\label{ssec:case-chains}
A easier way to obtain a logical basis is by looking at its top and bottom logicals, which is made especially simple since the top is trivial. First, the top code is simply the all-one matrix, making it dual to the repetition code with exactly weight-2 logicals 
\begin{equation}\label{eq:case-top}
      \mathcal{L}^{t} =(\chi_{i}+\chi_{i+1})\otimes \textbf{1}^{1\times m}\quad, i\in\mathbb{Z}_{12},
\end{equation}
each of weight $2m=22$, where $e_i$ here denote the indicator vector $[\chi_i]_j = \delta _{i,j} $.  Consecutive pairs share a block so they anti-commute, but $ \mathcal{L}^{t}_i$ and $ \mathcal{L}^{t}_{i+2}$ are disjoint split into two sectors of six pairwise disjoint operators $\mathcal{L}^{t}_{e}$ and $ \mathcal{L}^{t}_{o}$, each of rank $5$ since they sum to $\mathbf 1\in\rs H$. Together the two sectors span all shift-invariant logicals. The bottom logical is exactly the latent row, which is also pairwise disjoint and has rank $10$ with the same argument. The remaining bottom logicals come from rank deficiencies of the active rows, which now correspond to the logical operators of the parent matrix, making them logials of the self-dual 2BGA code exactly with abelian quasi-cyclic lifts. 

Let us organize these logicals into five sector disjoint chains, as illustrated in Figure~\ref{fig:case-structure}. Together give $3\times11+2\times6=45$ generators of rank $30$, each with weight $12$ or $22$, and no two of them equal modulo $\rs H$. We can pick any two of the three bottom chains to form a complete basis.  We note that within a chain the generators have disjoint supports and therefore commute, so the chains are not themselves symplectic bases.

\begin{figure}[t]
\centering
\includegraphics[width=0.95\linewidth]{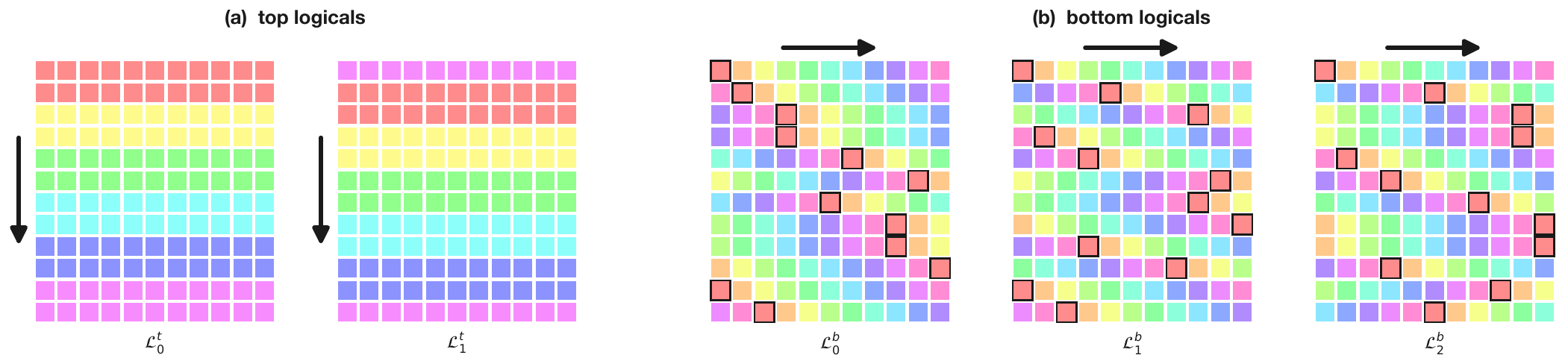}
\caption{The two logical families of the self-dual $[[132,30,12]]$ GALA code, drawn on the
$12\times11$ (block $\times$ $C_{11}$) grid of its $132$ qubits.  Within one chain the operators
are pairwise disjoint, one color each, and cover every qubit exactly once.
(a) The top logicals.  Each $\mathcal{L}^T_i$ is a pair of complete blocks of weight
$22$; taking every even/odd $i$ makes six of them disjoint, so the twelve modes fall into the two sectors $\mathcal{L}^{t}_{0}$ and $\mathcal{L}^{t}_{1}$, offset by one block. 
(b) The bottom logicals.  The outlined qubits correspond to the representative logical, with orbits under cyclic shift highlighted in different colors.}
\label{fig:case-structure}
\end{figure}

\paragraph{Product structure of the chains}\label{ssec:case-products}

Each of the five chains partitions the $132$ qubits, so summing a whole chain gives the
all-ones vector, which is a stabilizer. One consequence is that long products \emph{inside} a chain never need a high-weight representative, since
$\sum_{j\neq j_0}\bar X^{(a)}_j=\mathbf 1+\bar X^{(a)}_{j_0}\equiv\bar X^{(a)}_{j_0}$ modulo
$\rs H$: the product of all but one member of a chain is again weight-$12$.  The redundancy of the third bottom chain extends this \emph{across} chains --- eliminating $\mathcal{L}^{b}_{2}$ against the other two gives
\begin{equation}\label{eq:case-cross}
    \bar X^{(2)}\ =\ (1+x^{8})\,\bar X^{(0)}
    \ +\ (x^{2}+x^{3}+x^{4}+x^{5}+x^{6})\,\bar X^{(1)} ,
\end{equation}
that is, $\bar X^{(2)}_j=\bar X^{(0)}_j+\bar X^{(0)}_{j+8}+\sum_{t=2}^{6}\bar X^{(1)}_{j+t}$ for every $j$, so this particular seven-fold product across $\mathcal{L}^{b}_{0}$ and $\mathcal{L}^{b}_{1}$ also has a weight-$12$ representative.  
With the chains in hand, the dirty cyclic shift can be implemented with a single rigid AOD translation from Section~\ref{sec:logic}; this shifts the bottom chains by one but  fixes the top chain.  The gate is ``dirty'' in the sense that the $11$ sites of a chain carry only $10$
logical qubits, so the $11$-cycle permutes an overcomplete label set rather than a basis.

Since the ZX-duality of this code is $\tau=\mathrm{id}$, the two fold-transversal operators become to depth-$1$ layers of single-qubit gates:
\begin{equation}\label{eq:case-fold}
    H_\tau=H^{\otimes 132},\qquad S_\tau=(S^{\dagger})^{\otimes 132},
\end{equation}
and $\tau$ is fault tolerant by definition since one qubit gate cannot spread errors. This is special to $\tau=\mathrm{id}$.  For a generic ZX-dual GALA code the rearrangement of Prop.~\ref{def:dual-gala} must actually be performed, and while the resulting $H_\tau$ still only permutes and rotates single qubits, the phase-type gate $S_\tau$ carries a $CZ$ across every $\tau$-pair, which may drops fault distance. 

The logical action of the fold-transversal Clifford gates on structured basis given in Section~\ref{ssec:case-chains} appears as a generic logical Clifford circuit, in part owing to the fact that the basis is not orthgonormal. In particular, the logical actions are read off from the anti-commutation matrix of the $X$-logical representatives under,
\begin{equation}\label{eq:case-alpha}
    \alpha_{ij}=\braket{\bar X_i,\bar X_j}
    =|\supp{\bar X_i}\cap\supp{\bar X_j}|\ \mathrm{mod}\ 2 .
\end{equation}
Since $\mathbf 1\in\rs H$, every logical has even weight, so $\alpha$ has zero diagonal; and $\alpha$ is nondegenerate, because a $v\in\ns H$ pairing trivially with all of $\ns H$ lies in $(\ns H)^{\perp}=\rs H$ and hence represents the trivial class.  Work in an orthonormal frame, $\braket{\bar X_i,\bar Z_j}=\delta_{ij}$, as supplied by $(\mathcal{L}_X^{\min},\mathcal{L}_Z')$ of Table~\ref{tab:case-bases}, and let $\alpha$ and $\beta$ be the Gram matrices of the $X$- and the $Z$-representatives.  Transversal Hadamard sends $X(v)\mapsto Z(v)$ on the same support, and naming the resulting class by pairing it against the frame gives
\begin{equation}\label{eq:case-H}
    \bar X_i\ \longmapsto\ \prod_j\bar Z_j^{\,\alpha_{ij}},\qquad
    \bar Z_i\ \longmapsto\ \prod_j\bar X_j^{\,\beta_{ij}},\qquad
    \beta_{ij}=\braket{\bar Z_i,\bar Z_j}.
\end{equation}
Dual bases of a nondegenerate symmetric form have mutually inverse Gram matrices, so
$\alpha\beta=I$: transversal $H$ is a logical Hadamard composed with the
$\mathrm{GL}(30,\mathbb{F}_2)$ basis change $\alpha$.  Transversal $S^\dagger$ fixes every
$\bar Z$ and sends
\begin{equation}\label{eq:case-S}
    \bar X(v)\ \longmapsto\ (-1)^{q(v)}\,\bar X(v)\,\bar Z(\alpha v),\qquad
    q(v)=\tfrac12\wt{v}\ \mathrm{mod}\ 2 ,
\end{equation}
the sign being the residual phase $(-i)^{\wt v}=(-1)^{\wt v/2}$ of the calculation above.  So
$S_\tau$ is the diagonal logical Clifford carrying a $\overline{CZ}$ on every $\alpha$-edge and a
single-qubit logical phase wherever $q=1$.  On the weight-$12$ basis of
Table~\ref{tab:case-bases} every $q_i=\tfrac{12}2\bmod 2=0$ and $S_\tau$ is a pure
$\overline{CZ}$ network of $198$ edges; on the chains, the weight-$22$ top modes have
$q=\tfrac{22}2\bmod2=1$ and do carry a phase gate.
\begin{figure}[t]
\centering
\includegraphics[width=0.95\linewidth]{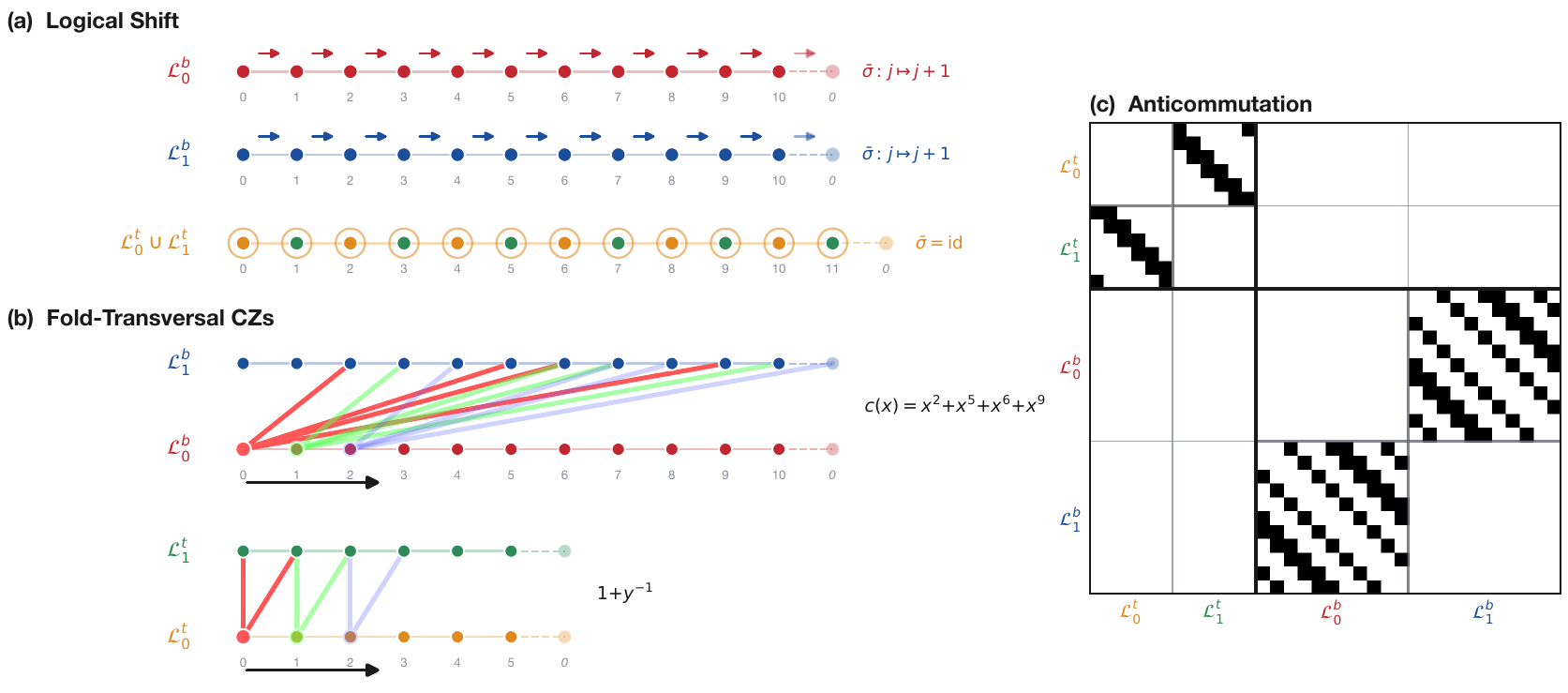}
\caption{Logical action of the transversal gates on the structure logical basis of the $[[132,30,12]]$ code.
(a) The $\mathbb Z_{11}$ shift: one rigid AOD translation induces
$\bar\sigma:j\mapsto j+1$ simultaneously on the bottom chains while leaving the top chains unchanged.
(b) The logical $\overline{CZ}$ circuit induced by $S_\tau=(S^{\dagger})^{\otimes132}$.  A bundle of four $CZ$s corresponding to the offsets $c(x)=x^2+x^5+x^6+x^9$ repeated at every site. 
On the top chain it is nearest-neighbor coupling corresponding to $12$-cycle
$\Lambda_{12}$.
(c) The anticommutation matrix $\alpha_{ij}=\braket{\bar X_i,\bar X_j}$ over the
$34=6+6+11+11$ generators of the four chains $\mathcal{L}^{t}_{0}$, $\mathcal{L}^{t}_{1}$,
$\mathcal{L}^{b}_{0}$, $\mathcal{L}^{b}_{1}$, that is, over the $45$ generators of
Section~\ref{ssec:case-chains} with the redundant third bottom chain $\mathcal{L}^{b}_{2}$
dropped; the remaining $34$ still have rank $30$ modulo $\rs H$.}
\label{fig:case-gates}
\end{figure}

The chains make the structure of $\alpha$ explicit (Figure~\ref{fig:case-gates}(b,c)), and in
the labeling of Section~\ref{ssec:case-chains} the two families come out with the same
shape.  Since $\sigma$ is a permutation it preserves the pairing, and the chain generators are
$\sigma$-orbits, so $\alpha_{(a,j),(b,j')}$ depends only on $j'-j$: every block of $\alpha$ is
circulant.  In the ordering
$(\mathcal{L}^{t}_{0},\mathcal{L}^{t}_{1},\mathcal{L}^{b}_{0},\mathcal{L}^{b}_{1})$
\begin{equation}\label{eq:case-alpha-blocks}
    \alpha\ =\
    \begin{pmatrix}
        0 & 1+y^{-1}\\ 1+y & 0
    \end{pmatrix}
    \ \oplus\
    \begin{pmatrix}
        0 & c(x)\\ c(x) & 0
    \end{pmatrix},
    \qquad
    c(x)=x^{2}+x^{5}+x^{6}+x^{9},
\end{equation}
where $y$ shifts the block index inside a top sector and $c$ is symmetric.  The top block is the adjacency matrix $\Lambda_{12}$ of the
$12$-site cycle, written in the two-sector ordering: $\bar D_b$ and $\bar D_{b'}$ share an
odd number of qubits only when $b'=b\pm1$, and the cycle alternates between
$\mathcal{L}^{t}_{0}$ and $\mathcal{L}^{t}_{1}$, so it is bipartite with the two offsets
$\{0,-1\}$.

\subsubsection{Low-Weight Logicals From  Greedy Search}\label{ssec:case-fold}

Choosing an orthonormal basis instead of the chains turns the transversal Clifford gate into a perfect
matching: $\alpha$ decomposes into $15$ disjoint pairs, which is the $30$ logical
qubits, split as $10$ pairs across the two bottom chains and $5$ pairs across the two top sectors.  In that basis $S_\tau$ is $15$ disjoint $\overline{CZ}$ gates together with phases on $13$ of the $30$ logical qubits, and $H_\tau$ is $15$ copies of $(\bar H\otimes\bar H)\cdot\overline{\mathrm{SWAP}}$.  That basis is the legible one, but its representatives run from weight $12$ up to $24$; the chains keep every representative at $12$ or $22$, so they are the ones to compile against. See Figure~\ref{fig:gatesortho}.

\begin{figure}
    \centering
    \includegraphics[width=0.95\linewidth]{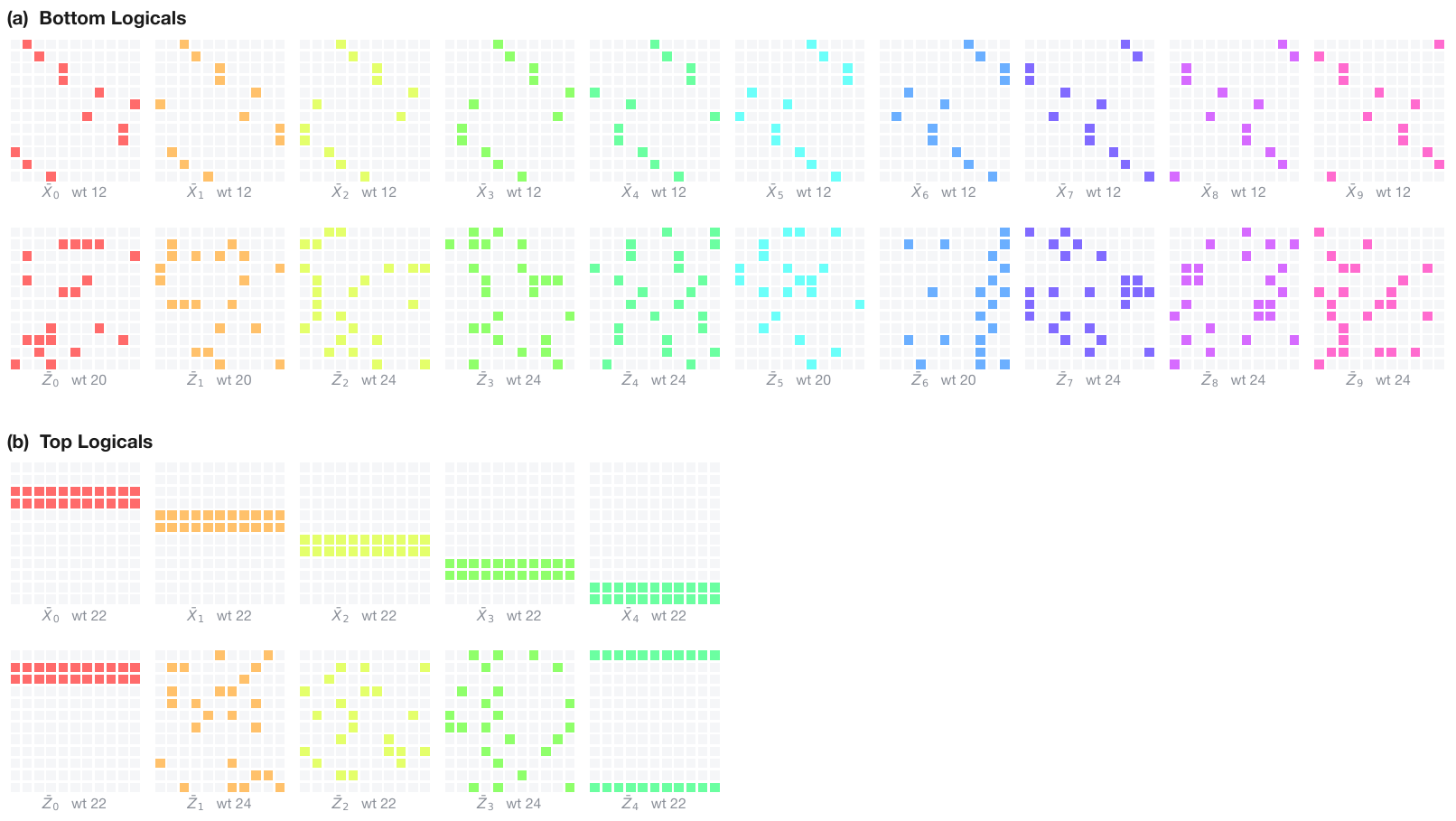}
    \caption{The logical basis with higher weight but simple transversal Cliffords. (a) the logical basis for the bottom $X$ logical and partner $Z$ logicals. (b) the logical basis for the top $X$ logical and partner $Z$ logicals. It is clear while a cyclic shift realizes a logical cyclic shift for the bottom $X$ logicals, it also implements an arbitrary circuit on the $Z$ logicals.}
    \label{fig:basisorthgo}
\end{figure}
\begin{figure}
    \centering
    \includegraphics[width=0.95\linewidth]{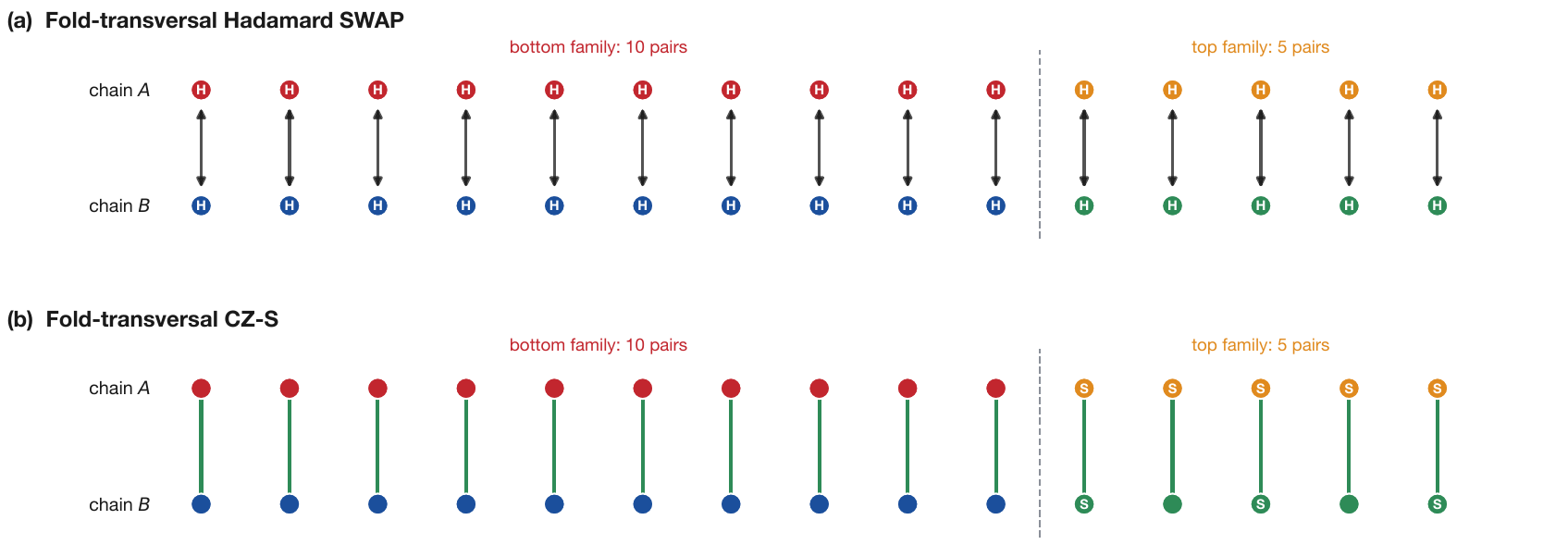}
    \caption{Logical action of the transversal Cliffords under the orthonormal basis. (a) $H^{\otimes 132}$ implements a logical Hadamard on all qubits, plus a swap between the two bottom chains and two top chains. (b) $S^{\otimes 132}$ implements a logical $CZ$ between two sectors of top and bottom logicals, plus a logical $S$ gate on 8 top logicals.}
    \label{fig:gatesortho}
\end{figure}

\subsubsection{Covariant logical basis from exhaustive search}\label{ssec:case-covar}
The earlier two logical operators does not consider the ZX dual symmetries. In this final section, we employ the aforementioned exhaustive search to find palindromic code words up to weight $20$. The results are summarized in Figure~\ref{fig:case-covar-basis} for $[[132,30,12]]$ code, and Figure~\ref{fig:case-free}

\begin{figure}
    \centering
    \includegraphics[width=0.75\linewidth]{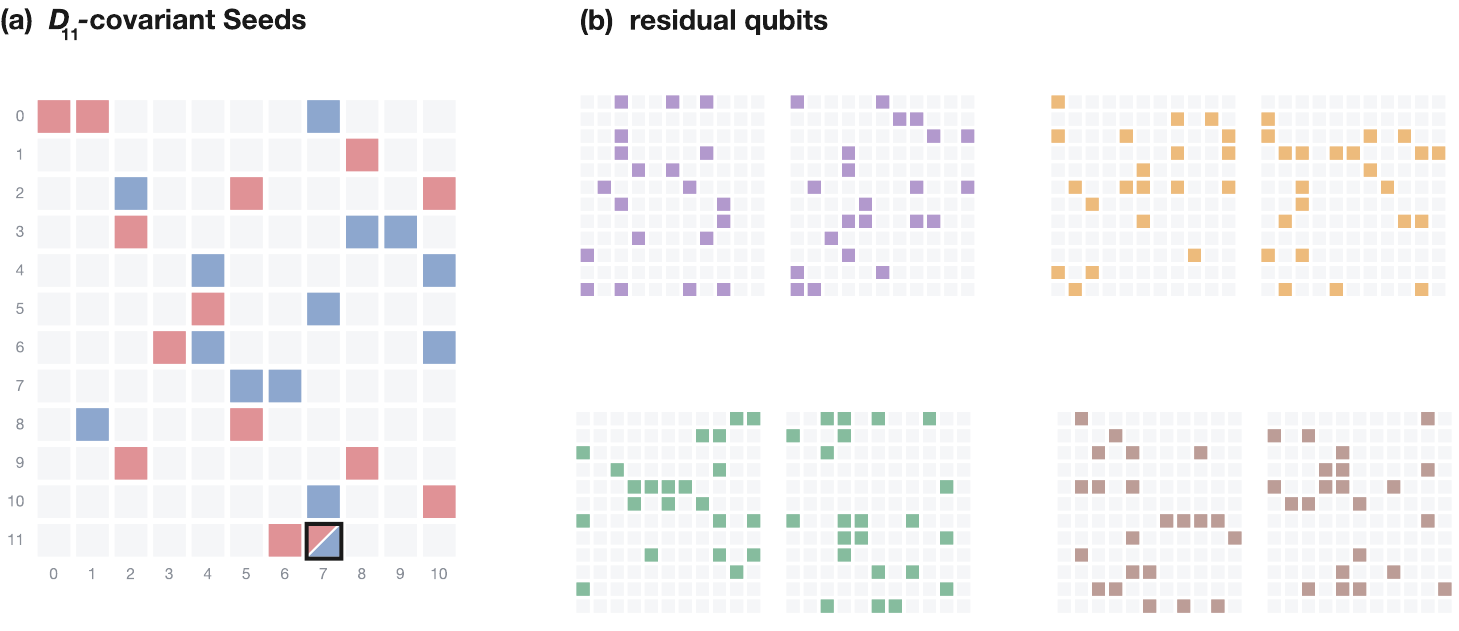}
    \includegraphics[width=0.95\linewidth]{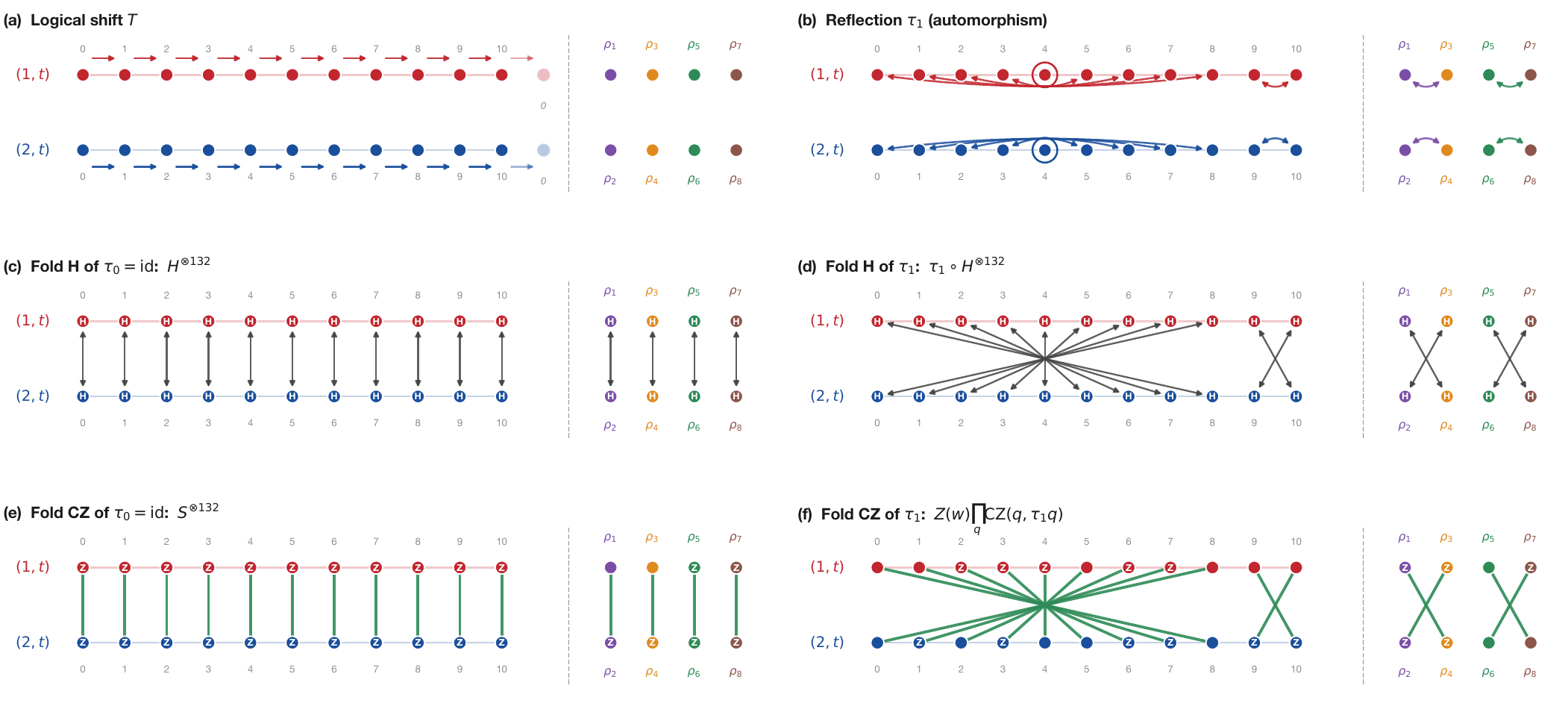}
    \caption{The  $D_{11}$-covariant basis for $[[132,30,12]]$ code. \emph{Top}:
    (a). Logical seeds for the weight-14 $D_{11}$-covariant basis, with same layout as before. The logical qubits can be constructed as $X_{a,i} = z^ia$, $X_{b,i} = z^ib$, $z^i$ is a cylic shift of all $L$ blocks by $i$. (b). There are $8$ $D_{11}$ invariant logical basis.\emph{Bottom}: Logical action of the fold-transversal gates from (a-c) the identity self-duality and (d-f) the meta duality. Note the CZ gate in figure (f). has circuit distance $d_{\bar{S}} = 6$.}
    \label{fig:case-covar-basis}
\end{figure}

\begin{figure}
    \centering
    \includegraphics[width=0.95\linewidth]{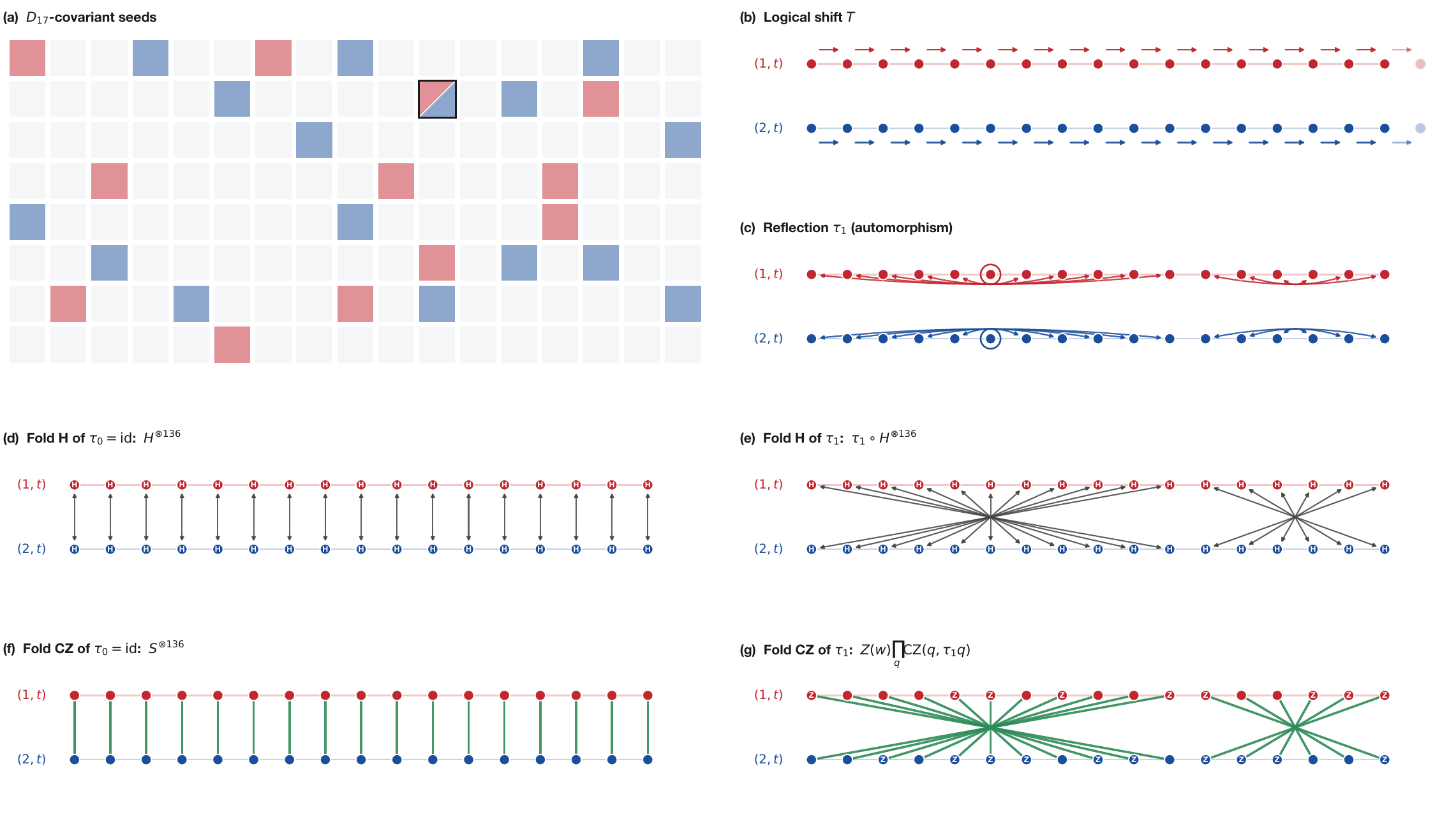}
    \caption{The $D_{17}$-covariant basis for $[[136,34,12]]$ code. Note there is no residual logical qubits.}
    \label{fig:case-free}
\end{figure}
\clearpage
\begin{table}[t]
\centering
\scriptsize\linespread{1.35}\selectfont
\begin{tabular}{lcccc}
\toprule
parameters & rate & $L$ & $H_k\star C_m$ & generators \\
\midrule
$[[312, 156, 8]]$ & 0.500 & 8 & $S_{3} \times C_{13}$ & \parbox[t]{0.46\textwidth}{\raggedright $\mathcal{F}$: $x^{3}$, $\sigma_0\,x^{6}{+}\sigma_0\,x^{9}$, $x^{1}$, $x^{5}{+}x^{9}$ \\ $\mathcal{G}$: $x^{2}$, $\tau_1\,x^{4}{+}\sigma_0\,x^{10}$, $x^{1}$, $x^{2}{+}x^{10}$} \\
$[[408, 204, 8]]$ & 0.500 & 8 & $S_{3} \ltimes_{\mathcal{R}} C_{17}$ & \parbox[t]{0.46\textwidth}{\raggedright $\mathcal{F}$: $\sigma_0\,x^{[14,7,3]}$, $\sigma_0\,x^{[0,10,6]}{+}\sigma_0\,x^{[16,9,5]}$, $\sigma_0\,x^{[8,1,14]}$, $x^{8}{+}x^{6}$ \\ $\mathcal{G}$: $\sigma_0\,x^{[13,6,2]}$, $\sigma_1\,x^{[10,16,6]}{+}\sigma_1\,x^{[5,11,1]}$, $\sigma_1\,x^{[6,12,2]}$, $x^{6}{+}x^{12}$} \\
$[[304, 152, 8]]$ & 0.500 & 8 & $S_{2} \ltimes C_{19}^{2}$ & \parbox[t]{0.46\textwidth}{\raggedright $\mathcal{F}$: $x^{3}$, $x^{16}{+}x^{1}$, $\tau\,x^{[14,16]}$, $\tau\,x^{[2,4]}{+}\tau\,x^{[4,6]}$ \\ $\mathcal{G}$: $\tau\,x^{[7,9]}$, $x^{2}{+}x^{10}$, $x^{3}$, $\tau\,x^{[5,7]}{+}\tau\,x^{[8,10]}$} \\
$[[672, 340, 8]]$ & 0.506 & 12 & $S_{4} \ltimes C_{14}^{4}$ & \parbox[t]{0.46\textwidth}{\raggedright $\mathcal{F}$: $ba\,x^{[10,4,4,12]}$, $x^{[12,12,0,12]}$, $x^{[1,3,1,11]}$, $x^{[9,9,5,9]}$, $x^{[1,9,1,13]}$, $cba\,x^{[4,10,9,1]}$ \\ $\mathcal{G}$: $x^{13}$, $x^{[12,9,9,9]}$, $c^{2}ba\,x^{[9,11,5,5]}$, $aba\,x^{[3,0,6,13]}$, $x^{0}$, $x^{[7,13,13,13]}$} \\
$[[480, 240, 10]]$ & 0.500 & 8 & $S_{3} \times C_{2}\!\times\!C_{2}\!\times\!C_{5}$ & \parbox[t]{0.46\textwidth}{\raggedright $\mathcal{F}$: $x^{(0,0,1)}$, $\sigma_0\,x^{(1,1,4)}{+}\tau_1\,x^{(0,1,3)}$, $x^{(1,1,3)}$, $x^{(1,0,4)}{+}x^{(1,0,3)}$ \\ $\mathcal{G}$: $x^{(0,0,0)}$, $\sigma_0\,x^{(1,0,1)}{+}\sigma_0\,x^{(0,1,0)}$, $x^{(1,1,1)}$, $x^{(1,1,4)}{+}x^{(0,1,2)}$} \\
$[[1008, 508, 10]]$ & 0.504 & 12 & $S_{3} \times C_{4}\!\times\!C_{7}$ & \parbox[t]{0.46\textwidth}{\raggedright $\mathcal{F}$: $x^{(1,5)}$, $\sigma_0\,x^{(1,6)}$, $x^{(1,3)}$, $x^{(0,1)}$, $\tau_0\,x^{(2,0)}$, $x^{(1,4)}$ \\ $\mathcal{G}$: $x^{(3,2)}$, $x^{(1,0)}$, $\tau_0\,x^{(1,0)}$, $x^{(1,2)}$, $x^{(1,6)}$, $\sigma_0\,x^{(3,6)}$} \\
$[[560, 280, 10]]$ & 0.500 & 8 & $S_{2} \ltimes_{\mathcal{R}} C_{5}\!\times\!C_{7}$ & \parbox[t]{0.46\textwidth}{\raggedright $\mathcal{F}$: $x^{(2,3)}$, $x^{(0,3)}{+}x^{(4,0)}$, $\tau\,x^{[(0,6),(2,0)]}$, $\tau\,x^{[(2,5),(1,1)]}{+}\tau\,x^{[(1,3),(0,6)]}$ \\ $\mathcal{G}$: $\tau\,x^{[(2,2),(1,5)]}$, $x^{(2,5)}{+}x^{(3,5)}$, $x^{(1,3)}$, $\tau\,x^{[(4,1),(1,2)]}{+}\tau\,x^{[(0,0),(2,1)]}$} \\
$[[936, 472, 10]]$ & 0.504 & 12 & $S_{3} \ltimes_{\mathcal{R}} C_{2}\!\times\!C_{13}$ & \parbox[t]{0.46\textwidth}{\raggedright $\mathcal{F}$: $x^{(1,8)}$, $\sigma_0\,x^{[(1,9),(1,0),(0,9)]}$, $x^{(0,9)}$, $\sigma_1\,x^{[(1,12),(0,5),(1,1)]}$, $x^{(0,10)}$, $x^{(0,5)}$ \\ $\mathcal{G}$: $\sigma_1\,x^{[(1,3),(0,3),(0,12)]}$, $x^{(0,10)}$, $\sigma_1\,x^{[(0,5),(1,11),(0,7)]}$, $x^{(1,4)}$, $x^{(0,8)}$, $x^{(0,4)}$} \\
$[[528, 264, 10]]$ & 0.500 & 8 & $S_{2} \ltimes (C_{3}\!\times\!C_{11})^{2}$ & \parbox[t]{0.46\textwidth}{\raggedright $\mathcal{F}$: $\tau\,x^{[(1,8),(0,10)]}$, $\tau\,x^{[(2,8),(1,10)]}{+}\tau\,x^{[(2,1),(1,3)]}$, $\tau\,x^{[(1,7),(0,9)]}$, $x^{(1,9)}{+}x^{(1,6)}$ \\ $\mathcal{G}$: $\tau\,x^{[(1,7),(0,9)]}$, $\tau\,x^{[(2,2),(1,4)]}{+}\tau\,x^{[(0,1),(2,3)]}$, $\tau\,x^{[(2,1),(1,3)]}$, $x^{(2,9)}{+}x^{(1,4)}$} \\
$[[768, 388, 10]]$ & 0.505 & 12 & $S_{4} \ltimes C_{16}^{4}$ & \parbox[t]{0.46\textwidth}{\raggedright $\mathcal{F}$: $acb\,x^{[8,8,13,4]}$, $x^{[4,4,13,4]}$, $x^{[13,5,13,7]}$, $ba\,x^{[5,11,8,2]}$, $x^{[1,4,1,8]}$, $x^{[11,11,3,11]}$ \\ $\mathcal{G}$: $x^{[12,2,2,2]}$, $x^{14}$, $c^{2}ba\,x^{[15,2,1,10]}$, $x^{12}$, $x^{[12,13,13,13]}$, $aba\,x^{[12,7,15,1]}$} \\
$[[672, 336, 12]]$ & 0.500 & 8 & $C_{2 }\!\times\!C_{ 3 }\!\times\!C_{ 14}$ & \parbox[t]{0.46\textwidth}{\raggedright $\mathcal{F}$: $x^{(1,1,5)}$, $x^{(0,2,7)}{+}x^{(0,0,10)}$, $x^{(1,2,5)}$, $x^{(1,0,6)}{+}x^{(0,0,8)}$ \\ $\mathcal{G}$: $x^{(0,1,13)}$, $x^{(1,0,10)}{+}x^{(0,2,10)}$, $x^{(0,0,6)}$, $x^{(1,0,8)}{+}x^{(0,1,12)}$} \\
$[[840, 420, 12]]$ & 0.500 & 8 & $S_{3} \ltimes_{\mathcal{R}} C_{5}\!\times\!C_{7}$ & \parbox[t]{0.46\textwidth}{\raggedright $\mathcal{F}$: $x^{(1,2)}$, $\sigma_0\,x^{[(4,3),(1,6),(1,6)]}{+}\sigma_0\,x^{[(2,2),(4,5),(4,5)]}$, $\sigma_1\,x^{[(1,5),(3,1),(1,5)]}$, $\sigma_0\,x^{[(3,5),(0,1),(0,1)]}{+}\sigma_0\,x^{[(3,4),(0,0),(0,0)]}$ \\ $\mathcal{G}$: $\sigma_0\,x^{[(0,0),(2,3),(2,3)]}$, $\sigma_0\,x^{[(0,5),(2,1),(2,1)]}{+}\sigma_0\,x^{[(2,1),(4,4),(4,4)]}$, $x^{(3,0)}$, $\sigma_0\,x^{[(3,1),(0,4),(0,4)]}{+}\sigma_0\,x^{[(4,5),(1,1),(1,1)]}$} \\
$[[1080, 540, 12]]$ & 0.500 & 8 & $S_{3} \ltimes (C_{5}\!\times\!C_{9})^{3}$ & \parbox[t]{0.46\textwidth}{\raggedright $\mathcal{F}$: $x^{(1,6)}$, $x^{(2,3)}{+}x^{(0,1)}$, $\tau_0\,x^{[(0,6),(0,3),(1,0)]}$, $\sigma_0\,x^{[(3,0),(1,2),(1,0)]}{+}\sigma_1\,x^{[(3,6),(1,6),(3,4)]}$ \\ $\mathcal{G}$: $\sigma_0\,x^{[(1,1),(4,3),(4,1)]}$, $x^{(1,6)}{+}x^{(4,1)}$, $x^{(0,5)}$, $x^{[(2,6),(3,4),(3,4)]}{+}\tau_0\,x^{[(2,4),(4,0),(0,6)]}$} \\
$[[1752, 880, 14]]$ & 0.502 & 12 & $S_{2} \ltimes_{\mathcal{R}} C_{73}^{2}$ & \parbox[t]{0.46\textwidth}{\raggedright $\mathcal{F}$: $x^{17}$, $x^{44}$, $x^{14}$, $x^{6}$, $\tau\,x^{[8,14]}$, $\tau\,x^{[42,40]}$ \\ $\mathcal{G}$: $x^{36}$, $x^{20}$, $x^{61}$, $x^{10}$, $\tau\,x^{[72,5]}$, $\tau\,x^{[22,20]}$} \\
$[[1764, 886, 14]]$ & 0.502 & 12 & $S_{3} \times C_{49}$ & \parbox[t]{0.46\textwidth}{\raggedright $\mathcal{F}$: $x^{18}$, $x^{6}$, $x^{22}$, $\sigma_0\,x^{21}$, $\tau_0\,x^{7}$, $x^{33}$ \\ $\mathcal{G}$: $\tau_0\,x^{17}$, $x^{48}$, $x^{20}$, $x^{32}$, $x^{47}$, $\sigma_0\,x^{11}$} \\
$[[1896, 952, 14]]$ & 0.502 & 12 & $S_{2} \ltimes_{\mathcal{R}} C_{79}$ & \parbox[t]{0.46\textwidth}{\raggedright $\mathcal{F}$: $x^{34}$, $\tau\,x^{[61,76]}$, $x^{34}$, $x^{47}$, $\tau\,x^{[13,62]}$, $x^{32}$ \\ $\mathcal{G}$: $x^{55}$, $x^{10}$, $\tau\,x^{[72,42]}$, $x^{40}$, $x^{35}$, $\tau\,x^{[37,52]}$} \\
$[[2232, 1120, 16]]$ & 0.502 & 12 & $S_{3} \times C_{62}$ & \parbox[t]{0.46\textwidth}{\raggedright $\mathcal{F}$: $\sigma_0\,x^{48}$, $\tau_1\,x^{28}$, $x^{61}$, $x^{30}$, $x^{3}$, $x^{39}$ \\ $\mathcal{G}$: $x^{38}$, $x^{51}$, $\sigma_0\,x^{8}$, $\tau_1\,x^{13}$, $x^{5}$, $x^{48}$} \\
$[[2328, 1168, 16]]$ & 0.502 & 12 & $S_{2} \ltimes_{\mathcal{R}} C_{97}^{2}$ & \parbox[t]{0.46\textwidth}{\raggedright $\mathcal{F}$: $\tau\,x^{[37,14]}$, $x^{59}$, $x^{85}$, $x^{62}$, $\tau\,x^{[54,3]}$, $x^{22}$ \\ $\mathcal{G}$: $x^{87}$, $x^{18}$, $x^{42}$, $\tau\,x^{[32,78]}$, $x^{85}$, $\tau\,x^{[59,36]}$} \\
\bottomrule
\end{tabular}
\caption{Certified frontier GALA codes with stabilizer weight $w=12$, girth $\geq6$, rate $\geq1/2$ and distance $\geq8$. The generator column lists the lifts $\mathcal F=(F_0,\dots,F_{L/2-1})$ and $\mathcal G$, each entry a group-ring element $\sum_p h\,x^{s}$ with $h\in H_k$ with an identity $h$ omitted, and $x^{s}\in C_m$ are (products of) cyclic shifts. Each element in a tuple corresponds to one of the factors of $C_m$, and several tuples may be needed for semidirect lifts. $S_3$ elements use the $\sigma/\tau$ labels of Section~\ref{sec:construction} and $S_4$ elements are words in $\langle a,b,c \mid a^2=b^3=c^4=abc=e\rangle$ realized by $a=(2\,3)$, $b=(0\,1\,2)$, $c=(0\,2\,3\,1)$.}
\label{tab:frontier-nondual}
\end{table}

\begin{table}[t]
\centering
\scriptsize\linespread{1.35}\selectfont
\begin{tabular}{lccccc}
\toprule
parameters & rate & $L$ & $H_k\star C_m$ & duality & generators \\
\midrule
$[[576, 292, 8]]$ & 0.507 & 12 & $S_{3} \times C_{2}\!\times\!C_{8}$ & $r_1^{}$ & \parbox[t]{0.46\textwidth}{\raggedright $\mathcal{F}$: $x^{(1,0)}$, $\tau_2\,x^{(0,1)}$, $\sigma_1\,x^{(0,5)}$, $x^{(1,7)}$, $x^{(1,5)}$, $\sigma_0\,x^{(1,2)}$ \\ $\mathcal{G}$: $x^{(1,1)}$, $x^{(1,3)}$, $\sigma_1\,x^{(1,6)}$, $x^{(1,0)}$, $\tau_2\,x^{(0,7)}$, $\sigma_0\,x^{(0,3)}$} \\
$[[576, 294, 8]]$ & 0.510 & 12 & $S_{3} \times C_{2}\!\times\!C_{8}$ & $r_0^{}$ & \parbox[t]{0.46\textwidth}{\raggedright $\mathcal{F}$: $x^{(0,4)}$, $\sigma_1\,x^{(0,7)}$, $\tau_2\,x^{(1,6)}$, $x^{(1,4)}$, $\sigma_0\,x^{(0,6)}$, $\sigma_1\,x^{(1,6)}$ \\ $\mathcal{G}$: $x^{(0,4)}$, $\sigma_0\,x^{(0,1)}$, $\tau_2\,x^{(1,2)}$, $x^{(1,4)}$, $\sigma_1\,x^{(0,2)}$, $\sigma_0\,x^{(1,2)}$} \\
$[[720, 364, 10]]$ & 0.506 & 12 & $S_{3} \times C_{4}\!\times\!C_{5}$ & $r_0^{}$ & \parbox[t]{0.46\textwidth}{\raggedright $\mathcal{F}$: $x^{(2,4)}$, $\sigma_1\,x^{(0,2)}$, $x^{(0,2)}$, $x^{(3,3)}$, $\sigma_0\,x^{(0,1)}$, $\tau_0\,x^{(2,2)}$ \\ $\mathcal{G}$: $x^{(2,1)}$, $\sigma_0\,x^{(0,3)}$, $x^{(0,3)}$, $x^{(1,2)}$, $\sigma_1\,x^{(0,4)}$, $\tau_0\,x^{(2,3)}$} \\
$[[864, 436, 10]]$ & 0.505 & 12 & $S_{3} \times C_{2}\!\times\!C_{3}\!\times\!C_{4}$ & $r_1^{}$ & \parbox[t]{0.46\textwidth}{\raggedright $\mathcal{F}$: $x^{(1,0,2)}$, $\sigma_1\,x^{(0,0,3)}$, $\tau_0\,x^{(0,0,1)}$, $x^{(1,1,3)}$, $\sigma_0\,x^{(0,2,0)}$, $x^{(0,0,1)}$ \\ $\mathcal{G}$: $x^{(1,2,1)}$, $\sigma_1\,x^{(0,1,0)}$, $x^{(0,0,3)}$, $x^{(1,0,2)}$, $\sigma_0\,x^{(0,0,1)}$, $\tau_0\,x^{(0,0,3)}$} \\
$[[1008, 510, 10]]$ & 0.506 & 12 & $S_{3} \times C_{2}\!\times\!C_{2}\!\times\!C_{7}$ & $r_0^{}$ & \parbox[t]{0.46\textwidth}{\raggedright $\mathcal{F}$: $x^{(0,0,4)}$, $x^{(0,0,2)}$, $\sigma_1\,x^{(1,0,5)}$, $x^{(1,0,3)}$, $\tau_0\,x^{(0,0,5)}$, $\sigma_0\,x^{(0,1,0)}$ \\ $\mathcal{G}$: $x^{(0,0,3)}$, $x^{(0,0,5)}$, $\sigma_0\,x^{(1,0,2)}$, $x^{(1,0,4)}$, $\tau_0\,x^{(0,0,2)}$, $\sigma_1\,x^{(0,1,0)}$} \\
$[[1056, 532, 12]]$ & 0.504 & 12 & $S_{4} \times C_{2}\!\times\!C_{11}$ & $r_0^{}$ & \parbox[t]{0.46\textwidth}{\raggedright $\mathcal{F}$: $x^{(0,10)}$, $aca\,x^{(1,2)}$, $bacb\,x^{(0,3)}$, $x^{(0,4)}$, $c^{2}\,x^{(0,2)}$, $ba\,x^{(0,3)}$ \\ $\mathcal{G}$: $x^{(0,1)}$, $ba\,x^{(1,9)}$, $bacb\,x^{(0,8)}$, $x^{(0,7)}$, $c^{2}\,x^{(0,9)}$, $aca\,x^{(0,8)}$} \\
\bottomrule
\end{tabular}
\caption{Every certified \emph{ZX-dual} GALA code with stabilizer weight $w=12$, girth $\geq6$, rate $\geq1/2$ and distance $\geq8$. The duality column gives the sector involution of Prop.~\ref{def:dual-gala} in the notation of Eq.~\eqref{eq:zx-sector-maps}. The generator column lists the lifts $\mathcal F=(F_0,\dots,F_{L/2-1})$ and $\mathcal G$, similarly to Table~\ref{tab:frontier-nondual}.}
\label{tab:frontier-selfdual-good}
\end{table}

\begin{table}[t]
\centering
\scriptsize\linespread{1.35}\selectfont
\begin{tabular}{llcccccc}
\toprule
parameters & rate & $L$ & $J$ & $H_k\star C_m$ & $t_4$ & duality & generators \\
\midrule
$[[136, 36, 8]]$ & 0.265 & 8 & 3 & $C_{17}$ & 816 & $r_2$ & \parbox[t]{0.46\textwidth}{\raggedright $\mathcal{F}$: $x^{7}$, $x^{0}{+}x^{5}$, $x^{6}$, $x^{1}{+}x^{2}$ \\ $\mathcal{G}$: $x^{10}$, $x^{16}{+}x^{15}$, $x^{11}$, $x^{0}{+}x^{12}$} \\
$[[256, 68, 8]]$ & 0.266 & 8 & 3 & $S_{2} \ltimes_{\mathcal{R}} C_{16}$ & 1,536 & $r_2$ & \parbox[t]{0.46\textwidth}{\raggedright $\mathcal{F}$: $x^{9}$, $\tau\,x^{[9,4]}{+}\tau\,x^{[3,14]}$, $x^{13}$, $x^{12}{+}\tau\,x^{[5,0]}$ \\ $\mathcal{G}$: $x^{7}$, $x^{4}{+}\tau\,x^{[0,11]}$, $x^{3}$, $\tau\,x^{[12,7]}{+}\tau\,x^{[2,13]}$} \\
$[[132, 30, 12]]$ & 0.227 & 12 & 5 & $C_{11}$ & 660 & $r_2$ & \parbox[t]{0.46\textwidth}{\raggedright $\mathcal{F}$: $x^{2}$, $x^{4}$, $x^{3}$, $x^{6}$, $x^{3}$, $x^{9}$ \\ $\mathcal{G}$: $x^{9}$, $x^{2}$, $x^{8}$, $x^{5}$, $x^{8}$, $x^{7}$} \\
$[[136, 34, 12]]$ & 0.250 & 8 & 3 & $C_{17}$ & 544 & $r_2$ & \parbox[t]{0.46\textwidth}{\raggedright $\mathcal{F}$: $x^{2}$, $x^{1}$, $x^{3}{+}x^{16}$, $x^{13}{+}x^{12}$ \\ $\mathcal{G}$: $x^{15}$, $x^{4}{+}x^{5}$, $x^{14}{+}x^{1}$, $x^{16}$} \\
$[[168, 42, 12]]$ & 0.250 & 8 & 3 & $C_{3}\!\times\!C_{7}$ & 1,008 & $r_2$ & \parbox[t]{0.46\textwidth}{\raggedright $\mathcal{F}$: $x^{(0,1)}$, $x^{(0,3)}{+}x^{(0,2)}$, $x^{(2,4)}{+}x^{(0,4)}$, $x^{(1,3)}$ \\ $\mathcal{G}$: $x^{(0,6)}$, $x^{(2,4)}$, $x^{(1,3)}{+}x^{(0,3)}$, $x^{(0,4)}{+}x^{(0,5)}$} \\
$[[192, 40, 12]]$ & 0.208 & 12 & 5 & $C_{16}$ & 1,216 & $r_2$ & \parbox[t]{0.46\textwidth}{\raggedright $\mathcal{F}$: $x^{10}$, $x^{3}$, $x^{2}$, $x^{15}$, $x^{12}$, $x^{6}$ \\ $\mathcal{G}$: $x^{6}$, $x^{10}$, $x^{4}$, $x^{1}$, $x^{14}$, $x^{13}$} \\
$[[228, 46, 12]]$ & 0.202 & 12 & 5 & $C_{19}$ & 1,444 & $r_2$ & \parbox[t]{0.46\textwidth}{\raggedright $\mathcal{F}$: $x^{10}$, $x^{15}$, $x^{6}$, $x^{4}$, $x^{5}$, $x^{17}$ \\ $\mathcal{G}$: $x^{9}$, $x^{2}$, $x^{14}$, $x^{15}$, $x^{13}$, $x^{4}$} \\
$[[576, 104, 12]]$ & 0.181 & 12 & 5 & $S_{3} \times C_{16}$ & 2,112 & $r_0$ & \parbox[t]{0.46\textwidth}{\raggedright $\mathcal{F}$: $\sigma_1\,x^{8}$, $x^{10}$, $x^{8}$, $\sigma_1\,x^{9}$, $\sigma_0\,x^{8}$, $\sigma_1\,x^{7}$ \\ $\mathcal{G}$: $\sigma_0\,x^{8}$, $x^{6}$, $x^{8}$, $\sigma_0\,x^{7}$, $\sigma_1\,x^{8}$, $\sigma_0\,x^{9}$} \\
$[[272, 72, 12]]$ & 0.265 & 8 & 3 & $S_{2} \ltimes_{\mathcal{R}} C_{17}$ & 1,088 & $r_2$ & \parbox[t]{0.46\textwidth}{\raggedright $\mathcal{F}$: $\tau\,x^{[16,11]}$, $x^{2}{+}\tau\,x^{[8,3]}$, $\tau\,x^{[8,3]}$, $x^{4}{+}\tau\,x^{[2,14]}$ \\ $\mathcal{G}$: $\tau\,x^{[6,1]}$, $x^{13}{+}\tau\,x^{[3,15]}$, $\tau\,x^{[14,9]}$, $x^{15}{+}\tau\,x^{[14,9]}$} \\
$[[576, 112, 12]]$ & 0.194 & 12 & 5 & $S_{3} \ltimes C_{16}^{3}$ & 2,880 & $r_2$ & \parbox[t]{0.46\textwidth}{\raggedright $\mathcal{F}$: $\sigma_0\,x^{[9,6,8]}$, $x^{8}$, $\sigma_0\,x^{[14,11,13]}$, $\sigma_0\,x^{[8,5,7]}$, $x^{5}$, $\sigma_0\,x^{[4,1,3]}$ \\ $\mathcal{G}$: $\sigma_1\,x^{[8,7,10]}$, $\sigma_1\,x^{[13,12,15]}$, $x^{11}$, $\sigma_1\,x^{[9,8,11]}$, $\sigma_1\,x^{[3,2,5]}$, $x^{8}$} \\
$[[1320, 232, 12]]$ & 0.176 & 12 & 5 & $S_{2} \ltimes (C_{5}\!\times\!C_{11})^{2}$ & 1,320 & $r_0$ & \parbox[t]{0.46\textwidth}{\raggedright $\mathcal{F}$: $\tau\,x^{[(2,5),(1,4)]}$, $\tau\,x^{[(2,1),(1,0)]}$, $\tau\,x^{[(4,3),(3,2)]}$, $x^{(0,7)}$, $x^{(1,2)}$, $x^{(3,8)}$ \\ $\mathcal{G}$: $\tau\,x^{[(4,7),(3,6)]}$, $\tau\,x^{[(4,0),(3,10)]}$, $\tau\,x^{[(2,9),(1,8)]}$, $x^{(0,4)}$, $x^{(4,9)}$, $x^{(2,3)}$} \\
$[[248, 62, 12]]$ & 0.250 & 8 & 3 & $C_{31}$ & 0 & --- & \parbox[t]{0.46\textwidth}{\raggedright $\mathcal{F}$: $x^{15}$, $x^{22}{+}x^{10}$, $x^{7}{+}x^{27}$, $x^{17}$ \\ $\mathcal{G}$: $x^{25}{+}x^{1}$, $x^{25}$, $x^{3}{+}x^{12}$, $x^{26}$} \\
$[[280, 70, 12]]$ & 0.250 & 8 & 3 & $C_{35}$ & 0 & --- & \parbox[t]{0.46\textwidth}{\raggedright $\mathcal{F}$: $x^{27}{+}x^{20}$, $x^{4}$, $x^{26}$, $x^{29}{+}x^{12}$ \\ $\mathcal{G}$: $x^{18}$, $x^{6}{+}x^{29}$, $x^{22}{+}x^{31}$, $x^{26}$} \\
$[[328, 82, 12]]$ & 0.250 & 8 & 3 & $C_{41}$ & 0 & --- & \parbox[t]{0.46\textwidth}{\raggedright $\mathcal{F}$: $x^{5}$, $x^{6}$, $x^{40}{+}x^{25}$, $x^{13}{+}x^{38}$ \\ $\mathcal{G}$: $x^{21}{+}x^{29}$, $x^{18}$, $x^{40}{+}x^{39}$, $x^{35}$} \\
$[[336, 84, 12]]$ & 0.250 & 8 & 3 & $C_{3}\!\times\!C_{14}$ & 0 & --- & \parbox[t]{0.46\textwidth}{\raggedright $\mathcal{F}$: $x^{(0,5)}$, $x^{(0,3)}$, $x^{(2,12)}{+}x^{(2,10)}{+}x^{(0,11)}$, $x^{(1,11)}$ \\ $\mathcal{G}$: $x^{(0,8)}{+}x^{(2,13)}$, $x^{(2,5)}{+}x^{(2,6)}$, $x^{(0,4)}$, $x^{(1,5)}$} \\
\bottomrule
\end{tabular}
\caption{Certified GALA codes with $J> \frac L4\implies r< \frac12,\ d\leq w$. $t_4$ counts the $4$-cycles of the $H_X$ Tanner graph, so $t_4=0$ is exactly girth $\geq6$. The duality column gives the sector involution of Prop.~\ref{def:dual-gala} in the notation of Eq.~\eqref{eq:zx-sector-maps}. The generator column lists the lifts $\mathcal F=(F_0,\dots,F_{L/2-1})$ and $\mathcal G$, similarly to Table~\ref{tab:frontier-nondual}.}
\label{tab:frontier-other}
\end{table}

\end{document}